\pdfoutput=1
\documentclass[aps,prl,onecolumn,10pt,nobibnotes]{revtex4-2}
\usepackage{graphicx} %
\usepackage{xcolor} %

\usepackage{enumitem}
\usepackage{tikz}
\usepackage{amsfonts}
\usepackage{amsmath,amssymb,amsthm}
\usepackage{graphicx}
\usepackage{fancyhdr}
\usepackage[breakable]{tcolorbox}
\usepackage{bbm}
\usepackage{url}
\usepackage{mathtools}
\usepackage{braket}
\usepackage{upgreek }
\usepackage{dsfont}
\usepackage{collectbox}
\usepackage{braket}
\usepackage{quantikz}
\usepackage{booktabs} %
\usepackage[normalem]{ulem}
\usepackage[ruled]{algorithm2e}
\usepackage{algpseudocode}
\usepackage{changepage}

\usepackage{hyperref}
\usepackage{cleveref} %
\newtheorem{thm}{Theorem}
\newtheorem*{thm*}{Theorem}
\makeatletter
\newcommand{\setthmtag}[1]{%
  \let\oldthethm\thethm%
  \newcommand{\thethm}{#1}%
  \g@addto@macro\endthm{%
    \addtocounter{thm}{-1}%
    \global\let\thethm\oldthethm}%
  }
\makeatother

\newtheorem{proposition}[thm]{Proposition}
\newtheorem*{prop*}{Proposition}
\newtheorem{lemma}[thm]{Lemma}
\newtheorem*{lemma*}{Lemma}
\newtheorem{cor}[thm]{Corollary}
\newtheorem*{cor*}{Corollary}

\newtheorem*{cj*}{Conjecture}
\newtheorem{definition}[thm]{Definition}
\newtheorem*{Def*}{Definition}

\theoremstyle{definition}
\newtheorem{remark}{Remark}
\newtheorem*{rem*}{Remark}

\newcommand{\tr}[1]{\operatorname{tr}\left[#1\right]}
\newcommand{\R}{\mathbb{R}}
\newcommand{\expt}[1]{\textrm{exp}}

\newcommand{\cA}{\mathcal{A}}
\newcommand{\cD}{\mathcal{D}}

\newcommand{\uGW}{\operatorname{uGW}}
\newcommand{\QUBO}{\operatorname{QUBO}}
\newcommand{\GW}{\operatorname{GW}}
\newcommand{\ketbra}[1]{|#1\rangle\langle#1|}
\newcommand{\poly}{\textrm{poly}}
\newcommand{\polylog}{\textrm{polylog}}
\newcommand{\Kron}{\textrm{Kron}}
\newcommand{\C}{\mathbb{C}}
\newcommand{\cO}{\mathcal{O}}
\Crefname{figure}{Fig.}{Figs.}
\Crefname{equation}{Eq.}{Eqs.}
\Crefname{definition}{Def.}{Defs.}
\Crefname{cor}{Cor.}{Cors.}
\Crefname{theorem}{Thm.}{Thms.}
\Crefname{appendix}{Appendix}{Appendices.}
\Crefname{proposition}{Prop.}{Props.}

\crefname{section}{section}{sections}
\Crefname{section}{Section}{Sections}

\tikzset{
    vertex/.style={circle, draw=black, fill=blue, inner sep=0pt, minimum size=6pt}
}

\begin{document}
\title{Exponential Speed-ups for Structured Goemans-Williamson relaxations via Quantum Gibbs States and Pauli Sparsity}
\begin{abstract}
Quadratic Unconstrained Binary Optimization (QUBO) problems are prevalent in various applications and are known to be NP-hard. The seminal work of Goemans and Williamson introduced a semidefinite programming (SDP) relaxation for such problems, solvable in polynomial time that upper bounds the optimal value. Their approach also enables randomized rounding techniques to obtain feasible solutions with provable performance guarantees.

In this work, we identify instances of QUBO problems where matrix multiplicative weight methods lead to quantum and quantum-inspired algorithms that approximate the Goemans-Williamson SDP exponentially faster than existing methods, achieving polylogarithmic time complexity relative to the problem dimension. This speedup is attainable under the assumption that the QUBO cost matrix is sparse when expressed as a linear combination of Pauli strings satisfying certain algebraic constraints, and leverages efficient quantum and classical simulation results for quantum Gibbs states.
    
We demonstrate how to verify these conditions efficiently given the decomposition. Additionally, we explore heuristic methods for randomized rounding procedures and extract the energy of a feasible point of the QUBO in polylogarithmic time. While the practical relevance of instances where our methods excel remains to be fully established, we propose heuristic algorithms with broader applicability and identify Kronecker graphs as a promising class for applying our techniques. We conduct numerical experiments to benchmark our methods. Notably, by utilizing tensor network methods, we solve an SDP with $D = 2^{50}$ variables and extract a feasible point which is certifiably within $0.15\%$ of the optimum of the QUBO through our approach on a desktop, reaching dimensions millions of times larger than those handled by existing SDP or QUBO solvers, whether heuristic or rigorous.
\end{abstract}

\author{\begingroup
\hypersetup{urlcolor=navyblue}
\href{https://orcid.org/0009-0001-4941-5448}{Haomu Yuan}
\endgroup}
\affiliation{Cavendish Laboratory, Department of Physics, University of Cambridge, Cambridge CB3 0HE, UK}
\email[Haomu Yuan ]{hy374@cam.ac.uk}

\author{\begingroup
\hypersetup{urlcolor=navyblue}
\href{https://orcid.org/0000-0001-9699-5994}{Daniel Stilck Fran\c{c}a}
\endgroup}
\affiliation{Department of Mathematical Sciences\\ University of Copenhagen\\ Universitetsparken 5, 2100 Denmark}
\email[Daniel Stilck Fran\c ca ]{dsfranca@math.ku.dk}

\author{\begingroup
\hypersetup{urlcolor=navyblue}
Ilia Luchnikov, Egor Tiunov, Tobias Haug, Leandro Aolita
\endgroup}
\affiliation{Quantum Research Center, Technology Innovation Institute, Abu Dhabi, UAE}
\date{\today}
\maketitle

\tableofcontents

\section{Introduction}

Quadratic Unconstrained Binary Optimization (QUBO) problems are ubiquitous across computer science, operations research, and statistical physics~\cite{Lucas2014}, but are NP-hard in general to solve. This has motivated a flurry of heuristic and rigorous algorithms to (approximately) solve them~\cite{Blekos2024}, even with special chips dedicated to them~\cite{Johnson2011,Yin2024,Mohseni2022,Alom2017}. A landmark advance was the semidefinite programming (SDP) relaxation of Goemans and Williamson~\cite{goemans1995improved}, which runs in polynomial time in the size of the instance and yields both an upper bound on the optimal QUBO value and, via randomized rounding, near-optimal feasible solutions with provable guarantees.

However, the polynomial degree to solve the Goemans-Williamson SDP using standard algorithms like interior point methods~\cite{boyd2004convex} is high and also are their memory requirements, restricting the dimension of problems that can be effectively solved. This prompted research into algorithms that have a better scaling in time and memory, such as the matrix multiplicative weight (MMW) method~\cite{Arora2005} or those based on sketching techniques~\cite{Yurtsever2021}. Interestingly for quantum and quantum-inspired algorithms, the MMW draws a connection between solving SDPs and preparing quantum Gibbs states.
Subsequent quantum algorithms for SDPs have exploited this connection and obtained polynomial speedups over classical methods~\cite{brandao2017quantum,vanApeldoorn2020,GSLBrandao2022fasterquantum}.  However, these speedups are typically modest (e.g. sub-quadratic in the dimension), require QRAM or other strong data-access assumptions, or depend on SDP parameters that are difficult to control in practice, casting doubt on their practical impact~\cite{henze2025solvingquadraticbinaryoptimization,Dalzell2023}.

At the same time, the theory of quantum Gibbs sampling is currently advancing rapidly~\cite{chen_generators,Ding2025,gilyen2024quantumgeneralizationsglaubermetropolis,jiang2024quantummetropolissamplingweak}, and many physical Hamiltonians are now known to admit both efficient quantum state preparation (e.g. through dissipative preparation~\cite{efficient_gibbs,Tong2025,rouze2024optimalquantumalgorithmgibbs,ramkumar2024mixingtimequantumgibbs,zhan2025rapidquantumgroundstate,smid2025polynomialtimequantumgibbs} or
efficient classical simulation~\cite{Bakshi2024,Crosson_2025}, in particular through tensor networks~\cite{Bauls2023}. Unfortunately, most existing SDP instances demand Gibbs states of highly non-local or unstructured matrices, for which neither quantum nor classical methods achieve superquadratic speedups. Furthermore, most practically relevant SDPs have a number of constraints that scales polynomially with the dimension of the problem. These roadblocks prevented exponential speedups for MMW-based SDP solvers.

In this work, we bridge this gap by identifying a family of SDP instances, arising from the Goemans-Williamson relaxation of QUBO on certain graph families, satisfying:
\begin{itemize}
    \item [1.] \emph{The SDP constraints reduce to estimating exclusively local observables of a many-body quantum Gibbs state at finite temperature;}
    \item [2.] The cost matrices of instances admit a \emph{Pauli-sparse} decomposition whose underlying Paulis are subject to certain algebraic conditions.
\end{itemize}
Furthermore, as we show, these conditions can be efficiently verified.

Under these conditions and assuming that the underlying quantum Gibbs states can be prepared efficiently, using either a quantum computer or a classical simulator, we show how MMW methods can be implemented to solve the Goemans-Williamson SDP to relative precision \(\epsilon\) in time
\[
\mathcal{O}\bigl(\textrm{polylog}(D)\,\mathrm{poly}(1/\epsilon)\bigr)
\]
where $D=2^n$ is the dimension of the QUBO.  This represents an \emph{exponential} speedup in \(D\) over all known classical and quantum SDP solvers for these instances. Furthermore, we identify QUBO instances whose cost matrix admits a decomposition as a $1D$ Hamiltonian for which all of our conditions are met, leading to either exponential classical or quantum speedups.

A hallmark of the GW SDP is that its optimal solution can be “randomly rounded’’ to a concrete \(\{\pm1\}\)-assignment with provable quality guarantees~\cite{Alon2004,goemans1995improved}.  In the high-dimensional regime we target, however, directly performing the rounding and then writing down an assignment would cost \(\Omega(2^n)\) time—far too expensive.  Instead, we show how to do the rounding with much cheaper “local’’ randomness and then use simple Monte Carlo sampling methods to recover the QUBO cost of the assignment up to multiplicative error $\epsilon$ in only $\mathcal{O}(\polylog(D)\epsilon^{-2})$ classical time.  This routine lets us certify that our output string achieves almost the same value as the SDP, but without ever writing down a length-\(2^n\) vector. However, it should be said that our simplified randomized rounding routine still does not have any performance guarantees, i.e. we cannot prove an approximation ratio like it is the case for the the standard Goemans-Williamson randomized rounding routine~\cite{goemans1995improved}, and leave establishing such guarantees to future work.
In combination with our exponential-speedup SDP solver, this %
yields end-to-end quantum-inspired algorithms that approximate QUBO in time essentially polylogarithmic in the ambient dimension. We confirm this empirically by using our methods to find certifiable, %
high-quality solutions of a QUBO with $2^{50}$ variables on a desktop computer%
---a scale significantly beyond what was previously achieved using rigorous or heuristic methods. That being said, it should be mentioned that we yet have to develop a full understanding of the complexity of the underlying QUBO and SDP relaxations for which we obtain exponential speedups, i.e. if they still correspond to QUBO problems that cannot be solved in polynomial time in the input's size under some complexity theoretic assumption.

While the ultimate practical relevance of the instances where we have rigorous exponential speedups remains to be fully explored, our work opens several promising directions to achieve practical value. We develop various heuristic methods to further relax Goemans-Williamson based on the regimes where we have exponential speedups to more general Pauli-sparse instances and numerically show that they perform well. In addition, we identify Kronecker graphs, which are widely used in modelling large-scale networks~\cite{kronecker}, as a promising class of practically relevant instances which could approximately satisfy our assumptions and also perform initial numerical experiments to confirm this. 

Thus, our results establish, for the first time, that particularly-structured
SDP relaxations of QUBO can be solved—and their solutions rounded—exponentially faster than previously thought possible, by exploiting locality in quantum Gibbs states and Pauli‐sparsity in the cost matrix.  By showing that matrix‐multiplicative‐weights, quantum Gibbs sampling, and Monte Carlo rounding can be combined to yield polylogarithmic‐time algorithms in the problem dimension, we hope to open a rich research program at the intersection of convex optimization, quantum many‐body theory, and combinatorial algorithms—one that promises both new theoretical insights and, potentially, practical applications.

\section{Notation and Basic Concepts}\label{sec:notations}
We adopt standard quantum information notation, such as bra–ket notation to denote vectors and matrices. Given a matrix $X$ and $p\in[1,+\infty]$, we denote by $\|X\|_{p}$ its Schatten $p-$norm. When $p$ is omitted, it refers to the operator norm ($p=+\infty$).
Given $n\in\mathbb{N}$, we denote by $[n]=\{1,\ldots,n\}$, and given some $A\subset[n]$, we denote by $2^{A}$ the power set of $A$.

\subsection{Recap of Goemans-Williamson}\label{sec:GWrecap}
Recall that a QUBO problem is specified by a symmetric cost matrix $C \in \mathbb{R}^{D \times D}$ for which we wish to solve:
\begin{equation}
\label{eq:original_problem}   
\textrm{QUBO}(C)=\max_{x \in\{-1,1\}^D}\langle x, C\,x\rangle.
\end{equation}
It is well-known that solving QUBOs is NP-hard in general, and many heuristic and rigorous algorithms exist to approximate it, both quantum and classical~\cite{Kochenberger2014}. Furthermore, they find widespread applications in various fields and are arguably one of the most fundamental combinatorial optimization problems. Indeed, it is known how to cast various other standard optimization problems into instances of QUBOs, which in turn are equivalent to finding the ground state of Ising models~\cite{Lucas2014}.

One of the most widely studied alternatives to approximate the value of a QUBO is the semidefinite relaxation put forth by Goemans and Williamson~\cite{goemans1995improved}. The relaxed problem is given by:
\begin{equation}
\label{eq:GW_relaxation}   
\begin{aligned}
\max_{Y\in \mathbb{R}^{D \times D}}& \tr{CY} \\
\textrm{subject to } &\langle i | Y|i\rangle=1,\ \ 1 \leqslant i \leqslant D \\
&Y \geq 0 .
\end{aligned}
\end{equation}
We denote the value of this problem by $\uGW(C)$, where $\uGW$ stands for unnormalized Goemans-Williamson. It is a standard fact that 
\begin{align}\label{eq:ugw_bound_1}
\uGW(C)\geq \QUBO(C).
\end{align}
However, remarkably, it is known that for various families of matrices $C$~\cite{Alon2004}, the other direction also holds for some constant $\alpha_R$, i.e. 
\begin{align}\label{equ:approximation}
    \alpha_R\uGW(C)\leq \QUBO(C)\leq  \uGW(C).
\end{align}
For instance, for $C$ with only positive entries, $\alpha_R\simeq 0.878$~\cite{goemans1995improved}. As semidefinite programs of the form of~\Cref{eq:GW_relaxation} can be solved in polynomial time in $D$, e.g., using interior-point methods~\cite[Chapter 11.8.3]{boyd2004convex}, Interestingly, under the unique games conjecture, this approximation ratio is optimal for polynomial-time algorithms~\cite{Khot2007}. 

To achieve the approximation ratio, one often resorts to randomized rounding techniques. That is, given a solution to $\uGW(C)$ denoted by $Y^*$ and a distribution $\mu$ on $\mathbb{R}^D$, we can consider the random variable on $x\in \{\pm 1\}^D$ given by $x=\operatorname{sign}(\sqrt{Y^*}y)$, where $y\sim \mu$ and the sign function acts component-wise. One can then derive bounds on the expectation value of this random variable in terms of $\uGW(C)$, and, as it is clearly a feasible point of $\QUBO$, this leads to inequalities of the form in~\Cref{equ:approximation}.

From now on, we only consider $D=2^{n}$, i.e., matrices $C$ that can be seen as acting on systems of $n$ qubits. Then, it is convenient to introduce the normalized Goemans-Williamson (GW) problem:
\begin{equation}
    \label{eq:GW_relaxation_normalized}   
    \begin{aligned}
    \max_{\rho\in\R^{2^n\times 2^n}}& \tr{\frac{C}{\|C\|}\rho} \\
    \textrm{subject to } &\langle i | \rho|i\rangle=\frac{1}{2^n},\ \ 1 \leqslant i \leqslant 2^n \\
    &\rho \geq 0 .
    \end{aligned}
    \end{equation}
We denote the value of this problem by $\GW(C)$, and note that it now takes values in $[-1,1]$\footnote{Note that we normalized by the exact value of the operator norm of the matrix $\|C\|$ here, which is typically hard to compute. But, in practice, it is sufficient to only have an upper bound on its value to ensure that the normalized operator has norm at most $1$.}.
The reason for this normalization is that now we are optimizing over density matrices, as the feasible points of~\Cref{eq:GW_relaxation_normalized} have unit trace and are positive semi-definite, and over bounded observables. Furthermore, it is clear that the value of~\Cref{eq:GW_relaxation_normalized} is just the one of~\Cref{eq:GW_relaxation} divided by $2^{n}\|C\|$. 
Thus, if we solve the $\GW$ problem up to $\epsilon>0$, then we solve $\uGW$ up to $\epsilon\, 2^n\|C\|$. Given that the value $\QUBO(C)$ is in $[-2^n\|C\|,2^n\|C\|]$, we see that $\epsilon\, 2^n\|C\|$ typically corresponds to a relative error of order $\epsilon$ for $\QUBO(C)$. In more physical terms, by solving the normalized problem up to error $\epsilon\, 2^n$, we are approximating the energy density of the Hamiltonian $C/\|C\|$ up to additive error $\epsilon$. In turn, in computer science terms, we are computing the average number of constraints satisfied. 

For the convenience of the reader,~\Cref{tab:GW_hierarchy} in the~\Cref{app:table_qubo_relations} summarizes these and all other relaxations of QUBO we consider in this work.

\subsection{Overview of classical and quantum SDP solvers}
\label{sec:SDP_solvers}
There are multiple algorithmic frameworks to solve semidefinite programs in polynomial time, such as interior-point or ellipsoid methods%
~\cite{Lee2015,boyd2004convex,Augustino2023}. Another powerful approach is the  matrix multiplicative weight (MMW) method~\cite{Arora2005}. 
The main distinction between interior-point methods and MMW is that the run-time of the former has a logarithmic dependency in the target precision $\epsilon$, whereas that of the latter has a polynomial one. However, MMW methods usually feature a better dependence on the matrix dimension, thus being advantageous whenever low-accuracy solutions are good enough. In addition, interior-point methods have high memory requirements. 
Another interesting approach is approximate classical solvers  based on sketching techniques, which have less stringent memory and runtime requirements, such as SketchyCGAL~\cite{Yurtsever2021}. This %
only requires $\mathcal{O}(D)$ memory to approximate a low-rank solution of $\GW$ up to constant precision. 
The MMW approach is particularly appealing to this work because it (as well as the closely related thermodynamic approach of~\cite{liu2025quantumthermodynamicssemidefiniteoptimization}) is the basis of various  quantum SDP solvers~\cite{brandao2017quantum,vanApeldoorn2020,GSLBrandao2022fasterquantum, watts2023quantumsemidefiniteprogrammingthermal}. Here, we use it as basis to build both quantum and quantum-inspired SDP solvers  tailored  to specific relaxations (see Def.~\ref{def:relaxed_SDP} below) of the GW problem in Eq. \eqref{eq:GW_relaxation_normalized}.

Let us now briefly review the run-time scalings of the main classical and quantum GW-SDP solvers. To  our knowledge, the best classical solvers in terms of interior point methods has time complexity $\mathcal{O}(D^4\log(\epsilon^{-1}))$~\cite{Lee2015}, whereas the best MMW-based classical solvers achieve $\tilde{\mathcal{O}}\left(\min \{D^{2.5}s\epsilon^{-2.5},D^{2.5}s^{0.5}\epsilon^{-3.5},D^2s\epsilon^{-9}\}\right)$~\cite{henze2025solvingquadraticbinaryoptimization,Arora2005}, where the tilde hides polylogarithmic factors in $1/\epsilon$ and $D$; and $s$ is the sparsity (i.e. the maximum number of nonzero entries per row) of $C$. On the other hand, the best known MMW-based quantum algorithm achieves $\tilde{\mathcal{O}}\left(D^{1.5}s^{0.5+o(1)}\epsilon^{-15+o(1)}\right)$ \cite{GSLBrandao2022fasterquantum}. In turn, for quantum interior-point methods it is difficult to analyse the performance rigorously~\cite{Dalzell2023}, as it depends on the condition number of various matrices that are constructed along the iterations of the algorithm, which are hard to bound analytically. 
Finally, the SketchyCGAL algorithm from~\cite{Yurtsever2021} requires time $\mathcal{O}(D\,s\epsilon^{-2.5})$. However, one should stress that its output is different from those of the previously mentioned methods, as it does not give the solution matrix but only a rank-one approximation.

\subsection{The Hamiltonian updates framework to solve GW}
\label{sec:HU}

From now on, we focus exclusively on the \textit{Hamiltonian updates} (HU) framework of~\cite{GSLBrandao2022fasterquantum}, with algorithmic details provided in~\Cref{sec:HU_description}. The HU framework allows one %
to reduce approximately solving SDPs to preparing quantum Gibbs states and approximating observable expectation values on them, as explained below. Given symmetric matrices $B_0,\ldots,B_m\in\mathbb{R}^{2^n\times 2^n}$ s.t. $\|B_i\|\leq 1$, real numbers $b_i$ for $1\leq i\leq m$ (in our case of interest, all $b_i$'s are zero), and a precision parameter $\epsilon>0$, we can use HU to solve the relaxed SDP:
\begin{equation}\label{eq:general_SDP}   
\begin{aligned}
&\max_{Y\in \mathbb{R}^{2^n\times 2^n}
} \tr{B_0Y} \\
&\textrm{subject to }  b_i-\epsilon\leq \tr{B_iY}\leq b_i+\epsilon, 1\leq i\leq m\\
&\tr{Y}=1\\
&Y \geq 0 .
\end{aligned}
\end{equation}
up to precision $\epsilon$ (i.e. we can determine the optimal value up to additive precision $\epsilon$) by preparing Gibbs states. %
More precisely, we need to estimate the expectation value of $B_i$ %
up to precision $\mathcal{O}(\epsilon)$ on a sequence of $\mathcal{O}(n\epsilon^{-2}\log(\epsilon^{-1}))$ Gibbs states of the form:
\begin{align}\label{equ:Gibbs_parametrized_0}
\sigma(\lambda)\propto \exp\left(
\sum_{i=0}^m\lambda_iB_i
\right),
\end{align}
where $\lambda\in\mathbb{R}^{m+1}$, with $\|\lambda\|_{\ell_1}=\sum_{i=0}^m|\lambda_i|=\mathcal{O}(n\epsilon^{-1})$. The sequence of Gibbs states is iteratively constructed given the expectation values, and the details of how to construct them are discussed in~\Cref{sec:HU_description}. For the reader's convenience, we give a summary of the framework here. In a nutshell, one %
reduces the optimization in~\Cref{eq:general_SDP} to a series of feasibility problems of the form 
\begin{equation}\label{eq:feasibility_SDP}   
    \begin{aligned}
    \max_{Y\in\mathbb{R}^{2^n\times 2^n}%
    } &1 \\
    &\textrm{subject to }  b_i-\epsilon\leq \tr{B_iY}\leq b_i+\epsilon, 1\leq i\leq m\\
    &\tr{B_0Y}\geq \mu\\
    &\tr{Y}=1\\
    &Y \geq 0 ,
    \end{aligned}
    \end{equation}
and performs %
a binary search over $\mu$ to find the maximal $\mu$ for which the program is feasible. The algorithm starts by checking if
$\sigma(0)=I/2^n$ is a feasible point and constructs new guesses $\sigma(\lambda_t)$ iteratively. At each step, the algorithm computes $
\tr{B_i\,\sigma(\lambda_t)}$ and checks if the constraints are satisfied up to $\epsilon$. If we find that constraint $i$ is violated by more than $\epsilon$, we set $\lambda_{t+1}=\lambda_t+\frac{\epsilon}{4}e_i$, where $e_i$ is the $i$-th element of the Euclidean basis of $\mathbb{R}^{m+1}$, and continue. If all constraints are met, we output that a given value of $\mu$ is feasible. One can then show that, if such a feasible point exists, the algorithm converges in $\mathcal{O}(n%
\epsilon^{-2})$ iterations. Thus, if we have not converged to a feasible point after this many iterations, we conclude that the problem is not feasible for that value of $\mu$ and try again with a lower value, until the desired optimal feasible value is found through the binary search.

\subsection{Pauli strings and Pauli-sparse matrices}
We make extensive use of the Pauli matrices $X,Y,Z$ and $I$. Given $(x,z)\in \{0,1\}^{2n}$, we denote by $P_{(x,z)}$ the $2^n\times 2^n$ complex matrix given by the Pauli string:
\begin{align}
  P_{(x,z)} := \bigotimes_{j=1}^n i^{x_j z_j} X_j^{x_j} Z_j^{z_j},
\end{align}
where $X_j$ and $Z_j$ are the usual Pauli matrices acting on the $j$-th qubit. Furthermore, for a Pauli string $P_{(x,z)}$, we denote by $w(P_{(x,z)})$ the number of its Paulis that do not correspond to the identity $I$. We refer to this as the weight of the Pauli string. We denote by $\mathcal{Z}_n$ the subgroup of Pauli $Z$ strings (including the sign), i.e. strings of the form $\pm P_{(0,z)}$. As Pauli $Z$ strings play a special role in our work, we further introduce a simpler notation for them. Given some $A\subset[n]$, we let $Z_A=P_{(0,z_A)}$, where $z_A$ is the vector with ones for entries in $A$ and $0$ otherwise. 

We often use the fact that the group of Pauli strings under multiplication modulo the sign is isomorphic to the group $(\mathbb{Z}_2^n\times \mathbb{Z}_2^n,+)$---see e.g.~\cite{Aaronson2004} for an overview of various properties of this correspondence. In particular, this representation allows us to easily compute various algebraic properties of a set of Pauli strings in terms of linear algebra over finite fields. For instance, if we wish to determine the subgroup that a set of Pauli strings $\{P_{(x_i,z_i)}\}_{i=1}^m$ generates, all we need to do is to find the span (with respect to addition modulo 2) of the set of vectors $\{(x_i,z_i)\}_{i=1}^m$.

One object that is extremely relevant to rigorously identify regimes where our methods provide large speedups is the diagonal Pauli group associated to $C$:
\begin{definition}[Diagonal group of $C$]\label{def:diagonal_group}
Given a symmetric matrix $C\in\mathbb{R}^{2^n\times 2^n}$ with Pauli expansion
\begin{align}\label{equ:Pauli_expansion2}
C=\sum_{(x,z)\in\{0,1\}^{2n}}c_{(x,z)}\,P_{(x,z)},
\end{align}
we define $\mathcal{G}_C$ to be the group of Paulis generated by $\{P_{(x,z)}:c_{(x,z)}\not=0\}$. With this, we then define the diagonal group generated by %
$C$, $\mathcal{D}_C$, to be given by $\mathcal{D}_C=\mathcal{G}_C\cap \mathcal{Z}_n$, representing the Pauli $Z$ strings contained in $\mathcal{G}_C$.
Moreover, it is also be useful to denote the traceless subset of $\mathcal{D}_C$ by $\dot{\cD}_C:=\cD_C\backslash I$; this is the actual set that defines the constraints in our SDP relaxation.
\end{definition}

The Pauli strings form a basis of the set of Hermitian matrices, i.e. any Hermitian or symmetric matrix $C$ can be expanded as a linear combination of Pauli strings with real coefficients. One of the main focus of this work is to consider instances of QUBOs whose cost matrix $C$ is (approximately) sparse when expressed over the Pauli basis. That is, it can be expressed as 
\begin{align}\label{equ:Pauli_expansion1}
C=\sum_{(x,z)\in\{0,1\}^{2n}}c_{(x,z)}\,P_{(x,z)},
\end{align}
with $|c_{(x,z)}|\not=0$ for substantially less than $4^{n}$ terms.  More precisely:
\begin{definition}[Pauli-sparse]\label{equ:pauli_sparse}
A matrix $C\in\mathbb{R}^{2^n\times 2^n}$ is called Pauli-sparse if its expansion in the Pauli basis
\begin{align}
C=\sum_{(x,z)\in\{0,1\}^{2n}}c_{(x,z)}\,P_{(x,z)},
\end{align}
is such that 
\begin{align}
|\{c_{(x,z)}:c_{(x,z)}\not=0\}|=\operatorname{poly}(n)\ll4^n.
\end{align}
Furthermore, it will be called of weight at most $k$ if  
\begin{align}
\forall c_{(x,z)}\not=0\implies w(P_{(x,z)})\leq k
\end{align}
and of low weight we have $k=\cO(1)$.
\end{definition}
\begin{remark}
Note that if $C$ is of weight at most $k=\mathcal{O}(1)$, then this implies that $C$ is Pauli-sparse, as there are $\sum\limits_{l=1}^k 3^l {n \choose l}=\operatorname{poly}(n)$ Paulis of weight at most $k$.
\end{remark}
\begin{remark}
    In the case where we have that only $\poly(n)$ terms are not $0$, this gives a succinct representation of the problem~\cite{Galperin1983}, as we can query any entry of the matrix in time $\poly(n)$ and store this representation only using $\poly(n)$ memory.
\end{remark}
\begin{remark}
    Although such a sparse decomposition is natural from a quantum computing point of view, as e.g. quantum many-body Hamiltonians are Pauli-sparse, and they have been considered in other works~\cite{Wang2024,Cifuentes2024}---to the best of our knowledge, Pauli-sparse QUBOs have not been considered before in the literature.
\end{remark}

In \Cref{sec:classification_graphs}, we discuss some structural results on Pauli-sparse graphs, i.e. which graphs can be expressed in such terms and connect local Pauli-sparse graphs to subgraphs of the Hamming ball graphs. Importantly, we show that practically relevant graphs(\Cref{sec:kronecker_graphs}), such as some Kronecker graphs~\cite{kronecker}, satisfy this assumption and it is possible to find such a representation in time $\poly(n)$.

\begin{remark}
In this work, we %
focus on adjacency matrices that admit a sparse representation in the Pauli basis, which establishes a connection between solving the underlying SDPs and quantum many-body problems with local dimension $d=2$. However, most of the techniques we discuss easily extend to qudits, i.e. arbitrary local dimension. All we need to do is to replace the Pauli basis with its natural unitary extension to qudits, the discrete Weyl matrices~\cite[Example 2.1]{wolf2012quantum}. Indeed, they also have the property of being orthogonal to each other with respect to the Hilbert-Schmidt scalar product and containing the identity and are a projective representation of $\mathbb{Z}_d\times\mathbb{Z}_d$. Furthermore, a subset of them defines a basis for diagonal matrices, providing us with a natural analog of how to relax the diagonal constraint of GW is equivalent to just enforcing that the expectation of a subset of them is $0$. However, for simplicity, we restrict throughout to the $d=2$ case
and comment on possible extensions where appropriate. 
\end{remark}

Our techniques also work if the underlying matrix is not exactly Pauli-sparse, but approximately. To that end, it is important to consider the following quantity:
\begin{definition}[Pauli $\ell_1$ norm]\label{defi:Pauli_norm}
Let $C\in\mathbb{R}^{2^n\times 2^n}$ with Pauli expansion
\begin{align}\label{equ:Pauli_expansion}
C=\sum_{(x,z)\in\{0,1\}^{2n}}c_{(x,z)}\,P_{(x,z)}.
\end{align}
We define its Pauli $\ell_1$ norm, $\|C\|_{P,\ell_1}$, as 
\begin{align}
\|C\|_{P,\ell_1}=\sum_{(x,z)\in\{0,1\}^{2n}}|c_{(x,z)}|.
\end{align}
\end{definition}

Interestingly, as we %
see in~\Cref{sec:pauli_sparsification}, the quantity $n\|C\|_{P,\ell_1}^2$ quantifies how many Pauli terms we need in a  matrix for it to faithfully approximate $C$ in terms of its QUBO value.
This implies that, when $C$ is itself not Pauli-sparse but $\|C\|_{P,\ell_1}=\operatorname{poly}(n)$, $\QUBO(C)$ can still be approximated via a Pauli-sparse sparsification of $C$.
\section{Main results}\label{sec:main_results}

We now present our main results. As mentioned before, we %
consider the problems $\QUBO$ and $\GW$ for a matrix $C$ that is (approximately) Pauli-sparse. We have results both for quantum and classical implementations of the algorithms developed here.
Our main insights are twofold: first, we show (in \Cref{thm:main1_cons}) that by considering the group properties of the Paulis in the Pauli expansion of $C$, it is possible to relax the $\GW$ problem without losing the quality of relaxation, but substantially reducing the number of constraints. The details of this procedure are described in \Cref{sec:further_relaxing}. Moreover, certain Pauli-sparse matrices, we even show (in \Cref{cor}) that it is possible to exponentially reduce the number of constraints to solve $\GW$ without affecting the tightness of the GW %
relaxation. Our second %
insight (\Cref{proposition}) is that, for some of these instances, solving $\GW$ can be reduced to a quantum many-body problem, namely approximating the expectation value of a local observable of a quantum many-body state. This makes it possible to solve huge instances by resorting to quantum or classical algorithms to simulate quantum many-body systems.

\subsection{Rigorous results on solving Goemans-Williamson}\label{sec:riggw}
Our work is based on the following further relaxation of GW, which was also considered in~\cite{Patti2023}. 
\begin{definition}[Relaxed Goemans-Williamson]\label{def:relaxed_SDP}
Given $S \subseteq 2^{[n]}, \epsilon>0$ and a symmetric $C\in\mathbb{R}^{2^n\times 2^n}$, consider the SDP
\begin{equation}\label{equ:relaxed_SDP2}
\begin{aligned}
\max%
_{\rho} & \tr{\frac{C}{\|C\|}\rho}\\ 
\textrm{subject to } & \forall A \in S: |\tr{Z_{A}\rho}| \leq \epsilon, \\
& \tr{\rho}=1, \rho \geq 0.
\end{aligned}
\end{equation}
We define $\GW(C, S, \epsilon)$ as the value of the program %
in~\Cref{equ:relaxed_SDP2}. Note that this specific formulation is obtained from the generic SDP in Eq. \eqref{eq:general_SDP} by taking $B_0=\frac{C}{\|C\|}$, $B_i=Z_{A}$ for $i \in [1, m]$, and $m=|S|$. 
\end{definition}

Our first main result is as follows:
\begin{thm}[Reducing the number of constraints, informal]\label{thm:main1_cons}
With %
$\dot{\cD}_C$ the traceless subset of the diagonal group $\cD_C$ generated by $C$, defined in~\Cref{def:diagonal_group}, we have:
\begin{align}\label{equ:main1_cons_eq1}
\GW(C, \dot{\cD}_C, \epsilon)=\GW(C, \epsilon).
\end{align}
Furthermore, it holds that:
\begin{align}\label{equ:stability}
|\GW(C, \dot{\cD}_C%
, \epsilon)-\GW(C, 0)|=\mathcal{O}\Big(\epsilon^{\tfrac{1}{3}}|\dot{\cD}_C%
|^{\tfrac{1}{6}}\Big).
\end{align}
\end{thm}
 The precise statement and proof of this theorem are given in \Cref{prop:stability_sdp} in \Cref{sec:how_to_relax}, where we also show how to obtain a relative error $\epsilon$ in Eq.~\eqref{equ:stability} for some cases. The significance of~\Cref{thm:main1_cons} is that we can use it to find instances of matrices $C$, namely those for which $\dot\cD_C$ is small, such that the number of constraints we need to consider is smaller than what the naive implementation of $\GW$ would require: %
$2^n-1$ constraints. %
Furthermore,~\Cref{equ:stability} shows that solving the problem to a precision that is comparable to the size of $\dot\cD_C$ suffices to obtain a good solution to the original problem.
Hence, by virtue of ~\Cref{thm:main1_cons}, we identify instances $C\in\mathbb{R}^{2^n\times 2^n}$ --- namely, those for which $|\dot{\cD}_C|=\poly(n)$ --- for which one can approximate $\textrm{GW}(C)$ in time $\poly(n,\epsilon^{-1})$ but for which no method is known to compute $\textrm{QUBO}(C)$ in time less than doubly exponential in $n$, remarkably.
To prove the theorem, we resort to the HU procedure summarized in~\Cref{sec:HU} to solve the SDP and show that all the constraints that are not in $\dot\cD_C$ are satisfied automatically and hence do not need to be taken into account. Our proof of %
\Cref{equ:stability} also leverages %
this insight to obtain improved bounds on the total variation distance between the diagonal of a solution and the maximally mixed state, which can then be used to import other stability bounds for SDPs~\cite{henze2025solvingquadraticbinaryoptimization}.

One simple corollary of this result is:
\begin{cor}[Condition for spectral relaxation, informal]
If $\dot{\cD}_C=\varnothing$,%
, then:
\begin{align}
\uGW(C)=\lambda_{\max}(C)
\end{align}
with $\lambda_{\max}(C)$ the maximal eigenvalue of $C$. Furthermore, if $\lambda_{\max}$ has multiplicity $1$, then the corresponding (unnormalized) eigenvector $v_{\max}$ (with $\ell_2$ norm $\sqrt{2^n}$) is the solution to $\operatorname{QUBO}(C)$.
\label{cor}
\end{cor}

That is, for matrices with a trivial diagonal group and nondegenerate leading eigenvalue, we can find %
$\QUBO(C)$ in a time that is polynomial in the %
dimension of $C$ by diagonalization. 
Note that, although the matrices in question have a restricted structure, obtaining their maximal eigenvalue is still in general NP-Hard (in the input size, which is $\textrm{poly}(n)$): It is not difficult to see that Pauli-sparse matrices with only $X$ Pauli strings have trivial diagonal groups and correspond to solving NP-Hard combinatorial optimization problems. 
Moreover, as far as we know, there is no result showing that one can solve the corresponding QUBO in time polynomial in the dimension even if the diagonal group is trivial but the maximal eigenvalue is not unique. 
Interestingly, we will see numerical examples in \Cref{sec:numerics} that demonstrate that $\GW(C,\dot\cD_C)< \lambda_{\max}(C)$ even when $|\dot\cD_C|=3$%
, showing that already a small number of constraints can significantly improve the bound on $\QUBO(C)$ provided by the SDP when compared to the spectral bound.

The fact that we have fewer constraints when $\dot\cD_C$ %
is small can be leveraged to obtain large speedups to solve $\GW(C)$ using the HU framework:
\begin{proposition}[Complexity of solving $\GW$, informal]
Given a symmetric $C\in\mathbb{R}^{2^n\times 2^n}$, define
\begin{align}\label{equ:Gibbs_parametrized}
\sigma(\lambda)\propto \exp\left(\lambda_0\tfrac{C}{\|C\|}+\sum_{Z_A\in\dot\cD_C%
}\lambda_iZ_A\right),
\end{align}
where $\lambda\in\mathbb{R}^{|\mathcal{D}_C|}%
$ with $\|\lambda\|_{\ell_1}=\mathcal{O}(n|\mathcal{D}_C|\,\epsilon^{-1})$. Let $\operatorname{Cost}_C(\sigma(\lambda),\epsilon)$ and $\operatorname{Cost}_Q(\sigma(\lambda),\epsilon)$ be the worst-case cost to estimate the expectation value of all $Z_A\in\dot\cD_C$ and $\tfrac{C}{\|C\|}$ up to additive precision $\epsilon$ on either a classical or quantum computer. Then, we can solve up to additive precision $\epsilon$ in time:
\begin{align}
\mathcal{O}(n\epsilon^{-2}\textrm{poly}(|\mathcal{D}_C|)\operatorname{Cost}_C(\sigma(\lambda),\epsilon))\quad \text{and}\quad \mathcal{O}(n\epsilon^{-2}\textrm{poly}(|\mathcal{D}_C|)\operatorname{Cost}_Q(\sigma(\lambda),\epsilon))
\end{align}
on a classical and quantum computer, respectively.\footnote{The attentive reader likely noticed that in the statement above $\|\lambda\|_{\ell_1}$ scales with $|\dot\cD_C|$, whereas in most other statements it is independent of $|\dot\cD_C|$. The reason for that is that here we are solving $\GW$, and to relate the value of $\GW$ to that of the relaxed one, we need to rescale the precision by $|\dot \cD_C|$.}
\label{proposition}
\end{proposition}
Importantly, as we discuss in more detail in~\Cref{sec:numerics}, there are examples (see last line of~\Cref{tab:rigorous_speedups}) of matrices $C$ such that, either $\operatorname{Cost}_C(\sigma(\lambda),\epsilon)$ or $\operatorname{Cost}_Q(\sigma(\lambda),\epsilon)$, $|\mathcal{D}_C|$ are \emph{polynomial in $n$}. Thus, for these instances, our work achieves \textbf{exponential speedups compared to state-of-the-art} solvers. To demonstrate this point and the power of our method, in \Cref{sec:numerics} we approximately solve SDPs for an instance in dimension $2^{50}$ on a desktop computer, which is several orders of magnitude larger than what is achievable using state-of-the-art methods based on sketching~\cite{Yurtsever2021}, which can handle up to dimension $\sim 2^{12}$.

That being said, our rigorous results are restricted to instances where both ${|}\mathcal{D}_C{|}$ and $\operatorname{Cost}_C(\sigma(\lambda),\epsilon)$ are polynomial in $n$, which is certainly not a generic property. Some examples for classical computers are commuting models on a tree with tensor networks~\cite{Bauls2023} and on a quantum computer we can prepare the more general class of non-commutative, $1D$ models using the methods of~\cite{Bilgin2010}, albeit only for $\epsilon=\Omega(1)$.
In \Cref{sec:classification_graphs}, we start the classification of which graphs satisfy such constraints. In contrast, methods like that of~\cite{Yurtsever2021} just require sparsity of $C$, i.e. that the number of nonzero entries per row of $C$ remains small. Note that is a weaker condition than Pauli sparsity.

Our rigorous findings are informally summarized in~\Cref{tab:rigorous_speedups}. We also refer the reader to \Cref{sec:pauli_sparse} to a discussion on which families of graphs are Pauli-sparse.

\begin{table}[!ht]
    \centering
    \setlength{\tabcolsep}{5pt}
    \begin{tabular}{p{0.26\textwidth} p{0.34\textwidth} p{0.34\textwidth}}
        \toprule
        \bf Algorithm & Classical & Quantum \\
        \midrule
        \bf Assumptions on $C\in\mathbb{R}^{2^n\times 2^n}$ & $|\mathcal{D}_C| = \poly(n)$. & $|\mathcal{D}_C| = \poly(n)$.  \\
        \midrule
        \bf Assumptions on Gibbs state &Efficient classical algorithm to compute $\tr{\sigma(\lambda)O}$ for $O\in\{C/\|C\|,\dot\cD_C\}$ up to $\sim \epsilon$. &Efficient quantum algorithm to compute $\tr{\sigma(\lambda)O}$ for $O\in\{C/\|C\|,\dot\cD_C\}$ up to $\sim \epsilon$. \\
        \midrule
        \bf Conclusion & Solve $\GW(C)$ to additive $\epsilon$ in $\poly(n,\epsilon^{-1})$ time ($[\GW(C, \epsilon)$, \Cref{equ:main1_cons_eq1}]). & Solve $\GW(C)$ to additive $\epsilon$ in $\poly(n,\epsilon^{-1})$ time ($[\GW(C, \epsilon)$, \Cref{equ:main1_cons_eq1}]). \\
        \midrule
        \bf Examples & $C$ and $\dot\cD_C$ \emph{local, commuting} Hamiltonian on tree. & $C$ and $\dot\cD_C$ local, $1D$-Hamiltonian, $\epsilon=\Omega(1)$. \\    
        \bottomrule
    \end{tabular}
        \caption{Rigorous speedups for computing the normalized Goemans–Williamson value $\GW(C)$ under a small diagonal algebra and efficient Gibbs state preparation for classical and quantum computers. See~\Cref{equ:Gibbs_parametrized} for a definition of $\sigma(\lambda)$. For the classical case,  tensor network algorithms~\cite{Bauls2023} can compute the properties of the Gibbs states with the claimed efficiency. For the quantum case,  the results of~\cite{Bilgin2010} can be used to prepare the quantum Gibbs states. }
    \label{tab:rigorous_speedups}
\end{table}

    \subsection{Solving Goemans-Williamson: steps into making our results practically relevant}

    Although, to the best of our knowledge, our results give the first examples of nontrivial SDPs corresponding to convex relaxations of combinatorial optimization problems that can be solved in polylogarithmic time given an appropriate succinct classical representation, it is unclear at this stage to what extent such instances are of practical relevance. Thus, we also consider instances that have a structure similar to ours, but known practical relevance.
    
    To this end, we use our rigorous results as a departure point to obtain efficient and effective relaxations of Goemans-Williamson. Indeed, it should be noted that the SDP in \Cref{def:relaxed_SDP} is always be an upper bound on $\GW(C)$ regardless of the choice of the subset of constraints $S$. However, in the case where ${|}\mathcal{D}_C{|}$ is too large, using our \Cref{thm:main1_cons} to enforce all the relevant constraints might be too costly. Thus, in \Cref{sec:how_to_relax} we give various procedures on how to pick the set $S$ in an effective way and benchmark these techniques numerically in \Cref{sec:numerics}. Our main insight is that Pauli $Z$ strings that can be formed by the product of few Pauli strings in the decomposition of $C$ should be the first to be tried as a set of constraints.
    
    Furthermore, we have the following result on which matrices $C$ can be approximately decomposed in a sparse way in the Pauli basis. Results like this are well-known e.g. in the Hamiltonian simulation literature~\cite{Campbell2019}, as they are the backbone of randomized compilation methods like qDRIFT.

    \begin{proposition}[Pauli sparsification]\label{thm:pauli_sparsification1}
    Let $C\in\mathbb{R}^{2^n\times 2^n}$ with Pauli expansion
    \begin{align}\label{equ:Pauli_sparsification}
    C=\sum_{(x,z)\in\{0,1\}^{2n}}c_{(x,z)}\,P_{(x,z)}.
    \end{align}
    Then, for every $\epsilon>0$, there exists a matrix $\tilde{C}$ that is $\mathcal{O}(n\|C\|_{P,\ell_1}^2\epsilon^{-2})$ Pauli-sparse and such that:
    \begin{align}
    \|C-\tilde{C}\|\leq \epsilon.
    \end{align}
    Furthermore, we can obtain such a $\tilde{C}$ with probability of success $2/3$ given $\mathcal{O}(n\|C\|_{P,\ell_1}^2\epsilon^{-2})$ samples from the distribution $p$ on Paulis with density  $p(x,z)=\frac{|c_{(x,z)}|}{\|C\|_{P,\ell_1}}$.
    \end{proposition}
    To prove this, we resort to standard matrix concentration inequalities~\cite{Tropp2011}.
    One example of a family of graphs for which we can use \Cref{thm:pauli_sparsification1} to efficiently sparsify are so-called Kronecker graphs~\cite{kronecker}. We discuss why this is the case and give their precise definition in \Cref{sec:kronecker_graphs}. But here we note that such graphs find various applications in modelling networks in a highly efficient manner~\cite{kronecker}.
    
    Thus, these families \emph{provide us with practically relevant graphs that are approximately Pauli-sparse and for which we can find such a representation efficiently.} In principle it is possible to generate Kronecker graphs that have small ${|}\mathcal{D}_C{|}$, but we test more generally how our methods work for such graphs in \Cref{sec:numerics}.
    
    Unfortunately, we are not aware of classes of Kronecker graphs for which $\operatorname{Cost}_C$ is in general $\poly(n)$. As such, they give a natural family of Gibbs states to test out recently developed dissipative preparation methods~\cite{chen_generators,Ding2025} to see if it is possible to reach an exponential advantage for solving such semidefinite relaxations on quantum computers.

\subsection{Heuristic rounding algorithms}
As mentioned in \Cref{sec:notations}, one of the attractive features of the GW SDP is that the solution to the SDP can be used to obtain a feasible point of the QUBO which performs well. Thus, this way it is possible to approximate the optimal value from above and below and obtain good solutions efficiently. We now describe procedures to obtain feasible points of the QUBO in our setting and benchmark them, focusing on classical computers for now.
Although we leave obtaining provable performance guarantees to future work, in \Cref{sec:numerics} we show that our methods perform well in practice and we can still benchmark the feasible point we obtain in time $\poly(n)$ in some cases. But before we give our methods, the attentive reader should have already noted that it is impossible to specify an arbitrary feasible point of the QUBO in time $\poly(n)$, as it requires $2^n$ bits. Instead we focus on estimating quantities of the form $2^{-n}\bra{x}C\ket{x}$ up to an additive precision $\epsilon>0$ for some point $x\in\{\pm1\}^{2^n}$ we obtain from our SDP solution. Thus, even though we cannot write out the explicit solution $X$, we can still estimate expectation values $2^{-n}\bra{x}C\ket{x}$ it achieves. 
Now, let us start with our method.

\paragraph{Rounding on classical computers}
Recall that the usual randomized rounding technique of Goemans-Williamson to obtain good feasible points is to take a solution $\rho$ to the SDP and consider the random vector $\operatorname{sign}(\sqrt{\rho}\ket{N})$, where $\ket{N}$ is a vector with i.i.d. normal entries. For some of the methods we use to classically simulate the Gibbs states $\rho$, i.e. tensor networks~\cite{Bauls2023}, it is feasible to compute $\bra{i}\sqrt{\rho}\ket{\psi}$ to high precision as long as the state $\ket{\psi}$ is not ``too complex'', i.e. a MPS of polynomial bond dimension. 
Inspired by this observation, our randomly rounded vectors for this case are of the form $\operatorname{sign}(\sqrt{\rho}\,U\ket{0})$, for $U$ being a sufficiently shallow random real circuit s.t. we can still compute $\bra{i}\sqrt{\rho}\,U\ket{0}$ efficiently, where $\rho$ is the solution to the SDP. Note that $\rho$ is given as a Gibbs state, so computing $\sqrt{\rho}$ from it is trivial: we only need to halve the inverse temperature. In addition, as we are ultimately interested only in the sign, it is not necessary to normalize the trace of $\sqrt{\rho}$, i.e. no partition-function estimation is required.

Let us discuss the intuition behind this construction of randomized rounding solutions. First, note that the usual rounding strategy would correspond to taking $U$ uniformly at random from the orthogonal group, as then $U\ket{0}$ would be uniformly distributed on the sphere and we are back in the usual setup of GW\footnote{note that the distribution is invariant under rescaling the random vector by a positive number, and a normalized Gaussian is uniformly distributed. Thus, it does not matter if we round with a Gaussian or uniformly distributed vector}. But producing such an $U$ incurs an exponential cost. However, a shallow quantum circuit is an approximate design~\cite{chen2024incompressibilityspectralgapsrandom} and hence approximate the uniform distribution over orthogonal matrices. On top of that, %
the closer the solution of the SDP is %
to a pure state, i.e. $\rho\simeq\ketbra{\psi}$\footnote{Note that, ideally, one wants $\ket{\psi}$ to be a phase state proportional to a string $x$ that solves Eq. \eqref{eq:original_problem}.}, the more irrelevant %
the random vector used in the rounding becomes%
, as $\operatorname{sign}(\bra{\psi}U\ket{0}\ket{\psi})$ depends on $U$ by an irrelevant global phase. 

Now that we have motivated our choice of rounding algorithms, we show how to use Monte Carlo techniques to extract expectation values from the rounded point efficiently:
\begin{proposition}
Let $C\in\mathbb{R}^{2^n\times 2^n}$ be a Pauli-sparse matrix with
\begin{align}
C=\sum\limits_{j=1}^m\alpha_jP_{(x_j,z_j)}
\end{align}
and $\rho$ a quantum state on $n$ qubits. Given an ensemble of random orthogonal matrices $U$, let $E(\rho)$ be the worst-case cost of computing $\bra{i}\sqrt{\rho}\,U\ket{0}$ up to a constant relative precision, for $\ket{i}$ a computational basis states and $U$ distributed according to $\mu$%
. Moreover, let $x=\operatorname{sign}(\sqrt{\rho}\,U\ket{0})$. Then, we can estimate $2^{-n}\bra{x}%
C\ket{x}$ up to an additive error $\epsilon$ with probability of success $2/3$ in time 
\begin{align}\label{equ:complexity_bound_density2}
\mathcal{O}(\epsilon^{-2}|S|%
E(\rho)\|C\|_{P,\ell_1}^2).
\end{align}
\end{proposition}
To show this result, we use a Monte Carlo algorithm that allows one to sample from a random variable with mean $2^{-n}\bra{x}C\ket{x}$. As in \Cref{thm:pauli_sparsification1}, we see that $\|C\|_{P,\ell_1}$ and the complexity of computing properties of the corresponding Gibbs state govern the complexity of the procedure.
In principle, the same strategy can be used to compute quantities of the form $2^{-n}\bra{x}A\ket{x}$ for any matrix $A$ that is sparse and we can determine the nonzero entries of a row of $A$ in time $\poly(n)$ or, more generally, if we can express $A$ as a matrix product operator~\cite{Bauls2023}.
Furthermore, we can use tensor networks to obtain examples of instances where $E(\rho)=\poly(n)$, such as in the case of $1D$ commuting Gibbs states. As evidence of the effectiveness of this method, we numerically show in \Cref{sec:numerics} that it produces a solution with an approximation ratio of at least $\sim 0.85$ to a non-trivial QUBO instance  with $2^{50}$ variables, see also \Cref{sec:numerics} for details on the experiments. This is achieved by comparing the value obtained for $\GW(C,S,\epsilon)$ with that corresponding to its rounded solution $x$ %
and seeing that they are close to each other. Furthermore, we show that, for that instance, the value given by the spectral bound overshoots considerably. That is, our GW relaxation is improving the performance of the bound considerably and hence %
not solving a ground state problem in disguise.
To the best of our knowledge, no SDP has ever been %
solved at such a scale prior to this work, and we are not aware of other strategies that would yield a feasible point of the QUBO with a similar performance, see Fig.~\ref{fig:sdp_1d_cluster} for details.

\subsection{Final discussion and outlook}\label{sec:previous_Work}
We showed that, by imposing certain structure to the cost matrix $C$ of a QUBO instance of size $D=2^n$, it is possible to solve its standard Goemans-Williamson SDP relaxation in time %
$\poly(n,\epsilon^{-1})$ on either quantum or classical computers. 
The required structure corresponds to cost matrices that are sparse in the Pauli basis and for which their generated diagonal group is small --- i.e., for which the group generated by its Pauli terms contains $\poly(n,\epsilon^{-1})$ diagonal matrices. 
This class of instances enables the --- to the best of our knowledge --- first end-to-end exponential speedups for convex optimization under standard input models by any quantum or quantum-inspired methods. Moreover, the class is rich enough for the corresponding QUBO problems themselves to in principle be computationally hard. In other words, no method is known able to solve general QUBO problems within this class in time less than doubly exponential in $n$, remarkably.

The core of our algorithms is a relaxation to the Goemans-Williamson formulation itself (i.e., an extra relaxation to the standard one from QUBOs) that reduces the resulting SDP to the estimation of local properties of quantum many-body Gibbs states. 
We can then exploit the fact that powerful classical algorithms such as tensor-network methods
or quantum algorithms based on Gibbs-state sampling allow us to efficiently compute these quantities for certain instances. 
In fact, the understanding of quantum Gibbs samplers is currently under intense development~\cite{chen_generators}, and we expect that the advent of  quantum computers will make it possible for them to further enlarge the set of instances for which our methods deliver an exponential speedup.
Furthermore, the fact that our results rely only on finite-temperature quantum Gibbs states may make them better suited to near-term, noisy quantum computers~\cite{bharti2022noisy}.
Besides, even for general instances with small diagonal groups but where local properties of quantum Gibbs states cannot be efficiently estimated, our relaxation still delivers interesting polynomial runtime speedups over the GW relaxation via standard SDP solvers, in both the classical and quantum cases.

As for what concerns QUBO solvers, we introduced a simplified variant of randomized rounding for quantum-inspired algorithms that maps the obtained SDP solutions to approximate solutions of the original, underlying QUBO problem itself. Since the length of such solution bit-strings is exponential in $n$, we do not aim at writing them down. In contrast, the main relevance of our rounding scheme is that it provides a means to efficiently lower-bound the $\QUBO$ optimal value (and an upper bound to it comes from our SDP relaxation). In addition, it also allows one to estimate expected values of arbitrary matrices over the solution bit-string. 
We note that randomized rounding is rarely used in practice, since tailored heuristic QUBO solvers routinely outperform the Goemans-Williamson relaxation plus randomized rounding.
However, we believe that, for the class of instances considered here, %
our approach is favorable over any other known method. More precisely, we are not aware of any other method that can can even get close to the performance of ours for the enormous problem sizes we consider, as mentioned next.

Unlike previous classical implementations~\cite{henze2025solvingquadraticbinaryoptimization}, we show that our quantum-inspired algorithms are practical by numerically solving  instances of huge sizes (up to $D=2^{50}$), which are certainly beyond the capability of any other quantum or classical solver. 
For this, we restrict to a class of $C$'s given by \emph{physically motivated $n$-qubit Hamiltonians}. We can then import various results about regimes in which the corresponding Gibbs states can be simulated in polynomial time in $n$. This way we also bypass the various lower bounds on quantum SDP solvers that demand a runtime 
polynomial in $D$~\cite{vanApeldoorn2020}. However, this comes at the price of requiring structural assumptions beyond just Pauli sparsity and small generated diagonal group. Yet, the resulting class is still not known to be computationally easy, and we provide numerical evidence against that possibility.

Identifying practically relevant instances whose underlying parameters, input and output model fall into the regime where quantum algorithms excel in an end-to-end sense is one of the main challenges of the field of quantum computing~\cite{Dalzell2025-ul,Aaronson2015}. Furthermore, previous work cast serious doubts as to whether quantum solvers based on HU would ever be practical~\cite{henze2025solvingquadraticbinaryoptimization}. 
In this regard, one of the key open question that our work offers is understanding to what extent the structure required to achieve large speedups with our methods is present in real-world applications. The example of Kronecker graphs discussed here gives a promising direction to make our work practically relevant. Further investigating how to optimally sparsify %
a graph %
in terms of Pauli matrices is likely to play an important role in that as well. In addition, developing quantum or quantum-inspired algorithms with a better scaling in the precision $\epsilon^{-1}$ would be essential to improve the practicality of our scheme.

From a more fundamental perspective, in turn, it would be important to either prove computational-complexity hardness results for solving the QUBO problems related to Pauli-sparse matrices or to find out if other solvers can also exploit the underlying structure of Pauli matrices. 
The former would give a stronger theoretical indication that our assumption of this additional structure does not render the instances ``easy", whereas the latter would weaken it.
Alternatively, another interesting question is whether there exist other classes of convex optimizations that admit similar relaxations in terms of quantum many-body problems.
Finally, the fact that we found a novel application outside of Physics for algorithms to simulate finite-temperature quantum Gibbs states is an important conceptual contribution on its own. This has been an area of intense research recently~\cite{chen_generators,Ding2025,gilyen2024quantumgeneralizationsglaubermetropolis,jiang2024quantummetropolissamplingweak}%
. 

In summary, we have identified instances of convex optimization problems that allow for large quantum and quantum-inspired speedups, by mapping such problems to estimating local properties of quantum many-body Gibbs states. Given the importance of convex optimizations and the vast array of techniques to simulate quantum Gibbs states, both on quantum and classical computers, %
we believe that the  bridge between the two areas established by our work %
can open a novel, highly fertile ground for future investigations.

\section{Details of GW Relaxations}\label{sec:further_relaxing}
As discussed before, the main source of our speedups is that we identify instances of $\GW$ for which we can exponentially reduce the number of constraints while only requiring to prepare physically motivated Gibbs states.
From this starting point, we add additional relaxations, which we first motivate with the following simple observation: 

\begin{lemma} For a density matrix $\rho \in \mathbb{C}^{2^{n} \times 2^{n}}$, we have that $\langle i|\rho|i\rangle=2^{-n} \,\, \forall i $
is equivalent to $\tr{Z_{A}\rho}=0 \quad \forall A \subseteq [n] \setminus \{\varnothing\}$.
\end{lemma}
\begin{proof} Note that the Pauli strings comprised only of Pauli $Z$ and Pauli $I$ are linearly independent and there are $2^{n}$ of them. Furthermore, they are all diagonal. As there are also $2^{n}$ diagonal matrices, we conclude that they form a basis of diagonal matrices. As they only have $\pm 1$ entries, it follows that for a matrix $\rho$, $\text{diag}(\rho)=\frac{1}{2^n}$ implies that $\tr{Z_{A}\rho}=0
$ for $A \neq \varnothing$, as $0=\tr{Z_{A}}=\tr{Z_{A}\rho}$.
\end{proof}

This immediately suggests a further relaxation of GW, that was also considered in~\cite{Patti2023} and we already defined in~\Cref{def:relaxed_SDP}, which we repeat here for the reader's convenience.
\begin{definition}[Relaxed Goemans-Williamson]
Given $S \subseteq 2^{[n]}\backslash\{\varnothing\}, \epsilon>0$ and a symmetric $C\in\mathbb{R}^{2^n\times 2^n}$, consider the SDP
\begin{equation}\label{equ:relaxed_SDP}
\begin{aligned}
\max & \tr{\frac{C}{\|C\|}\rho}\\ 
\textrm{subject to } & \forall A \in S: |\tr{Z_{A}\rho}| \leq \epsilon, \\
& \tr{\rho}=1, \rho \geq 0.
\end{aligned}
\end{equation}
We define $\GW(C, S, \epsilon)$ as the value of the program above in~\Cref{equ:relaxed_SDP}.    
\end{definition}

Note that $\GW(C, 2^{[n]}\backslash\{\varnothing\}, 0)$ corresponds to the value of the original GW problem. We then have the following:
\begin{lemma}
For $S_{1} \subseteq S_{2}$, we have

\[
\GW(C, S_{1}, \epsilon) \geq \GW(C, S_{2}, \epsilon).
\]

Similarly, for $\epsilon_{1}\leq \epsilon_{2}$, we have

\[
\GW(C, S, \epsilon_{1}) \leq \GW(C, S, \epsilon_{2}).
\]
In particular, we have for all $S\subset2^{[n]}$:
\begin{align}\label{equ:interpolation_inequality_1}
    \frac{\lambda_{\max}(C)}{\|C\|}=\GW(C, \varnothing, 0)\geq \GW(C, S, 0)\geq  \GW(C, 2^{[n]}\backslash\{\varnothing\}, 0)=\GW(C)\geq \frac{1}{2^{n}\|C\|}\QUBO(C)
\end{align}
\end{lemma}
\begin{proof} 
The claims about the monotonicity are clear. To see that $ \frac{\lambda_{\max}(C)}{\|C\|}=\GW(C, \varnothing, 0)$, note that if we do not impose any constraints on the diagonal, the SDP in~\Cref{equ:relaxed_SDP} corresponds to the variational characterization of the largest eigenvalue of $C$, up to the normalization with $\|C\|$. 
\end{proof}

\Cref{equ:interpolation_inequality_1} shows that this relaxation can be seen as an interpolation between the spectral and SDP relaxation on the value of the QUBO problem. However, our main motivation to consider this further relaxation, besides reducing the number of constraints, is that, depending on our choice of $S$ and the structure of $C$, solving the SDP is related to preparing quantum many-body Gibbs states and evaluating local expectation values on them:

\begin{lemma}\label{lemma:structure_SDP}
For $C\in\R^{2^n\times 2^n}$, $S\subseteq2^{[n]}$, and $\lambda=(\lambda_{A}, \lambda_{C})\in\mathbb{R}^{1+|S|}$, define the $n$ qubit Hamiltonian
\begin{align}
H(\lambda,C,S)=-\lambda_{C} \frac{C}{\|C\|} - \sum_{A \in S} \lambda_{A} Z_{A}.
\end{align}
Given $\epsilon>0$, there exists a solution $\rho$ to $\GW(C, S, \epsilon)$ that is of the form $\rho=\sigma_{H(\lambda,C,S)}$ with $|\lambda_{C}|+\sum_{A\in S}|\lambda_{A}| = 16n\epsilon^{-1}$. Furthermore, we can explicitly find the parameters $\lambda_A,\lambda_C$ given the ability to compute, up to precision $\epsilon/4>0$, the expectation value of $Z_A$ for $A\in S$ and $C$ on states of the form $\sigma_{H(\lambda',C,S)}$ with $\|\lambda\|_{\ell_1}\leq 16n\epsilon^{-1}$.
\end{lemma}
\begin{proof} This follows from Jaynes' principle and the MMW algorithm. See~\cite{GSLBrandao2022fasterquantum} or~\Cref{algorithm:HamiltonianUpdate} for more details.
\end{proof}
\begin{remark}
Although the total sum of the $\lambda$ are of order $n\epsilon^{-1}$, this typically translates to a Gibbs state of inverse temperature $\mathcal{O}(\epsilon^{-1})$ because of the chosen normalizations. To see this, note that if $C$ corresponds to a many-body Hamiltonian on a lattice, $\|C\|$ is typically of order $n$, so the term $\lambda_C \frac{C}{\|C\|}$ is of order $\epsilon^{-1}$. However, we cannot rule out the possibility that one of the $\lambda_A$, which typically correspond to local Pauli Zs, is not of order $n\epsilon^{-1}$.
\end{remark}
These observations immediately lead to the following questions:

\begin{enumerate}
    \item Given $C$, how to pick $S$ in a judicious manner to ensure a good approximation (i.e. such that $\GW(C,S,0)\simeq \GW(C,2^{[n]}\backslash\{\varnothing\},0)$ in some sense) and such that relevant expectation values on the states $\sigma_{H(\lambda,C,S)}$ can be computed efficiently?
  \item Given such an $S$, how does $\GW(C, S,\epsilon)$ compare to $\GW(C, S,0)$, and, thus, to $\GW(C)$. That is, with which precision do we need to solve the relaxed problem?
\end{enumerate}

We discuss both rigorous and heuristic approaches to these questions.

\subsection{How to relax the constraints of the SDP}\label{sec:how_to_relax}
Let us discuss how to pick the set of constraints $S$ in order to ensure a good approximation.
We start with the case where we can find a subset of constraints that has the same value as imposing all constraints, i.e. $\GW(C,S,\epsilon)=\GW(C,2^{[n]}\backslash\{\varnothing\},\epsilon)$ for some $|S|\ll2^n$. Our analysis is based on how the iterations of HU change the underlying guess state and converge to a solution. We are now ready to prove:
\begin{proposition}\label{prop:how_to_constrain}
Given $C\in\mathbb{R}^{2^n\times 2^n}$ with Pauli expansion
\begin{align}
C=\sum_{i=1}^Mc_{(x_i,z_i)}\,P_{(x_i,z_i)},
\end{align}
let $\cD_C$ be its diagonal group (without identity) as in~\Cref{def:diagonal_group}. Then:
\begin{equation}
\GW(C,\dot\cD_C,\epsilon)=\GW(C,2^{[n]}\backslash\{\varnothing\},\epsilon).
\end{equation}
\end{proposition}
\begin{proof}
We use the HU framework to show this claim. Let us start by outlining our proof strategy.
The proof is based on the observation that for all iterations of HU, the Hamiltonian $H_t$ satisfies $H_t\in\mathcal{A}_C$. From this, we will conclude that $\sigma_{t}\in\mathcal{A}_C$. It then follows that for any Pauli-$Z$ string $Z_{A}$ for some subset $A$ s.t. $Z_A\not\in\dot\cD_C$ we have that $\tr{Z_A\sigma_t}=0$. We then show that this implies that $H_{t+1}\in\mathcal{A}_C$, which gives the claim by mathematical induction. 

Clearly, we have that $H_0=0\in\cA_C$. Note that for any $A\in\cA_C$, $\exp(A)\in\cA_C$. Now assume that $H_t\in\cA_C$ for all $t\leq m$ and, thus, also $\sigma_t$. There are three possibilities now: either the matrix $C$ violates the constraints, or we have a diagonal violation of the constraint, or we converged. Clearly, if we add $C$ to the Hamiltonian we stay in the algebra. If we have converged there is nothing to show either. So the only case we still need to analyze is when we have to do a diagonal update. This means we now need to find a subset $A$ s.t. $\tr{Z_A(\sigma_t-I/2^n)}\geq\epsilon$. Note that as $\sigma_t\in\cA_C$, for all $Z_B\not\in \dot\cD_C$, we have that:
\begin{align}
    \tr{Z_B(\sigma_t-I/2^n)}=0.
\end{align}
To see this, note that we can expand $\sigma_t$ in terms of the Paulis in $\cA_C$. As $Z_B\not\in \dot\cD_C$, it follows that $\tr{Z_BZ_{A}}=0$ for all $Z_A\in \dot\cD_C$. We conclude that $H_{t+1}\in\cA_C$. It follows by induction that for all $t$ we have that $H_{t}\in\cA_C$, and in particular the solution to the SDP.
\end{proof}
Thus, we see that the group generated by the Paulis in the expansion of the matrix completely determines the set of constraints we need to impose. But, as we will see later, this does not mean that the algebra $\mathcal{A}_C$ itself is exponentially small and one can obtain our results by a dimension reduction argument.
Nevertheless, for Pauli-sparse matrices we can find a set of generators for the diagonal group efficiently. As explained in~\Cref{sec:notations}, the algebra generated by $\{(x_i,z_i)\}_{i=1}^M$ is just the linear span of these vectors in $\mathbb{Z}_2^n\times \mathbb{Z}_2^n$ and the Pauli $Z$ correspond to the subspace $\textrm{span}\{(0,z_1),(0,z_2),\ldots,(0,z_n)\}$.  So all we need to do is determine the intersection between two linear spaces.
This gives:
\begin{proposition}\label{prop:checking_constraints}
In the notations of~\Cref{prop:how_to_constrain}, it is possible to find a set of generators of $\dot\cD_C$ in time $\textrm{poly}(M,n)$.
\end{proposition}
\begin{proof}
Let $K_1\in\mathbb{Z}_2^{2n\times m}$ be a matrix whose columns are the vectors $(x_i,z_i)$ and $K_2\in\mathbb{Z}_2^{2n\times n}$ be the matrix with columns $\{(0,e_1),(0,e_2),\ldots,(0,e_n)\}$, where $e_i$ are elements of the canonical basis. It is clear that a string $w$ is in $\dot\cD_C$ if there are vectors $v^1,v^2$ s.t. $w=K_1v^1=K_2v^2$. Indeed, the condition that there is a $v^1$ s.t. $K_1v^1=w$ follows from the fact it is in the algebra generated by the $(x_i,z_i)$ and similarly for $v^2$.
Thus, if we have a basis for the kernel of the matrix $[K_1,K_2]$, $((v^1_1,v^2_1),\ldots,(v^1_k,v^2_k))$, then $\{K_2v^2_1,\ldots,K_2v^2_k\}$ is a basis for $\dot\cD_C$. As we can find a (potentially overcomplete) basis for the kernel of $[K_1,K_2]$ in time $\textrm{poly}(M,n)$, the claim follows. We can then find a minimal basis of the span of $\{K_2v^2_1,\ldots,K_2v^2_k\}$ by Gaussian elimination. 
\end{proof}
Taken together,~\Cref{prop:how_to_constrain} and~\Cref{prop:checking_constraints} allow us to efficiently identify instances of Pauli-sparse matrices (as for these $M$ is polynomial) for which we need to check exponentially fewer constraints for the SDP relaxation without changing its value. In~\Cref{prop:stability_sdp} we show the stability of the solution in terms of $\epsilon$, as we can only solve the problems to finite precision.

Let us illustrate this point with a few examples. We start with the extreme case where $\mathcal{D}_C$ is trivial:
\begin{cor}\label{cor:diagonal_constraints}
In the notations of~\Cref{prop:how_to_constrain}, if $\mathcal{D}_C=\{I\}$, then:
\begin{align}
    \uGW(C)=2^n\lambda_{\max}(C).
\end{align}
Furthermore, if $\lambda_{\max}$ has multiplicity $1$, then we corresponding eigenvector $v_{\max}$ normalized to norm $\sqrt{2^n}$ is the solution to $\QUBO(C)$.
\end{cor}
\begin{proof}
The first statement follows from combining~\Cref{equ:interpolation_inequality_1} with~\Cref{prop:how_to_constrain}, as GW with $S$ as the empty set is the same as $\GW(C,2^{[n]}\backslash\{\varnothing\},0)$ by~\Cref{prop:how_to_constrain}. However, let us show a second proof of this fact that also yields the second statement.
As $\mathcal{D}_C=\{I\}$, we obtain that $\sigma_\beta=\tfrac{e^{\beta H}}{Z_\beta}$ satisfies $\tr{Z_A\sigma_\beta}=0$ for all $A\not=\varnothing$. Thus, it is a feasible point of $\GW(C,2^{[n]}\backslash\{\varnothing\},0)$ for all $\beta$. As for $\beta\to\infty$ we have that $\sigma_\beta$ converges to the normalized projector onto the minimal eigenvalue, this shows that $\GW(C,2^{[n]}\backslash\{\varnothing\},0)\geq\lambda_{\max}(C)$. The other direction of the inequality follows from~\Cref{equ:interpolation_inequality_1}. Finally, if the eigenvector is unique, then we have that it is a feasible point of the $\QUBO$.
\end{proof}

The problem above identifies some instances for which solving the SDP relaxation boils down to computing the largest eigenvalue of $C$ and for which extracting the solution can be done in polynomial time in the dimension of $C$, which is $2^n$. 
Recall that for general $C$, we expect that it takes time $\textrm{exp}(c2^n)$ to solve $\QUBO(C)$. Thus,~\Cref{cor:diagonal_constraints} already identifies instances that are much easier to solve and we can check them efficiently. In~\Cref{prop:comb_char} we give a combinatorial characterization of the graphs for which a result like that of~\Cref{cor:diagonal_constraints} holds in terms of number of cycles.
Although this is interesting in its own right, we later show in~\Cref{sec:logtime} how for certain low-weight Pauli matrices, by exploiting techniques for quantum many-body systems, it is even possible to solve QUBO in time $\textrm{poly}(n)$, i.e. doubly exponentially faster than one expects for the worst-case instances. To the best of our knowledge, these families of instances were not identified before. 
\begin{remark}
The requirement that the eigenspace has dimension $1$ to ensure that $\operatorname{QUBO}(C)=\textrm{GW}(C)$ is necessary in general. Although we were not able to find examples where $\textrm{GW}(C)=\lambda_{\max}(C)>\operatorname{QUBO}(C)$ with local dimension $2$ (qubits), one can numerically check that taking the cyclic shift $X$ in $\C^{3}$, $X\ket{i}=\ket{{(}i+1{)}\bmod 3}$, we have that, for $C=-(X+X^2)\otimes I-I\otimes (X+X^2)\in\C^{9\times 9}$, the values of QUBO and the spectral bound do differ, although we have $\mathcal{D}_C=\{I\}$. As expected, the reason is that the eigenspace is highly degenerate.
\end{remark}

Before we turn to more heuristic methods to pick the relaxation $S$, let us finally bound how we need to pick the precision parameter $\epsilon$ to ensure a good solution to the original problem.

\begin{proposition}\label{prop:stability_sdp}
Let $C$ be such that $|\mathcal{D}_C|=2^k$. Then, provided we use the HU framework to solve the SDPs, we have that:
\begin{align}\label{prop:stability_sdp_eq}
|\textrm{GW}(C,2^{[n]}\backslash\{\varnothing\},0)-\textrm{GW}(C,\dot\cD_C,\epsilon)|=\mathcal{O}\left((2^k-1)^{1/6}\epsilon^{1/3}\|C\|\right)
\end{align}
\end{proposition}
\begin{proof}
First, note that we have that $\textrm{GW}(C,2^{[n]}\backslash\{\varnothing\},0)=\textrm{GW}(C,\dot\cD_C,0)$ by~\Cref{prop:how_to_constrain}. In~\cite[Theorem 5]{henze2025solvingquadraticbinaryoptimization}, the authors show that if we have a $\rho$ such that:
\begin{align}
\left\|\sum_{i=0}^{2^n} \bra{i}\rho\ket{i}|i\rangle\langle i|-\frac{I}{2^n}\right\|_{tr}\leq \epsilon_1,
\end{align}
then there is a $\rho^*$ such that $\sum_{i=0}^{2^n-1} \bra{i}\rho\ket{i}|i\rangle\langle i|=I/2^n$ and $|\tr{C(\rho-\rho^*)}|=\mathcal{O}\left(\epsilon_1^{\tfrac{1}{3}}\|C\|\right)$. Thus, $\rho^*$ is a feasible point of $\textrm{GW}(C,2^{[n]}\backslash\{\varnothing\},0)$. This result implies that the statement follows if we can show that a solution $\sigma$ to $\textrm{GW}(C,\dot\cD_C,\epsilon)$ satisfies:
\begin{align}
    \left\|\sum_{i=0}^{2^n-1} \bra{i}\sigma\ket{i}|i\rangle\langle i|-\frac{I}{2^n}\right\|_{tr}\leq (2^k-1)^{1/2}\epsilon\,.
\end{align}
Let us show that this is indeed the case. To do so, we first recall the standard identity for any two quantum states $\rho,\sigma$:
\begin{align}\label{equ:comparisonnorm}
\|\rho-\sigma\|_{tr}\leq 2^{n/2}\|\rho-\sigma\|_2. 
\end{align}
As $\rho$ is a solution to $\textrm{GW}(C,\dot\cD_C,\epsilon)$ using HU,  we have that $\rho\in\mathcal{A}_C$. Thus, when letting $M$ be the measurement channel in the computational basis, we get by definition:
\begin{align}\label{equ:2normestimate}
\left\|M(\rho)-\frac{I}{2^n}\right\|_2^2=\sum\limits_{Z_A\in\dot\cD_C}2^{-n}\tr{Z_A\rho}^2\leq 2^{-n}(2^k-1)\epsilon^2.
\end{align}
The statement follows by combining~\Cref{equ:comparisonnorm} and~\Cref{equ:2normestimate}.
\end{proof}
Thus, with~\Cref{prop:stability_sdp}, we see that for the case that ${|}\mathcal{D}_C{|}$ is not too large (i.e. $k\sim \log(n)$), we can solve the relaxed problem to polynomial in $n$ precision and obtain a good solution to the original problem, as the required precision only scales with the number of relaxed constraints, not the dimension of the problem.

However, the scaling in~\Cref{prop:stability_sdp} is still undesirable when it comes to $\epsilon$. Note that~\cite{henze2025solvingquadraticbinaryoptimization} also numerically tests the tightness of their estimate on the continuity of the SDP w.r.t. the relaxation of the diagonal constraint. Although they improved the dependence to $\epsilon^{1/3}$ from $\epsilon^{1/4}$ from the original work of~\cite{GSLBrandao2022fasterquantum}. Although they improved the dependence $\epsilon^{1/4}$ from the original work of~\cite{GSLBrandao2022fasterquantum} to $\epsilon^{1/3}$, their numerics suggest that the correct scaling should be around $\epsilon^{0.79}$.
However, for a given $\cD_C$ it is possible to numerically evaluate a bound that can give a scaling of the error of the form $\mathcal{O}(\epsilon\|C\||\cD_C|)$, which is essentially the best one can hope for in scaling with $\epsilon$:
\begin{proposition}\label{prop:stability_epsilon}
For a $S\subset 2^{[n]}\backslash\{\varnothing\}$, $\epsilon>0$ and $C\in \mathbb{R}^{D\times D}$, let $\Xi(S,\epsilon)$ be the solution to the linear program:
\begin{align}
\sup_{\xi\in\R^{m}} &\|\xi\|_{\ell_1}\\
\textrm{subject to }&\sum_{i} \xi_iZ_{A_i}\geq -\GW(C,S,\epsilon)I
\end{align}
Then for all $\epsilon>0$:
\begin{align}
|\GW(C,S,\epsilon)-\GW(C,S,0)|=\epsilon \,\Xi(S,\epsilon)
\end{align}
\end{proposition}
We prove this statement in~\Cref{app:stability}, where we show that for some choices of $S$ this bound is significantly better than~\Cref{prop:stability_sdp}, as it scales as $ \Xi(S,\epsilon)=\GW(C,S,\epsilon)$, ensuring a relative precision. This will particularly be the case for the instances we perform numerics on, indicating they have greater stability than arbitrary instances.
\subsection{Krylov approach to adding constraints}\label{sec:approximate_constraint_set}
In this section, we  discuss a heuristic method to add constraints inspired by Krylov space approaches~\cite{Nandy2025}. In Krylov space methods for time evolutions, one tries to approximate the time evolution of a state $\ket{\psi}$ by a Hamiltonian $H$ for a time $t$ by truncating the evolution to the subspace spanned by $\{\ket{\psi},H\ket{\psi},H^2\ket{\psi},\ldots,H^k\ket{\psi}\}$. This way one can enforce that the Taylor expansion of $e^{-iHt}\ket{\psi}$ is correct up to order $k$. We follow a similar path in the following sense: recall that an $\epsilon$-optimal point of the SDP $\GW(C)$ can be represented by a Gibbs state of the form 
\begin{align}
\rho_{\lambda_0,C,D}\propto\exp\left(\lambda_0C+D\right)\in\mathbb{R}^{2^n\times 2^n}
\end{align}
for some diagonal matrix $D$ and a real number $\lambda_0$. As we are ultimately interested in optimizing $\tr{C\rho_{\lambda_0,C,D}}$, our goal in picking the set $S$ for $\GW(C,S,\epsilon)$ should be to make sure that the constraints in $S$ influence the value of $\tr{C\rho_{\lambda_0,C,D}}$ as much as possible. Thus, we now show how to pick $S$ to ensure that the Taylor expansion of $\tr{Ce^{\lambda_0C+D}}$ up to order $k$ is, for any choice of $D$, completely determined by the expectation values of $\tr{C^l\rho_{\lambda_0,C,D}}$ and $\tr{Z_{A_i}\rho_{\lambda_0,C,D}}$ for $l\leq k$ and $A_i\in S$. As such, adding additional constraints outside of $S$ only influences the value of the objective function at order $k+1$.
 We start with the following definition:
\begin{definition}[Diagonal constraints of order $k$]
Let $C\in\mathbb{R}^{2^n\times 2^n}$ with Pauli expansion
\begin{align}%
C=\sum_{(x,z)\in\{0,1\}^{2n}}c_{(x,z)}\,P_{(x,z)}
\end{align}
and $\textrm{Psupp}(C)=\{P_{(x,z)}|c_{(x,z)}\not=0\}$ be the Pauli support of $C$.
We define the diagonal constraints of order $k$, $\cD_{C,k}$, for some integer $k\geq2$ to be given by
\begin{align}
\mathcal{Z}_n\cap \{P_{(x_1,z_1)}\cdots P_{(x_l,z_l)}|P_{(x_i,z_i)}\in\textrm{Psupp}(C)\textrm{ and }l\leq k \},
\end{align}
that is, all diagonal Paulis we can form with products of at most $k$ different Paulis in the support of $C$.
\end{definition}
This is a relaxation of our definition of $\cD_{C}$, as $\lim_{k\to\infty}\cD_{C,k}=\cD_{C}$.

Let us now justify this definition as a relaxation of the diagonal group one with the following proposition:
\begin{proposition}
Let $C\in\mathbb{R}^{2^n\times 2^n}$ be traceless and define for some diagonal matrix $D$ and a real nuber $\lambda_0$ the Gibbs state $\rho_{\lambda_0,C,D}\propto\exp\left(\lambda_0C+D\right)\in\mathbb{R}^{2^n\times 2^n}$. Then the Taylor expansion of $\tr{C\rho_{\lambda_0,C,D}}$,
\begin{align}\label{equ:taylor_value}
\tr{C\rho_{\lambda_0,C,D}}=\sum\limits_{l=1}^\infty\frac{\tr{C\left(\lambda_0C+D\right)^l}}{l!Z_{\lambda_0,C,D}}
\end{align}
is determined up to order $k$ by $\tr{C^l\rho_{\lambda_0,C,D}}$ for $l\leq k+1$ and the values of $\tr{Z_{A_i}\rho_{\lambda_0,C,D}}$ for $Z_{A_i}\in\cD_{C,k}$.

\end{proposition}
\begin{proof}
Note that we can express all the terms in the Taylor expansion in Eq.~\eqref{equ:taylor_value} up to order $k+1$ linear combinations of traces of the form:
\begin{align}\label{equ:trace_expansion}
\tr{CC^{c_1}D^{d_1}\cdots C^{d_r}D^{d_r}}
\end{align}
with $r\leq k$ and $\sum_i d_i+c_i\leq k$, as can be seen by just expanding the power $\tr{C\left(\lambda_0C+D\right)^l}$.
Thus, the claim follows if we can show that we can express all such traces in terms of the values of $\tr{C^l\rho_{\lambda_0,C,D}}$ for $l\leq k+1$ and the values of $\tr{Z_{A_i}\rho_{\lambda_0,C,D}}$ for $Z_{A_i}\in\cD_{C,k}$. 
As $D$ is diagonal, we can express it as a linear combination of Pauli $Z$ strings. Now note that the trace in Eq.~\eqref{equ:trace_expansion} can be expressed as a product of $1+\sum_ic_i$ Paulis in $\textit{\textrm{Psupp}}(C)$, the one coming from the first $C$ in the trace, times $\sum_id_i$ Pauli $Z$ strings coming from $D$. In case $\sum_id_i=0$, the value of the trace is clearly determined by a quantity of the form $\tr{C^l\rho_{\lambda_0,C,D}}$. Let us now consider the case $\sum_id_i=0\geq1$, i.e. we have a diagonal component. We again expand the trace as a linear combination of a product of Paulis in the support of $C$ and $D$.
Now note that each one of these traces is zero if this sequence of Paulis coming from $C$ and from $D$ do not multiply to the identity. As the matrix $D$ can only generate diagonal Paulis, a necessary condition for the sequence multiplying to the identity is if there is a $l\leq k$ s.t. there is a sequence of Paulis of length $l$ in $\textit{\textrm{Psupp}}(C)$ for $2\leq l\leq k$ that is in $\mathcal{Z}_n$. But those are exactly the Paulis in $\cD_{C,k}$. Thus, the claim follows.
\end{proof}
At this stage we do not yet have rigorous control of how fast the value of the SDP with constraints $\cD_{C,k}$ converges to the one with constraints $\cD_{C}$, which corresponds to that of $\GW(C)$, leaving this to future work. However, we benchmark this relaxation numerically for the case of Kronecker graphs in~\Cref{sec:numerics} and see that it is highly effective already at $k=2,3$. Furthermore, note that $|\cD_{C,k}|\leq |\textit{\textrm{Psupp}}(C)|^k$, which means that for Pauli sparse it will be of polynomial size for constant $k$. In addition, if $C$ is local, $\cD_{C,k}$ also only contains local $Z$ strings for $k$ constant.
\section{Pauli-sparse graphs}\label{sec:pauli_sparse}
The most important working assumption of this work is that the underlying cost matrix $C$ admits an approximate representation that is sparse in the Pauli basis. Furthermore, it is of course important that we are able to efficiently find such a representation. Because of this, it is natural to ask on which graphs QUBOs admit a Pauli-sparse representation and whether they are of practical interest. In this section, we give a partial classification of graphs that admit an exact Pauli-sparse representation in terms of \emph{local} Pauli matrices, which are certain subgraphs of Hamming neighborhood graphs. In addition, we discuss one example of families of graphs that are known to have practical applications and fit well into our framework: Kronecker graphs~\cite{kronecker}. Indeed, we show that these graphs admit much better Pauli decompositions than general graphs and, importantly, we can also efficiently obtain such representations by the Monte Carlo sampling method.

Before we discuss these examples, it is noteworthy that the Pauli decomposition of the QUBO does not respect the natural symmetry of the problem. This is because QUBOs are invariant under permutations of the computational basis. Paulis, on the other hand, are not, and the $\ell_1$ norm of the expansion can vary exponentially. To see this, consider the graph on $n=3$ vertices with adjacency matrix given by $A=XII$. Conjugating it with the Toffoli gate $T$, which is non-Clifford permutation, yields $TAT^\dagger=\tfrac{1}{2}(XII+XZI+XIX-XZX)$. We see that the $\|\cdot\|_{C,\ell_1}$ norm increased by a factor of $2$. Applying this construction to adjacency matrices of the form $X_1X_4X_7\ldots X_{3n+1}$ and conjugating with $T^{\otimes n}$, we see that the $\|\cdot\|_{C,\ell_1}$ will increase by $2^n$. Thus, isomorphic graphs can have exponentially differing $\|\cdot\|_{C,\ell_1}$ norms. 
Furthermore, in the case, the graph does not have $2^n$ vertices, it is necessary to embed it into a $2^n$ space, and different embeddings can again have quite different $\|\cdot\|_{C,\ell_1}$ norms. Thus, one should in principle also optimize these embeddings and permutations to increase the sparsity of the representation. We leave how to optimize the sparsity of the representation to future work.
\subsection{Partial Classification of Pauli-sparse, low-weight graphs}\label{sec:classification_graphs}
We now discuss some examples of graphs such that certain weighted QUBOs over them can be expressed with Pauli-sparse matrices. The following family of graphs plays a prominent role in our discussion:
\begin{definition}[Neighborhood $k$-binary Hamming graph]\label{def:NeighborhoodkHamming}
Given $n,k\in\mathbb{Z^{+}}$, let the neighborhood $k$-binary Hamming graph $\textrm{Ha}_{n,k}$ be the graph ${Ha}_{n,k}=(\{0,1\}^n,E)$ on bitstrings as nodes with the set of edges given by:
\begin{align}
E=\{(r,s)\in\{0,1\}^n\times \{0,1\}^n: 1\leq d_H(r,s)\leq k\},
\end{align}
where $d_H$ is the Hamming distance on bitstrings. 
\end{definition}
These graphs play an important role in coding and information theory~\cite{aiello1991coding}.

We see that low-weight Pauli cost matrices(\Cref{equ:pauli_sparse}) correspond to weighted adjacency matrices of subgraphs of the Hamming neighborhood graph. Let us consider the following examples to warm up.

\paragraph{(Generalized) Hamming graphs:}
A good example to begin our discussion of what graphs admit a Pauli expansion of low weight is the hypercube. For this example and hereafter we label the vertices of a graph by binary strings $r,s\in\{0,1\}^n$. Furthermore, let $x_i\in\{0,1\}^n$ be the vector that is $1$ on entry $i$ and $0$ else. For instance, we can write the adjacency matrix of the hypercube as:
\begin{align}
\label{eq:hypercube_cost}
    C=\sum_{i=1}^nP_{(x_i,0)}
\end{align}
Indeed, note that for the hypercube $r\sim s$ if and only if the strings $r,s$ are Hamming distance $1$. But this implies that there exists a Pauli $P_{(x_i,0)}$ s.t. $r=P_{(x_i,0)}s$ and for all other Pauli $P_{(x_j,0)}$ we have $r\not=P_{(x_j,0)}s$. Thus, we see that $\bra{s}C\ket{r}=1$ iff $s$ and $r$ are Hamming distance $1$.

More generally, let us define for $A\subset[n]$ the vector $x_A$ to be nonzero on entries in $A$. A similar argument shows that a term of the form $P_{(x_A,0)}$ for a subset $A\subset[n]$ corresponds to adding edges between all vertices that can be obtained from each other based on flipping all bits in $A$. In particular, if we consider the cost matrix:
\begin{align}
    C_k=\sum_{A\subset[n],|A|=k}P_{(x_A,0)}
\end{align}
we obtain the adjacency matrix of the Hamming graph of distance $k$. Recall that the Hamming graph of distance $k$ corresponds to the graph where two vertices are connected if they are at a Hamming distance $k$ to each other. As there are ${n \choose k}$ such Pauli $X$ matrices, we conclude that these graphs are Pauli-sparse as long as $k=\mathcal{O}(1)$. In addition, it is not difficult to see that Paulis of the form $P_{(x_A,z_B)}$ can only add or remove edges for vertices that differ on $A$.

Thus, we can summarize the discussions of this section with the following proposition:

\begin{proposition}
Let $C$ be the cost matrix of a Pauli-sparse matrix with weight at most $k$. Then $C$ corresponds to a QUBO on a subgraph of the Hamming neighborhood graph of distance $k$.
\end{proposition}
\begin{proof}
See the discussion above.
\end{proof}
Note that the reciprocal is not true, i.e. there are subgraphs of Hamming neighborhood graphs that do not admit a low-weight Pauli representation. Let us consider the graph with adjacency matrix:
\begin{align}
C=P_{(x_1,(0,1,\ldots1))}+\sum_{i=1}^nP_{(x_i,0)}.
\end{align}
The Pauli $P_{(x_1,(0,1,\ldots1))}$ acts on $n$ qubits, however it is easy to see that the graph defined this way is a subgraph of the Hamming cube.

\begin{remark}
If we were to consider the more general case of local dimension $d$ instead of $2$, we would obtain that all sparse, local graphs correspond to subgraphs of the Hamming neighborhood graph with the local alphabet having $d$ distinct symbols.
\end{remark}

\begin{remark}
As previously discussed, our main motivation to consider sparse, low-weight Paulis is the fact that solving the corresponding SDP has a strong connection to the simulation of quantum many-body Gibbs states. However, it is not only low-weight Pauli Hamiltonians that admit a correspondence to physical models. For instance, certain fermionic many-body problems are known to be mapped to qubit Hamiltonians with Pauli strings of weight of order $n$. We leave it to future work to consider which graphs correspond to these Hamiltonians. Thus, we expect that our methods can be applied beyond the subgraphs of the Hamming neighborhood graphs for $k$ of constant order. 
\end{remark}

In this work, we consider implementing the algorithm to solve the SDP both on classical and quantum devices. For certain classical solvers based on our methods, such as those that only exploit the efficient representation and sparsity of the Pauli strings, it is not essential that the Pauli strings in the expansion of the adjacency matrix have constant locality. This is because also tensor products of Paulis of high weight are sparse. However, the graphs we can express this way have a slightly more complicated structure than the local ones, although some of the previous discussion carries over.
Indeed, as we saw above, we can only add or remove edges between strings by adding Pauli $X$ or Pauli $Y$ terms to the Hamiltonian. Under the assumption of the graph being Pauli-sparse, by definition it can be expressed as the linear combination of polynomially Pauli strings. As such, we conclude that:
\begin{proposition}
Let $C$ be a cost matrix that is Pauli-sparse. Then there exist polynomially many subsets $A_1,\ldots,A_m\subset[n]$ such that $C$ corresponds to a QUBO on a subgraph of the graph given by:
\begin{align}
E=\{(r,s)\in\{0,1\}^n\times \{0,1\}^n: \exists A_i
\textrm{ s.t. }P_{(x_{A_i},0)}r=s\},
\end{align}
i.e. all the edges corresponds to strings $r,s$ such that we can go from $r$ to $s$ by flipping all the bits in one of the subsets $A_i$.
\end{proposition}
These characterization results raise multiple open questions. First, it would be interesting to develop more refined criteria to check if a matrix $C$ is Pauli-sparse or develop algorithms to test such decompositions efficiently. Our result in~\Cref{thm:pauli_sparsification} already shows that $\|C\|_{P,\ell_1}$ gives a way of quantifying approximate Pauli sparness.
Secondly, given that we will later show that it is possible to obtain better approximations to the SDP relaxation of the QUBOs corresponding to such matrices, it is natural to ask if solving the QUBO problems for local or Pauli-sparse matrices still gives rise to NP-complete instances of combinatorial optimization problems.

We can also give a combinatorial characterization of adjacency matrices $C$ s.t. $\cD_{C}=\{I\}$, and, thus, $\GW(C)=\lambda_{\max}(C)$:
\begin{proposition}[Combinatorial characterization of trivial algebra]\label{prop:comb_char}
Let $C\in\mathbb{R}^{2^n\times 2^n}$ be a cost matrix corresponding to the adjacency matrix of a graph $G=(\{0,1\}^n,E)$, (i.e. $C$ has binary entries). Furthermore, for a vertex $i$ and a natural number $k$, let $\operatorname{cy}(i,k)$ be the number of cycles including $i$ of length $k$ in $G$. 
Then, if $\forall i\in\{0,1\}^n$ we have that $\operatorname{cy}(i,k)=c_k$ for some constant independent of $i$, i.e. the number of cycles of length $k$ containing a given vertex is a constant independent of the given vertex, then:
\begin{align}
    \GW(C)=\lambda_{\max}(C).
\end{align}
\end{proposition}
\begin{proof}
It is not difficult to see that $\operatorname{cy}(i,k)=\bra{i}C^k\ket{i}$. Thus, we see that our assumption that $\operatorname{cy}(i,k)=c_k$ is equivalent to $\operatorname{diag}(C^k)=c_kI$. We can then argue as in~\Cref{cor:diagonal_constraints}.
\end{proof}
\begin{remark}
Graphs that have the property that all cycles have the same length are called walk-regular graphs. The same result as Prop.~\ref{prop:comb_char} was given in~\cite{van_Dam_2016}, however we note that, for Pauli sparse, we can efficiently verify if they are walk-regular.
\end{remark}
\subsection{Kronecker Graphs}\label{sec:kronecker_graphs}
We now consider Kronecker graphs~\cite{kronecker} that fit our framework well and have practical applications. These graphs are widely used to model real-world networks, including social, biological, and chemical networks.~\cite{reiser2022graph, kronecker}. 
\begin{definition}[Kronecker graph]\label{def:kronecker_graph}
Given symmetric matrices $A_1,\ldots,A_m\in\mathbb{R}^{2^k\times 2^k}$, let $C=\Kron(A_1,\ldots,A_m)$ be the adjacency matrix of the weighted graph with $2^{km}$ vertices:
\begin{align}
C=\bigotimes_{i=1}^m A_i.
\end{align}
\end{definition}
The extension of the~\Cref{def:kronecker_graph} to the case where $k$ depends on $i$ is straightforward. One can note that
\begin{eqnarray}
\QUBO(A_1)\QUBO(A_2) &&= \langle x^*_1, A_1 x_1^*\rangle\langle x^*_2, A_2 x_2^*\rangle\nonumber\\ &&= \langle x^*_1\otimes x^*_2, \left[A_1\otimes A_2\right] x^*_1 \otimes x_2^*\rangle\nonumber\\ &&\leq \QUBO(A_1\otimes A_2),
\end{eqnarray}
where one can give examples of $A_1$ and $A_2$ leading to strict inequality, implying that local solutions for components $A_i$ do not necessarily yield the global solution for $C$.

The following proposition quantifies the sparsity of the Pauli expansion of a Kronecker graph:
\begin{proposition}\label{prop:kronecker_sample}
Given symmetric matrices $A_1,\ldots,A_m\in\mathbb{R}^{2^k\times 2^k}$ with $\|A_i\|=1$, let $C=\Kron(A_1,\ldots,A_m)$. Then:
\begin{align}
\|C\|_{P,\ell_1}=\prod\limits_{i=1}^m\|A_i\|_{P,\ell_1}.
\end{align}
Furthermore, we can sample from the probability distribution $p(x,z)=\frac{\left|\tr{P_{(x,z)}C}\right|}{\|C\|_{P,\ell_1}}$ in time $\mathcal{O}(2^km)$.
\end{proposition}
\begin{proof}
The statement easily follows from the tensor product decomposition of the matrix $C$, which implies that $p(x,z)$ also has a tensor product form.
\end{proof}
Using \Cref{thm:pauli_sparsification} one concludes that for Kronecker graphs one can efficiently find a Pauli expansion that has $\sim \|C\|_{P,\ell_1}^2$ terms and, as long as $\|A_i\|_{P,\ell_1}\ll 4^k$, this expansion is sparse.
As such, we believe that such graphs are a natural playground to obtain practically relevant speedups using our methods. We discuss numerical examples in detail in \Cref{sec:numerics}.
\section{Classical and quantum speedups for SDP relaxations of QUBO}\label{sec:logtime}
Now that we have established results on how to pick the relaxed constraints and the continuity of the SDP, let us build upon them to analyze families of instances for which our methods give speed-ups to solve the underlying SDPs compared to previously known algorithms. The runtimes of previously known algorithms are summarized in~\Cref{sec:SDP_solvers}. Here, we discuss both classical and quantum speedups. The general philosophy to identify instances where we expect large speedups is as follows. We first restrict to cost matrices $C$ and sets of constraints $S$ for which we know that:
\begin{enumerate}
    \item We can efficiently compute (up to additive precision $\sim\epsilon$) the expectation value of $C$ and Pauli strings in $S$ on Hamiltonians that are linear combinations of $C$ and strings in $S$ and have total operator norm upper-bounded by $\|\lambda\|_{\ell_1}=\mathcal{O}(n\epsilon^{-1})$%
    . 
    \item The set $S$ is such that solving it with these constraints gives a good approximation to the original problem.
\end{enumerate}
We already covered the second point in~\Cref{sec:how_to_relax}. Let us now discuss some examples for the first. 
To that end, we need to discuss the complexity of preparing certain classes of quantum Gibbs states. As noted before, all the Gibbs states required by our algorithms are of the form 
\begin{align}\label{equ:Gibbs_parametrized2}
\sigma(\lambda)\propto \exp\left(\lambda_0\tfrac{C}{\|C\|}+\sum_{Z_A\in S}\lambda_iZ_A\right),
\end{align} 
with $\|\lambda\|_{\ell_1}=\mathcal{O}(n\epsilon^{-1})$. 
Note that if $C$ corresponds to a quantum many-body Hamiltonian on a lattice, we usually have $\|C\|\sim n$ and coefficients whose magnitudes are independent on $n$. Hence, even if $\lambda_0=\mathcal{O}(n\epsilon^{-1})$, the largest coefficient in the Pauli expansion of $\lambda_0\tfrac{C}{\|C\|}$ can have absolute value at most $\mathcal{O}(\epsilon^{-1})$. On the other hand, the magnitude of the coefficients associated to the $Z_A$ terms could, in principle, be of order $\mathcal{O}(n\epsilon^{-1})$, as the algorithm could perform a constant fraction of the updates in the $Z_A$ sector %
(although we do not observe this case in the numerics). Thus, we need to resort to results that ensure efficient quantum Gibbs-state preparation even in the regime where the Hamiltonian $H$ satisfies $\|H\|=\cO(n\epsilon^{-1})$ but some local terms might have coefficients that are not of constant order. To the best of our knowledge, this particular setup escapes the results available in the literature.

Another key concepts to understand the complexity of preparing a set of quantum many-body Gibbs states is the locality of the underlying Hamiltonian. In our case, we sill say that the family $\sigma(\lambda)$ is $k$-local w.r.t. a graph $G=(V,E)$ on $n$ vertices if all Paulis in the Pauli expansion of $C$ and the $Z_A$'s are supported on subsets that are contained in balls of radius $k$ around vertices of $G$. Let us now discuss some examples.
\subsection{Classical and Quantum exponential speedups}
\paragraph*{Exponential and superpolynomial classical speedups through commuting relaxations on trees:} 
It is well-known that if a Hamiltonian $H$ is $k$-local w.r.t. to a tree and all the terms of $H$ commute, then it is possible to express the partition function as a tensor network that can be contracted efficiently on classical computers~\cite[Chapter 1]{Bridgeman_2017}, independently of the temperature. This allows us to conclude that:
\begin{thm}\label{thm:efficient}
Let $C\in\mathbb{R}^{2^n\times 2^n}$ be such that $\mathcal{A}_C$ is commuting and both $C$ and $\mathcal{D}_C$ can be expressed as local terms on a tree with $|\mathcal{D}_C|=%
\poly(n)$. Then we can compute $\GW(C)=\GW(C,2^{[n]}\backslash\{\varnothing\},0)$ up to additive precision $\epsilon>0$ in time $\textrm{poly}(n,\epsilon^{-1})$ on a classical computer.
\end{thm}
\begin{proof}
It follows from~\Cref{prop:stability_sdp} that, by solving $\GW(C,\dot\cD_C,\epsilon)$ up to a precision $\epsilon>0$, we can approximate 
$\GW(C,2^{[n]}\backslash\{\varnothing\},0)$ up to an error $|\mathcal{D}_C|\,\epsilon^{1/3}$. Thus, it suffices to argue that we can obtain $\GW(C,\dot\cD_C,\epsilon)$ in time $\textrm{poly}(n,\epsilon^{-1})$. 
Now, note that, as $\mathcal{D}_C\subset\mathcal{A}_C$ and $\mathcal{A}_C$ is commuting, it is clear that all the elements of $\mathcal{D}_C$ commute with the elements of the expansion of $C$ into Pauli strings. 
Combining these observations with our assumption that all the elements of $C$ and $\mathcal{D}_C$ are on a tree, we see that all Gibbs states required by the HU algorithm will be local, commuting Gibbs states on trees. As such, we can compute local expectation values of such states at arbitrary temperatures in polynomial time. As the algorithm needs $\mathcal{O}(n\epsilon^{-1})$ iterations to converge to a solution, we see that we can approximate a solution up to $\epsilon^{1/3}|\mathcal{D}_C|$ in time $\textrm{poly}(n,\epsilon^{-1})$. Finally, as $|\mathcal{D}_C|$ is of polynomial size, the claim follows from an appropriate rescaling of $\epsilon$ by a constant factor.
\end{proof}
In~\Cref{sec:numerics} we give explicit examples of such Hamiltonians and run the algorithm described above.

We can easily generalize the statement above to arbitrary commuting Hamiltonians. It is known that the complexity of contracting a tensor network is exponential in the tree-width of its graph~\cite{contractiontn,Markov2008} and polynomial in the system size. Thus, we could also formulate the above proposition by fine-tuning the complexity to scale exponentially with the tree-width of the underlying hypergraph of the terms in $C$ and $\mathcal{A}_C$. For instance, for commuting terms on a square grid this would lead to algorithms with a scaling that is $2^{\sqrt{n}}$, a superpolynomial speedup compared to other SDP solvers. Moreover, in principle, we do not need to restrict to models where the Hamiltonian is commuting to expect good results with tensor networks, as these also perform well on Gibbs states for certain noncommuting models~\cite{Bauls2023}. However, even for $1D$ systems, existing rigorous algorithms do not allow us to provably prepare Gibbs states at the inverse temperatures required by our algorithm.

In turn,~\Cref{thm:efficient} uses tensor networks to compute the local expectation values of the Gibbs states. However, there exist many alternative techniques to perform this task, e.g. quantum Monte Carlo methods~\cite{Sandvik_2010}, and it would be interesting to explore which are efficient for the problem at hand, again potentially even going beyond commuting models.

Finally, given the large speedups our methods offer to compute these SDP approximations of the QUBOs, it is natural to inquire what is the computational complexity of solving them exactly. At the present stage, we only have numerical evidence that instances corresponding to commuting Hamiltonians might take time that is superpolynomial in their dimension, as seen in \Cref{sec:numerics}. In addition, in \Cref{sec:graphs_commuting_Hamiltonians}, we investigate some structural properties related to graphs that admit such a commuting decomposition and show that, for some instances, the spectral bound $\lambda_{\max}(C)$ is strictly worse than $\GW(C,\dot\cD_C,0)$. Thus, we currently do not have any evidence that QUBOs related to local, commuting Hamiltonians can be solved efficiently in their dimension in general, even for $1D$ instances.

As previously noted, the value of GW allows us to bound the value of the QUBO problem, but it is worth noting that we can also efficiently check if
the value of GW and QUBO actually coincide:
\begin{cor}\label{cor:spectrum_graph}
Under the conditions of~\Cref{thm:efficient}, assume further that $\mathcal{D}_C=\{I\}$, so that $\lambda\in\mathbb{R}$, and 
that
\begin{align}\label{equ:purity_condition}
\tr{\sigma_{\lambda}
^2}\geq 1/4+\delta
\end{align}
for some $\delta>0$ and $\lambda\geq n\epsilon^{-1}$. Then we can compute $\QUBO(C)$ up to a precision $2^n\epsilon\|C\|>0$ in time $\textrm{poly}(n,\epsilon^{-1})$. Furthermore, we can check the estimate in~\Cref{equ:purity_condition} in time $\textrm{poly}(n,\epsilon^{-1})$.
\end{cor}
\begin{proof}
As discussed in~\Cref{cor:diagonal_constraints}, if $\mathcal{D}_C$ is trivial and the adjacency matrix $C$ has a unique ground state, then the value of the QUBO and GW coincide. Thus, the statement would follow if we could show that, under~\Cref{equ:purity_condition}, the ground state is unique. Also note that the rescaling of $2^n$ of the precision comes from the fact that we are actually solving the QUBO with normalized sets of vectors. Indeed, it is not difficult to see that the function $\lambda\mapsto\tr{\sigma_{\lambda}^2}$ is monotone increasing in terms of $\lambda$. Furthermore, if the ground space of $C$ is degenerate, this means that $\lim_{\lambda\to\infty}\tr{\sigma_{\lambda}^2}\leq\tfrac{1}{4}$. Thus, if for any finite $\lambda$ the purity goes above $1/4$, we know that the ground state is unique. The statement about efficiency follows from the fact that the purity is nothing but a ratio of partition functions at different temperatures. As such, it can be computed efficiently, as we already discussed that for these classes efficient computation is possible.
\end{proof}

\begin{remark}
Note that~\Cref{cor:spectrum_graph} only concerns approximating the optimal value of the QUBO and does not extract the solution, which would correspond to the associated eigenvector. However, we note that, depending on the spectral gap of $C$, $\sigma_{\lambda}$ is also close to the corresponding projector onto the solution. In that case, it is then possible to also efficiently extract the solution. In~\Cref{sec:numerics} we investigate such instances and show how to efficiently extract good solutions from the corresponding Gibbs state.
\end{remark}

Importantly, it should also be stressed that for the instances of~\Cref{cor:spectrum_graph} we can compute the solution of the QUBO problem in time polynomial in the dimension by just diagonalizing the matrix $C$ and, thus, these are not hard instances of QUBOs. However, our result shows that for these special instances it is possible to extract the value in time that is polylogarithmic in the dimension up to a relative error $\epsilon\, 2^n$, making it \emph{doubly-exponentially} easier to solve than generic instances under standard complexity theoretic assumptions.

Finally, regarding~\Cref{thm:efficient}, it is unclear to what extent such instances correspond to
hard QUBOs. In other words, the instances covered by the theorem are indeed in principle rich enough to be NP-complete. In addition, as we show in the~\Cref{sec:numerics}, they contain instances for which the GW relaxation outperforms the spectral one. 
Even more pressing is the question of to what extent there are real-world instances that have a structure amenable to~\Cref{thm:efficient}, but this is beyond the scope of this work and will be explored elsewhere.

\paragraph*{Exponential quantum speedups through relaxations on $1D$ models:} 
It is believed that quantum computers should be able to efficiently prepare quantum Gibbs states of $1D$ models with a scaling that is $\poly(n,e^{\beta})$,
where the inverse temperature $\beta$ is often defined as the maximal coefficient in the Pauli expansion of the Hamiltonian. However, even such a result would need to be adapted for our purposes, to get a rigorous bound: Since %
we only have the promise that $\|\lambda\|_{\ell_1}=\cO(n\epsilon^{-1})$, %
terms in $\cD_C$ could have coefficients of order $n\epsilon^{-1}$ (see discussion after~\Cref{equ:Gibbs_parametrized2}), translating into $\beta=n\epsilon^{-1}$. 
Nevertheless, to the best of our knowledge, the best available rigorous result~\cite{Bilgin2010} gives a quantum algorithm with runtime $n^{\mathcal{O}(\beta^{-1})}$ to prepare a quantum Gibbs state of a $1D$-system. In fact%
, a careful inspection of the algorithm in~\cite{Bilgin2010} shows~\footnote{to see this, note that we can prepare the reduced density matrix on the constant number of sites where we have large terms and then glue the other sites together as described in the algorithm of Boixo and Bilgin.} that the algorithm exhibits the same scaling if only constantly many terms in the Hamiltonian are above $\beta$. 
We then have:
\begin{thm}\label{thm:efficient_quantum}
Let $C\in\mathbb{R}^{2^n\times 2^n}$ be such that $\mathcal{A}_C$ can be expressed as local terms on a $1D$ chain and $|\mathcal{D}_C|=\mathcal{O}(1)$ . Then we can compute $\GW(C,2^{[n]}\backslash\{\varnothing\},0)$ up to precision $\epsilon>0$ in time $\tilde{\mathcal{O}}(n^{\mathcal{O}(\epsilon^{-1})}\epsilon^{-2})$ on a quantum computer.
\end{thm}
\begin{proof}
We resort to the algorithm of~\cite{Bilgin2010} to prepare the underlying Gibbs states by noting that $|\mathcal{D}_C|=\mathcal{O}(1)$ implies that there are constantly many terms in the Hamiltonian that have coefficient larger than $\epsilon^{-1}$, giving a runtime of $\tilde{\mathcal{O}}(n^{\mathcal{O}(\epsilon^{-1})})$ to prepare the Gibbs states. We then proceed as described in \Cref{thm:efficient} for the rest of the proof
\end{proof}
For $\epsilon=\Omega(1)$, we obtain an exponential quantum speedup over the best-known methods for such instances. As mentioned before, we believe that recently introduced dissipative preparation schemes for quantum Gibbs states~\cite{chen_generators} will allow for the preparation of even more classes of quantum Gibbs states at low temperatures efficiently. Such results would immediately amplify the classes of Gibbs states to which~\Cref{thm:efficient_quantum} applies to.

However, we now argue that it is possible to still achieve quadratic polynomial quantum speedups for solving these instances without imposing much structure beyond the Pauli sparsity.

\subsection{Classical and Quantum polynomial speed-ups}\label{subsec:poly_speedups}
Let us now discuss how to obtain more modest speed-ups in more general settings. 
\paragraph{Classical algorithm:} 
For a generic Pauli-sparse matrix $C$ and a subset of considered constraints $S\subseteq\dot\cD_C$, we can obtain a polynomial speed-up in $D$ %
using the sampling method and simple sparse matrix vector multiplication. More precisely, noting that the sparsity of $C$ is upper bounded by its number of Pauli strings, $s=\poly(n)$, we can classically compute $\GW(C,S,0)$ up to additive precision $\epsilon>0$ in time 
\begin{equation}
\label{eq:scaling_sparse}
\tilde{\mathcal{O}}\big(n^2\epsilon^{-5}\log(1/W)%
\,|S|\,s\,D\big),
\end{equation}
where the ``O tilde"  notation here omits factors of order $\mathcal{O}(\operatorname{polylog}(n,1/\epsilon))$ 
and 
$W=\tr{\sigma_{\lambda}}$ is the partition function of the last state $\sigma_{\lambda}$ in our HU sequence, with $\lambda=\mathcal{O}(n\epsilon^{-1})$. Note that one can at worst have $\log%
(1/W)=\mathcal{O}(n\,\epsilon^{-1})$%
.
Hence, for $S=\dot\cD_C$ and $|\dot\cD_C|=\poly(n,\epsilon^{-1})$, 
the obtained run-time corresponds to (at least) a quadratic speed-up in $D$ over the classical algorithms from~\cite{henze2025solvingquadraticbinaryoptimization,Arora2005}, which have a runtime $\tilde{\mathcal{O}}\left(\min \{D^{2.5}s\epsilon^{-2.5},D^{2.5}s^{0.5}\epsilon^{-3.5},D^2s\epsilon^{-9}\}\right)$, as mentioned in~\Cref{sec:SDP_solvers}.

To prove Eq.~\eqref{eq:scaling_sparse}, we approximate each expectation value $\text{tr}[B_i\,\sigma(\lambda)]$ by the ratio between estimates for $V_i=\text{tr}\big[B_i%
e^{-H}\big]$ and $W=\text{tr}\big[e^{-H}%
\big]$, where the short-hand notation $H=\lambda_0\tfrac{C}{\|C\|}+\sum_{Z_A\in S}\lambda_iZ_A$ is used. 
To estimate $V_i$ and $W$, we approximate the traces by averages over expectation values on randomly sampled a state from the Haar measure and the exponential operator $e^{-H}$ by a polynomial approximation, for instance a fragmented Taylor expansion~\cite{berry2015simulating}.
We estimate both $V_i$ and $W$ up to additive error $W\,\epsilon/2$. This, since $|V|/W\leq 1$, implies an additive approximation error $\mathcal{O}(\epsilon)$ for the ratio $V_i/W=\text{tr}[B_i\,\sigma(\lambda)]$, which is what we need. 
The required precision (for both $V_i$ and $W$) can be guaranteed with a number of samples per trace $N=\mathcal{O}\big(\frac{\|B_i\,e^{-H}%
\|^2}{\epsilon^2 W^2}\big)%
$ %
and $r= \lceil\|H\|\rceil%
$ fragments each one of truncation degree $\kappa=\mathcal{O}\big(\log\big(\frac{\|H\|}{\epsilon\,W}\big)\big)%
$. Each of the $N$ expectation values can be computed via $r\times\kappa$ multiplications of an $s$-sparse matrix on a $D$-dimensional vector, with run-time $\mathcal{O}(s\,D)$ and memory $\mathcal{O}(D)$. In turn, we need to repeat the estimation for each $B_i$ ($|S|$ times) and for each $\sigma(\lambda)$ in the HU sequence (of length at most $T=\mathcal{O}(n\,\epsilon^{-2})$). Hence, the total time complexity is $\mathcal{O}\big(N\, r\,\kappa\, |S|\, T\, s\, D\big)$. This, using that $i$) $\mathcal{O}\big(\frac{\|B_i\,e^{-H}\|^2}{\epsilon^2 W^2}\big)
\leq \mathcal{O}(1/\epsilon^2)$ (because $\|B_i\|\leq 1$ and $\|e^{-H}\|/W\leq 1$), $ii$) $\|H\|\leq \|\lambda\|_{l_1}=\mathcal{O}(n\, \epsilon^{-1})$, and $iii$) $\mathcal{O}\big(\log\big(\frac{\|H\|}{\epsilon\,W}\big)\big)=\tilde{\mathcal{O}}(1/W)$ gives the promised scaling.

\paragraph{Quantum algorithm:} Let us now discuss how to obtain quantum algorithms to solve the relaxed SDP in the regime where $C$ is Pauli-sparse but we do not necessarily have efficient quantum Gibbs samplers. To this end, we can simply invoke generic %
quantum SDP solvers, like the one from~\cite{vanApeldoorn2020}. There, the authors show that it is possible to solve the SDP $\textrm{GW}(C,S%
,\epsilon)$ in time $\tilde{\mathcal{O}}(|S
{|^{0.5} D^{0.5}}%
\,\textrm{poly}(\epsilon^{-1}))$, where %
``O tilde" notation now refers to omitted factors polynomial in $n$. 
Hence, for $S=\dot\cD_C$ and assuming again $|\mathcal{D}_C|=\poly(n)$ (for which it suffices to take the precision $1/\epsilon=\poly(n)$ too), the latter scaling corresponds to a more than quartic speed-up in $D$ over the classical algorithms from~\cite{henze2025solvingquadraticbinaryoptimization,Arora2005}.

Moreover, under the same assumptions, the obtained scaling also corresponds to a cubic speed-up in $D$ over the best known MMW-based quantum solver~\cite{GSLBrandao2022fasterquantum}, with runtime $\tilde{\mathcal{O}}\left(D^{1.5}s^{0.5+o(1)}\epsilon^{-15+o(1)}\right)$. We emphasize that the reason why the previous MMW-based quantum solver fails to attain a scaling $\mathcal{O}(D^{0.5})$ with the dimension is because the standard GW relaxation requires itself a number of constraints $\mathcal{O}(D)$, as already mentioned. In contrast, our relaxation does not suffer from this limitation, allowing us to unlock even a quadratic quantum speed-up in $D$ with respect to the SketchyCGAL algorithm from~\cite{Yurtsever2021}  (see~\Cref{sec:SDP_solvers}) or the classical algorithm described in the previous two paragraphs.

\section{Randomized rounding on classical computers}\label{sec:randomized_rounding}
One of the most attractive features of the GW-SDP relaxation is that its solution can be rounded to a feasible point of the original QUBO which provably performs well. 
The standard randomized rounding procedure works as follows: we take a solution $\rho$ of GW and draw a random vector $\ket{\psi}$ with i.i.d. Gaussian entries of unit variance. Letting $\ket{\phi}=\sqrt{\rho}\ket{\psi}$, we obtain a feasible point of the QUBO by setting $\ket{x}=\operatorname{sign}(\ket{\phi})$, where we apply the sign function component-wise.
However, performing this rounding procedure is not efficient in $n$, as it requires computing the product of the matrix $\rho$ with the dense vector $\ket{\phi}$, with complexity growing polynomially in $2^n$. %
Conveniently, we find that for our quantum-inspired approaches it is still possible to efficiently (i.e. in time $\textrm{poly}(n)$) estimate the energy density achieved by a heuristic randomized rounding procedure based on Monte Carlo methods. 
A rigorous analytical analysis of the performance of our simplified rounding method as well as its generalization to a quantum algorithm are beyond the scope of the current work. However, we do numerically demonstrate in~\Cref{sec:numerics} that our simplified rounding is able to certify that we produce good approximate QUBO solutions for very large instances. 
More precisely, we show there that our heuristic rounding gives a feasible solution to the original QUBO problem whose value closely matches that of its SDP relaxation. %

Our first step is to simplify the usual rounding procedure via (nearly) random vectors with efficient descriptions. That is, instead of considering a dense, random Gaussian vector (or, equivalently, a uniformly random unit vector), we consider families of (real) vectors generated by random quantum circuits.
We consider random vectors $\ket{\phi}=U\ket{0}\in\mathbb{R}^{2^n}$, where $U$ is drawn from an ensemble of random orthogonal matrices. If we pick $U$ from the Haar measure, we again get a random vector, but this would require exponential time. Instead, we consider $U$ coming from ensembles of random circuits. We know that, in the limit of infinite depth, the distribution converges to the Haar measure and recover the original rounding scheme. However, we can instead consider rounding schemes with approximate $t$-designs, which can be generated more efficiently~\cite{Ambainis2007}. Again, a rigorous analysis of the rounding procedure is left to future work, but we test it numerically in~\Cref{sec:numerics}.

We then consider the randomized rounding procedure with random vectors of the form $\operatorname{sign}(\sqrt{\rho}\,U\ket{0})$, where $U$ is a sufficiently shallow random circuit s.t. we can still compute $\bra{i}\sqrt{\rho}\,U\ket{0}$ efficiently for a fixed $i$. However, even then it would still take time $\Omega(2^n)$ to output the full solution. Nevertheless, we can efficiently estimate the (relative) expectation value of the randomized rounding procedure through Monte Carlo. %
That is, we can still efficiently estimate the energy density achieved by the rounded solution:%

\begin{proposition}\label{prop:rando_rounding_mc}
Let $C\in\mathbb{R}^{2^n\times 2^n}$ be a Pauli-sparse matrix with
\begin{align}
C=\sum\limits_{j=1}^m\alpha_jP_{(x_j,z_j)}
\end{align}
and $\rho$ a quantum state on $n$ qubits. Given an ensemble of random orthogonal matrices $U$ distributed according to a probability measure $\mu$, let $E(\rho)$ be the worst-case cost to compute $\bra{i}\sqrt{\rho}\,U\ket{0}$ up to constant relative precision for any computational basis state $\ket{i}$. Moreover, let $\ket{x}=\operatorname{sign}(\sqrt{\rho}\,U\ket{0})$. Then we can estimate the energy density $2^{-n}\bra{x}C\ket{x}$ of the solution up to additive error $\epsilon$ with probability of success $2/3$ in time 
\begin{align}\label{equ:complexity_bound_density}
\mathcal{O}(\epsilon^{-2}mE(\rho)\|\alpha\|_{\ell_1}^2).
\end{align}
\end{proposition}
\begin{proof}
Consider the random variable $X$ given by sampling $i$ uniformly at random from $[2^n]$ with
\begin{align}\label{equ:def_X}
X=x_i\sum\limits_{k=1}^{2^n}C_{i,k}x_k\,.
\end{align}
where $x_i=\braket{i\vert x}$.
Clearly, we have that
\begin{align}
\mathbb{E}(X)=2^{-n}\sum\limits_{i,k}C_{i,k}x_ix_k=2^{-n}\bra{x}C\ket{x}\,.
\end{align}
Let us now argue that $|X|\leq\|\alpha\|_{\ell_1}$. Note that, given $\ket{i}$, for each $P_{(x_j,z_j)}$ there is precisely one other computational basis state $\ket{k_{j,i}}$ s.t. $\bra{k_{j,i}}P_{(x_j,z_j)}\ket{i}\not=0$ and we can compute $\ket{k_{j,i}}$ efficiently. Namely, up to a phase, $\ket{k_{j,i}}=P_{(x_j,z_j)}\ket{i}$. As such, there are at most $m$ nonzero terms in the sum of~\Cref{equ:def_X} and we can find them efficiently. This gives that
\begin{align}\label{equ:efficient_rep_X}
X=x_i\sum\limits_{j=1}^{m}\alpha_j\bra{k_{j,i}}P_{(x_j,z_j)}\ket{i}x_{k_{j,i}}.
\end{align}
From the formula above and the fact that $x_i,x_{k_{j,i}}=\pm 1$, we obtain that $|X|\leq \|\alpha\|_{\ell_1}$. 
Thus, it follows from Hoeffding's inequality that taking the empirical average over $\mathcal{O}(\epsilon^{-2}\|\alpha\|_{\ell_1}^2)$ samples from $X$ suffices to estimate the expectation value of $X$ up to an error $\epsilon$ with probability of success $2/3$. To finish our argument and obtain the bound in~\Cref{equ:complexity_bound_density}, note that it immediately follows from the representation in~\Cref{equ:efficient_rep_X} that one sample of $X$ can be %
obtained via computing the sign %
of $m$ matrix entries of the form $\bra{i}\sqrt{\rho}\,U\ket{0}$. Each such sign can in turn be obtained via computing the value of the entry up to constant relative precision and costs %
at most time $E(\rho)$. As we need to repeat this $\mathcal{O}(\epsilon^{-2}\|\alpha\|_{\ell_1}^2)$ times, we obtain our claim.
\end{proof}
For tensor network approaches, we can efficiently approximate random $U$ via a circuit of shallow depth---see~\Cref{app:details_numerics} for details. In this case, we have $E(\rho)=\mathcal{O}(\textrm{poly}(n))$ and we can find an approximation to the energy density achieved by $x$ efficiently.
Note that~\Cref{prop:rando_rounding_mc} applies more generally to any sparse matrix for which %
we can efficiently compute the location of nonzero entries. 
That is, our method can be expanded to compute other average properties of the obtained solution string besides its energy density.
Thus, we can efficiently extract relevant information from the rounded solution even in the regime where just writing it down would take exponential time.

\section{Numerical examples}\label{sec:numerics}
We next benchmark our method on two families of instances: commuting $1D$ Hamiltonians with a small %
diagonal group $\cD_C$ and adjacency matrices of Kronecker graphs. 

\subsection{Commuting 1D Hamiltonians with small diagonal algebras}\label{subsec:1d_hamiltonians}
We start by an objective matrix $C$ given by commuting $1D$ Hamiltonians with a small diagonal subgroup $\cD_C$. 
For these instances, we have an almost complete picture, in the sense that we have analytical guarantees on the quality of our %
solution and that we can solve the SDP in polylogarithmic time in the dimension $D$. In particular, we showcase here the scalability of our method by solving the GW SDP as well as performing the simplified rounding from \Cref{sec:randomized_rounding} for gigantic instances (of up to dimension $2^{50}\simeq 10^{15}$).
Although we do not have results on the classical hardness of solving the underlying QUBOs, we do provide numerical evidence that solving the underlying QUBOs requires time that is superpolynomial in $D$. 
We do this not only by benchmarking against various SDP solvers but also by computing several parametrized complexity parameters of the underlying graphs.

\begin{figure}[ht!]
	\centering
	\includegraphics{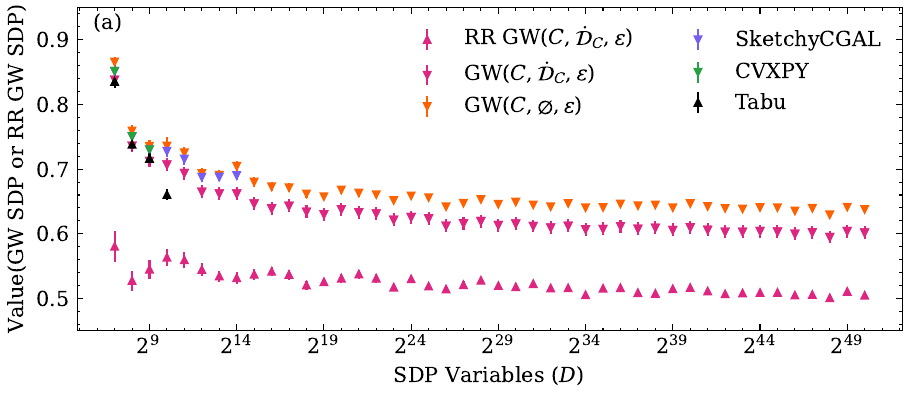}
    \\[0.3cm] 
    \includegraphics{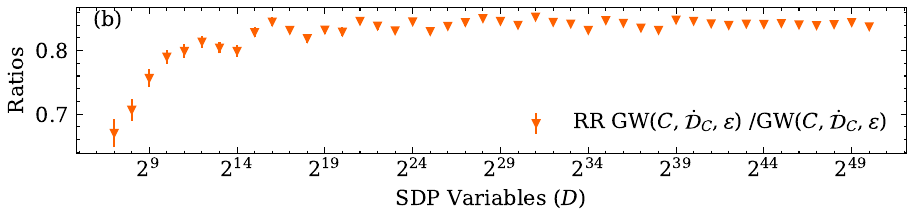}
	\caption[round sdp 1d]{{\bf Performance of %
    different solvers on $C$ given by a non-trivial 1D commuting Hamiltonian [see Eq. \eqref{eq:1d_hamiltonian}]: (a) GW SDP and rounded solutions (see~\Cref{app:details_numerics_tensor} for technique details of calculating the SDP and performing the randomized rounding with tensor network approach). (b) Ratio between the rounded and SDP solutions from HU solver.} $\GW(C,\dot{\cD}_C,\epsilon)$ denotes the SDP solution obtained from the HU solver. RR $\GW(C,%
    \dot{\cD}_C,\epsilon)$ refers to the mean objective value over rounded strings sampled from the $\GW(C,%
    \dot{\cD}_C,\epsilon)$, corresponding to a feasible, approximate solution of the associated $\QUBO$ problem. 
    The objective value $\QUBO(C)$, see~\Cref{eq:GW_relaxation_normalized}] of the exact solution is unknown, but it is guaranteed to lie within RR $\GW(C,%
    \dot{\cD}_C,\epsilon)$ and $\GW(C,%
    \dot{\cD}_C,\epsilon)$, up to additive error $\epsilon$. 
    The Monte Carlo sampling error for the randomized rounding is $\epsilon=0.5$---see \Cref{equ:complexity_bound_density} see for details.
    In turn, $\GW(C,\varnothing%
    ,\epsilon)$ represents the %
    solution to the SDP spectral relaxation. %
    SketchyCGAL and CVXPY denote the $\GW(C)$ SDP solutions obtained via SketchyCGAL~\cite{Yurtsever2021} and the standard CVXPY SDP Python library~\cite{diamond2016cvxpy}. The maximum problem sizes they can handle are $D=2^{14}$ and $D=2^9$, respectively, due to memory limitations.  
    Additionally, we plot the objective values $\QUBO(C)$ obtained using the standard Tabu search Python library (\texttt{dwave-tabu}) with 50 samples stably up to $D=2^{10}$. In the numerics, we observe that Tabu starts crashing frequently starting at $D=2^{11}$. Results are averaged over 100 randomly generated problem instances, as we observe a high degree of concentration for the value of the random instances. The error bars, except for those in $\GW(C,\dot{\cD}_C,\epsilon)$, represent the standard error of the mean across all instances. The error bar in $\GW(C,\dot{\cD}_C,\epsilon)$ denotes the error tolerance introduced in HU, which dominates the numerical errors. All numerical experiments run on a single, standard desktop computer with specifications in~\Cref{app:details_numerics}. 
    }\label{fig:sdp_1d_cluster} 
\end{figure}

We consider the family of instances given by objective matrices associated to $n$-qubit Hamiltonians of the form
\begin{align}
\label{eq:1d_hamiltonian}
\nonumber
C&=a^{\scriptscriptstyle(1)}_1 X_1Z_2  + \left(\sum_{i=2}^{n-5} a^{\scriptscriptstyle(2)}_iZ_{i-1} X_i Z_{i+1}\right) +\left(\sum_{i=2}^{n-6} a^{\scriptscriptstyle(3)}_i Z_{i-1} Y_iY_{i+1} Z_{i+2}\right)+ a^{\scriptscriptstyle(4)}_1Z_{n-6}Y_{n-5}Y_{n-4}X_{n-3}X_{n-2}X_{n-1}  \\
 &+a^{\scriptscriptstyle(4)}_2X_{n}+a^{\scriptscriptstyle(4)}_3 \frac{1}{6}(X_{n-3}X_{n-2}+Y_{n-3}Y_{n-2}+X_{n-2}X_{n-1}+Y_{n-2}Y_{n-1}+X_{n-3}X_{n-1}+Y_{n-3}Y_{n-1}),%
\end{align}
where $a^{\scriptscriptstyle(1)}_1, a^{\scriptscriptstyle(2)}_i, a^{\scriptscriptstyle(3)}_i, a^{\scriptscriptstyle(4)}_1,a^{\scriptscriptstyle(4)}_2$, and $a^{\scriptscriptstyle(4)}_3$ are randomly generated coefficients between $(0,1]$, and then renormalized to the factor given by $a^{\scriptscriptstyle(1)}_1+\sum a^{\scriptscriptstyle(2)}_i+\sum a^{\scriptscriptstyle(3)}_i+a^{\scriptscriptstyle(4)}_1+a^{\scriptscriptstyle(4)}_2+a^{\scriptscriptstyle(4)}_3$. %
The model is a modified 1D cluster Hamiltonian~\cite{Suzuki1971,Raussendorf2001} with additional $Z YY Z$ terms and boundary terms to make the generated diagonal subgroup %
nontrivial. Indeed, one can show that, for any $n\geq7$, the diagonal subgroup %
is $\cD_C=\{I,Z_{n-3}Z_{n-2},Z_{n-2}Z_{n-1},Z_{n-3}Z_{n-1}\}$.
We note that not all Pauli terms in this model commute with each other. This is important because, if that was the case, the graph would be disconnected and therefore computationally easy, as we show in Prop.~\ref{prop:commuting_disconnected} in \Cref{sec:commuting_regrouping}. However, the fact that the last 6 Pauli terms in Eq.~\eqref{eq:1d_hamiltonian} are grouped together with the same coefficient ($a^{\scriptscriptstyle(4)}_3$) makes the model commuting without computationally trivializing it (see~\Cref{sec:commuting_regrouping} for more details).

In Fig.~\ref{fig:sdp_1d_cluster} we compare the performance of our approach with other SDP and QUBO solvers on this class of instances.
A first important observation from Fig.~\ref{fig:sdp_1d_cluster} is that the spectral bound $\GW(C,\varnothing,\epsilon)$ and the $\GW$ bound $\GW(C,\dot{\cD}_C,\epsilon)$ do not coincide, the former being on average less tight than the later.
Second, we see that, for small instance sizes with $D\leq 2^9$, Tabu search outperforms our randomized rounding solutions. However, for $D\geq 2^{10}$ both these approaches give solutions of essentially equivalent quality on average. Moreover, while Tabu search cannot handle problem sizes larger than $D= 2^{10}$, with our approach we managed to reach $D\leq 2^{50}$ on a desktop, remarkably. In fact, as shown, something similar happens with the standard SDP solver CVXPY or the randomized SDP solver SketchyCGAL, which are able to track problem sizes only up to $D= 2^{9}$ and $D= 2^{15}$, respectively. For all the aforementioned algorithms, the main bottleneck to achieve larger system sizes was memory, as they required more memory than available to run larger instances. In contrast, our methods can leverage tensor networks to obtain highly memory efficient representations and we could, in principle, even consider larger system sizes than $50$ in the same machine.

Importantly, as mentioned, we currently lack rigorous understanding of the complexity of the QUBO problems associated to the instances studied in Fig.~\ref{fig:sdp_1d_cluster}. This means that, in principle, we run the risk of dealing with QUBO instances that are actually computationally easy, i.e. that can be solved in a time polynomial (instead of exponential) in the dimension $D=2^n$. However, as a sanity check, we provide numerical evidence against this possibility via lower bounds to three standard complexity parameters of the graphs underlying our instances:  For graphs defined by the nonzero entries of $C$ as their adjacency matrix, we compute the \emph{tree-width} (i.e., how close the graph is to a tree), the \emph{rank-width} (i.e., the maximum binary cut-rank across branch decompositions, closely related to clique-width), and the \emph{orientable genus} (i.e., the minimum number of handles needed to embed the graph on an orientable surface without edge crossings). 
We refer the interested reader to~\cite{Cygan2015} for a review on these parameters. 
The crucial point for our purposes is many exact and approximate algorithms for QUBO/Max-Cut run in time exponential in these parameters times a polynomial in the number of variables/vertices $D$. Thus, if any of these parameters scales as $\cO(n)$, the instances can be solved in time polynomial in $D$. Fig.~\ref{fig:complexity_parameter} shows a plot of two complexity parameters, where we can see that all three parameters grow superlinearly with $n$. 
This suggests that the studied instances are not easy for methods able to exploit structures quantified by the complexity parameters in question: with the parameters scaling faster than $\cO(n)$, those algorithms are expected to have a run-time at least superpolynomial in $D$ for our class of instances.

\begin{figure}[h]
	\centering
	\includegraphics[width=0.5\textwidth]{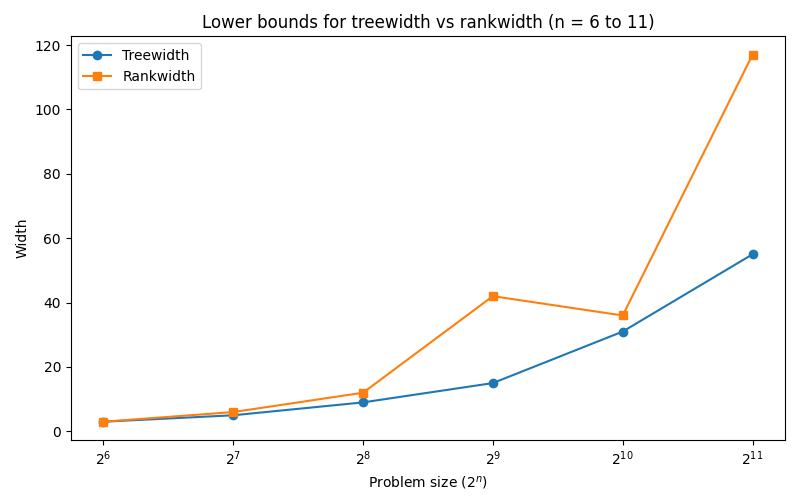}
	\caption[round sdp 1d]{Lower bounds to the tree- and rank-widths of the instances considered for Fig. \ref{fig:sdp_1d_cluster} %
    computed through convex relaxations~\cite{Cygan2015} Both lower bounds feature of a super-linear growth in the number of vertices (dimension $D$ of $C$). In addition, the genus of the graphs (not shown) is observed to grow even significantly faster.
    This implies that the known algorithms able to exploit structures associated to these complexity parameters cannot solve the corresponding QUBO instances in time less than super-polynomial in $D$ (see main text). Hence, this sanity check is consistent with the instances defined by Eq.~\eqref{eq:1d_hamiltonian} being computationally non-trivial. Note that these parameters only depend on the graph on which these instances are defined on, not on the particular choice of the weights $a_i$.
    }\label{fig:complexity_parameter} 
\end{figure}

\subsection{Kronecker graphs}\label{subsec:kronecker_graphs}
Kronecker graphs provide a highly efficient way to model the adjacency matrix of large networks--see~\cite{kronecker} and~\Cref{def:kronecker_graph} for their definition and details. Given their inherent tensor product structure, they also provide a natural testbed for our methods. 
In the numerics, we choose Kronecker graphs generated from the initiator graph $A\in \mathbb{R}^{4\times4}$ with adjacency matrix:

\begin{align}
A= \begin{pmatrix}
2 & 1 & 1 & 1 \\
1 & 2 & 1 & 0 \\
1 & 1 & 2 & 0 \\
1 & 0 & 0 & 2
\end{pmatrix}.
\end{align}
The Pauli matrices expansion of $A$ is:
\begin{equation}\label{eq:kronecker_initiator}
\begin{aligned}
A &= 2I\otimes I + 0.5I\otimes X + 0.5X\otimes I + X\otimes X + 0.5X\otimes Z +0.5Z\otimes X.
\end{aligned}
\end{equation}
The Kronecker graph model is generated by recursive Kronecker products of the chosen initiator graph as
\begin{align}\label{eq:kronecker_graph}
C = \left(\frac{A}{\|A\|}\right)^{\otimes k}
\end{align}
The recursive Kronecker graph model (see \Cref{def:kronecker_graph}) is then Pauli-sparsified using the procedures described in~\Cref{thm:pauli_sparsification} and the number of samples is given by:
\begin{equation}\label{eq:exact_sample_kronecker}
    \lfloor 2k\|C\|_{P,\ell_1}^{2}\epsilon^{-2}/1.5\rfloor
\end{equation}
where we set $\epsilon = 0.5$ in the numerics. The generated sparse approximation $\tilde{C}$ is renormalized as $\tilde{C}/\|\tilde{C}\|_{P,l_1}$, and, without loss of generality, we denote the renormalized matrix again by $\tilde{C}$.
The reason for this choice of initiator graph is based on preliminary numerical benchmarks, which indicate that these specific graphs yield a large gap between the spectral bound and $\GW(C,2^{[n]}\backslash\{\varnothing\},0)$.

After sparsifying, we perform the Krylov expansion described in~\Cref{sec:approximate_constraint_set} up to second and third order of $\tilde{C}$ to approximate the value of the semidefinite program. As can be observed in~\Cref{fig:kronecker}(c), the number of constraints for order $2$ is several order of magnitudes less than the full constraints, and for $3$ it is one order. Nevertheless, as can be seen in~\Cref{fig:kronecker}(a), the quality of the solution of our method and SketchyCGAL is essentially indistinguishable, indicating that the third order already captures most relevant constraints. 
In addition, we observe that the ratio between the rounded and SDP solutions from our method improves as the order of the constraints increases, reaching above $0.74$ for the third-order constraints in~\Cref{fig:kronecker}(b). This indicates that our method is effective at extracting high-quality QUBO solutions.

Interestingly, we only computed solutions for both SketchyCGAL and our method up to $n=12$, but for different reasons. Whereas we ran out of memory for $n=14$ for SketchyCGAL, we still had enough memory to run the third order for our method. However, the runtime would be prohibitively long and we refrained from completing the computation for this manuscript. But, importantly,~\Cref{fig:kronecker} shows that our method is competitive against state-of-the-art solvers for Kronecker graphs and the power of the Krylov heuristic discussed in~\Cref{sec:approximate_constraint_set} to obtain good solutions with significantly less constraints.

\begin{figure}[h!]
	\centering
    \begin{minipage}[b]{0.45\textwidth}
        \centering
        \includegraphics[width=0.868\textwidth]{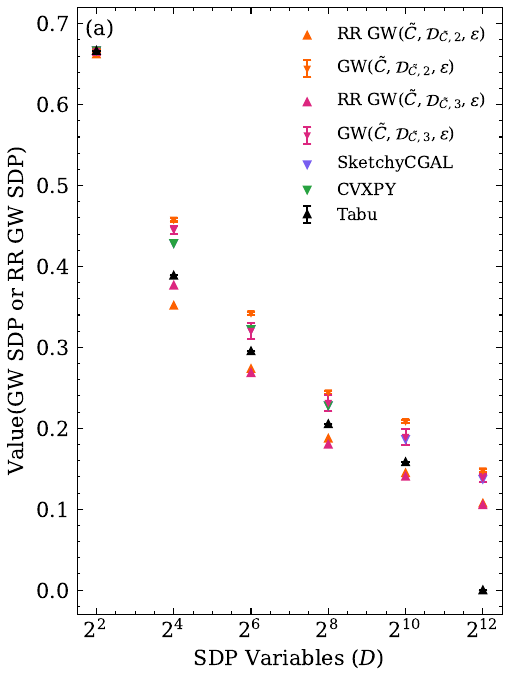}
    \end{minipage}%
    \hspace{0.01\textwidth}
    \begin{minipage}[b]{0.45\textwidth}
        \centering
        \includegraphics[width=0.9\textwidth]{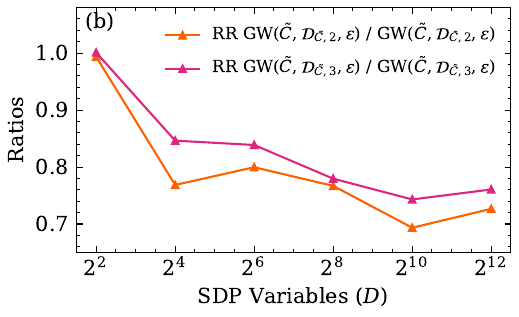}\\[0.4cm]
        \includegraphics[width=0.9\textwidth]{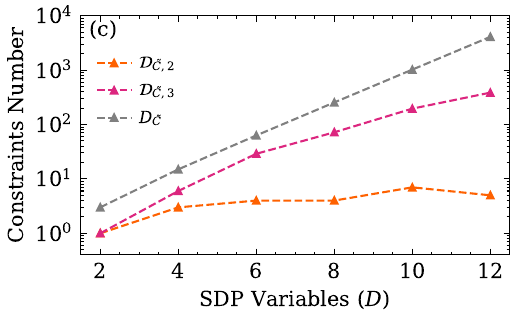}
    \end{minipage}
	\caption[Kronecker]{{\bf Performance of %
    different solvers on $\tilde{C}$ sampled from a Kronecker graph $A^{\otimes k}$ for $k=1,2,3,4,5,6$ [\Cref{eq:kronecker_graph}]: (a) GW SDP and rounded solutions(see~\Cref{app:details_numerics_sparse} for technique details of calculating the SDP and performing the randomized rounding(label with prefix RR) with sparse matrix-vector approach). (b) Ratio between the rounded and SDP solutions from HU solver. (c) Number of constraints used in the HU solver.}  
    The Pauli-sparse approximation $\tilde{C}$ is generated using the sampling method in~\Cref{thm:pauli_sparsification} with the number of samples determined by~\Cref{eq:exact_sample_kronecker}, and is then renormalized as $\tilde{C}/\|\tilde{C}\|_{P,l_1}$. The trace and the rounding are estimated using $50$ random vectors with the sparse matrix-vector multiplication method described in~\Cref{subsec:poly_speedups}.  
    $\GW(C,\cD_{\tilde{C},2},\epsilon)$ and $\GW(C,\cD_{\tilde{C},3},\epsilon)$ denote the SDP solutions obtained by the HU solver with constraint sets derived from the Krylov heuristics method~(\Cref{sec:approximate_constraint_set}) up to second and third order, respectively, and the corresponding error bars in (a) indicate the maximum error introduced within the HU framework.  
    RR $\GW(C,\cD_{\tilde{C},2},\epsilon)$ and RR $\GW(C,\cD_{\tilde{C},3},\epsilon)$ refer to the mean objective value over rounded strings sampled from $\GW(C,\cD_{\tilde{C},2},\epsilon)$ and $\GW(C,\cD_{\tilde{C},3},\epsilon)$, respectively, corresponding to feasible approximate solutions of the associated $\QUBO$ problem.  
    SketchyCGAL and CVXPY denote the $\GW(C)$ SDP solutions obtained via SketchyCGAL~\cite{Yurtsever2021} and the standard CVXPY SDP Python library~\cite{diamond2016cvxpy}. The maximum problem sizes they can handle are $D=2^{12}$ and $D=2^8$, respectively, due to memory limitations.  
    Additionally, we plot the objective values $\QUBO(C)$ obtained using the standard Tabu search Python library (\texttt{dwave-tabu}) with 300 samples. The results show a clear drop in performance at $D=2^{12}$, with error bars representing sampling variability in the Tabu procedure.  All numerical experiments run on a single, standard desktop computer with specifications in~\Cref{app:details_numerics}. 
    }\label{fig:kronecker} 
\end{figure}

\section{Acknowledgements}
DSF acknowledges funding from the European Union under Grant Agreement 101080142, the project EQUALITY and from the Novo Nordisk Foundation (Grant No. NNF20OC0059939 Quantum for Life). DSF acknowleges the support and hospitality of the Technology Innovation Institute, where parts of this work were performed.
Haomu Yuan acknowledges the leave to work away research fund from the University of Cambridge, the LabEx MILyon fund from the Universit\'e de Lyon, and the travel fund from the Cambridge Philosophical Society for supporting this project.

\bibliography{references}
\bibliographystyle{alpha}
\appendix
\newpage

\section{Table of relaxations of QUBO and relations}\label{app:table_qubo_relations}
The table below summarizes the various problems and relaxations considered in this work.
\begin{table}[!ht]
    \centering
    \setlength{\tabcolsep}{6pt}
    \begin{tabular}{p{0.1\textwidth} p{0.28\textwidth} p{0.58\textwidth}}
        \toprule
        \bf Problem & \bf Definition & \bf Relation to other problems \\
        \midrule
        \addlinespace
        QUBO\((C)\) &
        $\displaystyle\max_{x\in\{\pm1\}^D}x^\top C\,x.$ \ \ \  ~[\Cref{eq:original_problem}]&
        $\alpha_R\,\uGW(C)\le\QUBO(C)\le\uGW(C).$~[\Cref{equ:approximation}]\\
        \midrule
        \addlinespace
        $\uGW(C)$ &
        $\displaystyle
        \begin{aligned}
        \max_{Y}\;&\tr{C\,Y}, \quad\quad ~\text{[\Cref{eq:GW_relaxation}]}\\
        \text{s.t. }&Y\geq0,\\
        & Y_{ii}=1\;(i=1,\dots,D).
        \end{aligned}
        $ &
        \parbox[t]{\linewidth}{%
        \vspace{-1.0\baselineskip}%
        \raggedright
        1. $\uGW(C)\ge\QUBO(C)$,~[\Cref{eq:ugw_bound_1}]\\[0.7ex]
        2. $\GW(C)=\uGW(C)/(2^n\|C\|)\in[-1,1].$ 
        } \\
        \addlinespace
        \midrule
        \addlinespace
        $\GW(C)$ &
        $\displaystyle
        \begin{aligned}
        \max_{\rho}\;&\tr{\tfrac{C}{\|C\|}\,\rho},\quad ~\text{[\Cref{eq:GW_relaxation_normalized}]}\\
        \text{s.t. }&\rho\geq0,\;\tr\rho=1, \\
        &\langle i|\rho|i\rangle=2^{-n}\quad(i=1,\dots,2^n).
        \end{aligned}
        $ &
        \parbox[t]{\linewidth}{%
        \vspace{-1.0\baselineskip}%
        \raggedright
        1. $\GW(C) = \GW(C,2^{[n]}\backslash\{\varnothing\},0)$,\\[0.7ex]
        2. $\GW(C)=\uGW(C)/(2^n\|C\|)\in[-1,1].$ 
        } \\
        \addlinespace
        \midrule
        \addlinespace
        $\GW(C,S,\epsilon)$&
        $\displaystyle
        \begin{aligned}
        \max_{\rho}\;&\tr{\tfrac{C}{\|C\|}\,\rho}, \ \ \ ~\text{[\Cref{equ:relaxed_SDP}]}\\
        \text{s.t. }&|\tr{Z_A\,\rho}|\le\epsilon\quad(\forall A\in S),\\
        &\rho\geq0,\;\tr\rho=1.
        \end{aligned}
        $ &
        $\GW(C,\varnothing,0) \ge \GW(C,S,0)\ge \GW(C,2^{[n]}\backslash\{\varnothing\},0)\ge\frac{\QUBO(C)}{2^n\|C\|}.$~[\Cref{equ:interpolation_inequality_1}]\\
        \addlinespace
        \midrule
        \addlinespace
        $\GW(C,\varnothing,0)$&
        $\displaystyle
        \begin{aligned}
        \max_{\rho}\;&\tr{\tfrac{C}{\|C\|}\,\rho}, \\
        \text{s.t. }&\rho\geq0,\;\tr\rho=1.
        \end{aligned}
        $ &
    \parbox[t]{\linewidth}{%
    \vspace{-1.2\baselineskip}%
    \raggedright
    1. $\GW(C,\varnothing,0)=\lambda_{\max}(C)/\|C\|$,\\[0.7ex]
    2. $\frac{\lambda_{\max}(C)}{\|C\|}\ge \GW(C,S,0)$.~[\Cref{equ:interpolation_inequality_1}]
    } \\
        \bottomrule
    \end{tabular}
        \caption{Various relaxations for the QUBO cost \(\langle x,Cx\rangle\) considered in this work, their definitions and relation to other problems. QUBO\((C)\) denotes Quadratic unconstrained binary optimization; \(\uGW(C)\) denotes unnormalized Goemans-Williamson relaxation; \(\GW(C)\) denotes normalized Goemans-Williamson relaxation; \(\GW(C,S,\epsilon)\) denotes relaxed Goemans-Williamson relaxation with respect to a set \(S\subset2^{[n]}\) and error tolerance \(\epsilon>0\).}
    \label{tab:GW_hierarchy}
\end{table}

\section{Sparsifying the cost matrix with importance sampling}\label{sec:pauli_sparsification}
One pressing question posed by our work is to identify which cost matrices admit a Pauli-sparse representation. We now argue that the $\ell_1$ norm of the expansion in Pauli allows us to quantify how well we can approximate a given cost matrix in terms of Paulis. More precisely:
\begin{definition}[Pauli $\ell_1$ norm]\label{defi:Pauli_norm_appendix}
Let $C\in\mathbb{R}^{2^n\times 2^n}$ with Pauli expansion
\begin{align}\label{equ:Pauli_expansion_appendix}
C=\sum_{(x,z)\in\{0,1\}^{2n}}c_{(x,z)}\,P_{(x,z)}.
\end{align}
We define its Pauli $\ell_1$ norm, $\|C\|_{P,\ell_1}$, as 
\begin{align}
\|C\|_{P,\ell_1}=\sum_{(x,z)\in\{0,1\}^{2n}}|c_{(x,z)}|.
\end{align}
\end{definition}
This norm is a standard tool used e.g. for Hamiltonian simulation~\cite{Campbell2019}. Like in the QDRIFT approach to Hamiltonian simulation, it is possible to use standard matrix concentration results to conclude that it quantifies how well we can approximate a given Hamiltonian in operator norm by a random Pauli-sparse Hamiltonian:
\begin{thm}\label{thm:pauli_sparsification}
Let $C\in\mathbb{R}^{2^n\times 2^n}$ with Pauli expansion
\begin{align}%
C=\sum_{(x,z)\in\{0,1\}^{2n}}c_{(x,z)}\,P_{(x,z)}.
\end{align}
Then for every $\epsilon>0$, there exists a matrix $\tilde{C}$ that is $\mathcal{O}(n\|C\|_{P,\ell_1}^2\epsilon^{-2})$ Pauli-sparse and such that:
\begin{align}
\|C-\tilde{C}\|\leq \epsilon.
\end{align}
Furthermore, we can obtain such a $\tilde{C}$ with probability of success at least $2/3$ given $\mathcal{O}(n\|C\|_{P,\ell_1}^2\epsilon^{-2})$ samples from the distribution $p$ on Paulis with density  $p(x,y)=\frac{|c_{x,y}|}{\|C\|_{P,\ell_1}}$.
\end{thm}
\begin{proof}
As mentioned, the claim follows from standard matrix concentration inequalities. Indeed, consider the random matrix $X$ that is given by $\|C\|_{P,\ell_1}\operatorname{sign}(c_{(x,y)})P_{(x,y)}$ with probability $p(x,y)$. It is then not difficult to see that $\mathbb{E}(X)=C$. Furthermore, we have that:
\begin{align}
\|X\|\leq \|C\|_{P,\ell_1}\operatorname{ a.s.},\quad X^2\leq \|C\|_{P,\ell_1}^2 I\operatorname{ a.s.}.
\end{align}
It then follows from the Matrix Hoeffding inequality~\cite[Theorem 3.1]{Tropp2011} that for $X_1,\ldots,X_m$ i.i.d. samples from $X$ we have that:
\begin{align}
\mathbb{P}\left(\|m^{-1}\sum\limits_{k=1}^mX_i-C\|\geq \epsilon\right)\leq 2^n\textrm{exp}\left(-\frac{\epsilon^2m}{8\|C\|_{P,\ell_1}^2}\right),
\end{align}
from which it follows that $m=\mathcal{O}(n\|C\|_{P,\ell_1}^2\epsilon^{-2})$ suffices to ensure that $\|m^{-1}\sum\limits_{k=1}^mX_i-C\|\leq\epsilon$ with probability of success at least $2/3$. Furthermore, it is also clear that $\tilde{C}=m^{-1}\sum\limits_{k=1}^mX_i$ is $m$ sparse in the Pauli basis by construction and that it suffices to generate samples from $p$ to find it.
\end{proof}
The reason we prove such a concentration bound is the fact that a good approximation in operator norm ensures a good approximation for the values of $\QUBO$ or $\uGW$.
For each feasible point for $\QUBO(\tilde{C})$ or $\uGW(\tilde{C})$ for $\tilde{C}$ s.t. $C$ $\|C-\tilde{C}\|\leq \epsilon$ has the same value for $\QUBO(\tilde{C})$ or $\uGW(\tilde{C})$ up to an error of $\epsilon\, 2^n$. This follows from H\"older's inequality, as we have that:
\begin{align}
\left|\tr{(\tilde{C}-C)Y}\right|\leq \|C-\tilde{C}\|\|Y\|_{\ell_1}=\epsilon\, 2^n
\end{align}
for $Y$ a feasible point of either $\QUBO(\tilde{C})$ (in the sense of $Y=\ket{x}\bra{x}$) or $\uGW(\tilde{C})$. This shows that $\|C\|_{P,\ell_1}$ gives an operationally justified measure of how well we can approximate $C$ by Pauli-sparse matrices for the problem at hand. 

Let us show that there simple examples of graphs that are \emph{not} Pauli-sparse and not even sparse in the traditional matrix basis, but that admit a Pauli-sparse approximation. For instance, consider the complete graph. It is not difficult to see that its (normalized) adjacency matrix is given by:
\begin{align}
C=\ketbra{+}^{\otimes n}-I/2^n=\frac{1}{2^n}\left(\sum_{x\in\{0,1\}^n}P_{(x,0)}\right).
\end{align}
We see that $\|C\|_{P,\ell_1}=1$, but it has $2^n-1$ nonzero terms in its expansion. Thus, $C$ admits an efficient sparsification (by just picking $\sim n$ subsets of $[n]$ uniformly at random).

One pressing question is how to sample from $p$ above. 
For the case of Kronecker graphs considered in~\Cref{sec:kronecker_graphs}, this can be easily achieved, as the distributions are in product form and this can be achieved in time polynomial in $n$.

In full generality, it is certainly a difficult task to sample from this distribution, particularly computing the normalization $\|C\|_{P,\ell_1}$. However, we could e.g. use Monte Carlo Markov Chain algorithms to try to generate approximate samples if we can compute the ratios $c_{(x_1,z_1)}/c_{(x_2,z_2)}$. %

Finally, it is once again worth mentioning that the Pauli decomposition does not necessarily preserve the natural symmetries of $\QUBO$, in the sense that different permutations or embeddings of the graph could change the $\ell_1$ norm of the Pauli coefficients. In addition, embedding the graph into qudits instead of qubits could also be advantageous in the sense of generating a smaller $\ell_1$ norm. We leave investigating these aspects in more detail to future work.

\section{Algorithm to update the state}\label{sec:HU_description}
In this section, we present the algorithm to update the state $\rho_t$ with the binary search subroutines to converge to the candidate $\GW(C, S, \epsilon)$ value~$\mu$ to the $\epsilon$-approximation of the optimum, at an exponential rate.
The state update algorithm is based on the HU framework first introduced in~\cite{GSLBrandao2022fasterquantum}. Given a cost matrix $C$, a set of constraints $S \subseteq 2^{[n]}$, and an error threshold $\epsilon > 0$, the algorithm to solve the SDP proceeds as follows. Initially, we set a candidate value $\mu$ for the solution of $\GW(C, S, \epsilon)$, and define $T = 16 \epsilon^{-2} n$ and $\rho_0 = \frac{1}{2^{n}}$.
The algorithm iterates to update a $\rho$'s Hamiltonian along the $C$ and $S$ as long as $T > 0$ by performing the following steps is given at~\Cref{algorithm:HamiltonianUpdate}, and the iteration complexity is $\mathcal{O}\big(\log(n)\epsilon^{-2}\big)$~\cite{GSLBrandao2022fasterquantum}.

\begin{center}
\begin{minipage}{0.8\textwidth}
\begin{algorithm}[H]
    \caption{$\epsilon$-Hamiltonian Update (HU)}\label{algorithm:HamiltonianUpdate}
    \KwData{$\epsilon > 0, \mu > 0, T =16 \epsilon^{-2} n,  \rho_0 = \frac{1}{2^{n}}$ }
    \KwResult{Feasible, $\rho_t$; Infeasible.}
    \While{$T\geq 0$}{
        Compute $\mu_c = \tr{C \rho_t}$ and $\mu_A=\tr{Z_A \rho_t}, \forall A \in S$.

        \If{$\mu_c \geq \mu- \epsilon$ and $\forall A \in S, |\mu_A| \leq \epsilon$}{ 
            \Return Feasible, $\rho_t$.}
            \eIf{$\mu_c \leq \mu- \epsilon$}{
            Update $\rho_t$ using:
            \begin{align*}
                \rho_{t+1} &\leftarrow \exp\left(\log\left(\rho_t\right) - y_t C\right), \\
                t          &\leftarrow t+1,
            \end{align*}
            where $y_t = \mu - \mu_c$.

            Update $T \leftarrow T - y_t^2$.
            }{\If{there exists an $A_i$ such that $|\mu_{A_i}|> \epsilon$}{
                Update $\rho_t$ using:
                \begin{align*}
                        \rho_{t+1} & \leftarrow \exp\left(\log\left(\rho_t\right) - y_t z_{A_i}\right),\\
                         t         & \leftarrow t+1,
                    \end{align*}
                where $y_t = \mu - \mu_A$.

                Update $T \leftarrow T - (\mu - \mu_A/4)^2$.
                    }
            }
        }
        \If{$T< 0$}{\Return Infeasible.}
\end{algorithm}
\end{minipage}
\end{center}

The binary search begins with initial lower bound $\mu_l$ and upper bounds $\mu_u$ for~$\mu$, typically chosen as $-1$ and $1$, since $\GW(C, S, \epsilon) \in [-1,1]$. At each iteration, the interval is halved: if the feasibility condition obtained from the HU is satisfied, the lower bound is updated; otherwise, the upper bound is updated. Thus, the time complexity of the binary search method to converge to the chosen error $\epsilon$ is $\mathcal{O}\!\left(\log(\epsilon^{-1})\right)$.
A detailed description of the algorithm is provided in~\Cref{algorithm:binarysearch}.

\begin{center}
\begin{minipage}{0.8\textwidth}
\begin{algorithm}[H]
    \caption{$\epsilon$-Binary Search}
    \label{algorithm:binarysearch}
    \KwData{$\epsilon > 0, \mu_l, \mu_u$.}
    \KwResult{The $\epsilon$-approximation of the optimum $\mu$.}
    $\mu \leftarrow (\mu_l+\mu_u)/2$.\\
    \While{$\mu_u-\mu_l\geq \epsilon$}{
        Run HU with $\mu = \mu_u$.\\
        \eIf{Output feasible}{
            $\mu_l \leftarrow \mu$}{
            $\mu_u \leftarrow \mu$.
            }
        $\mu \leftarrow (\mu_l+\mu_u)/2$
    }
    \Return $\mu$.
\end{algorithm}
\end{minipage}
\end{center}

\section{Description of the numerical procedures}\label{app:details_numerics}
In this section, we introduce an efficient way of calculating the trace value and randomized rounding occurring in the HU estimating $\GW(C,S,\epsilon)$ to 1D-Hamiltonian [\Cref{subsec:1d_hamiltonians}] and the Kronecker graph [\Cref{subsec:kronecker_graphs}].
Our numerical experiments are performed on the following machine:\\
\indent\textbf{Operating System:} Ubuntu 22.04.4 LTS\\
\indent\textbf{Processor:} Intel Core i9-13900K\\
\indent\textbf{Memory:} 128GB DDR5\\
\indent\textbf{Storage:} 2TB NVMe SSD\\
\indent\textbf{Standardized Software Packages:} Python 3.10.14, cvxpy 1.4.2, cvxopt 1.3.2, dwave-tabu 0.5.0, scipy 1.11.4.

\subsection{The tensor network approach}\label{app:details_numerics_tensor}
As in the 1D Hamiltonian model, the Hamiltonian elements in updated state during the HU are mutually commuted. Then, we can construct a tensor network to efficiently estimate $\tr{\frac{C}{\|C\|}\sigma(\lambda)}$ and $\tr{Z_A\sigma(\lambda)}$, where $\sigma(\lambda)$ is the state generated during the HU. Consider a parameterized Hamiltonian generated in the HU as in~\Cref{lemma:structure_SDP},
\begin{equation}\label{eq:hamiltonian_update_Hamiltonian}
H(\lambda,C,S)=-\lambda_{C} \frac{C}{\|C\|} - \sum_{A \in S} \lambda_{A} Z_{A}.
\end{equation}
The the corresponding update state is given by the Gibbs state:
\begin{equation}\label{eq:1d_hamiltonian_state}
\sigma(\lambda)= \frac{\exp \left(H(\lambda,C,S)\right)}{\operatorname{tr}\!\left[\exp\left(H(\lambda,C,S)\right)\right]},
\end{equation}
where $\lambda=\{\lambda_C,\lambda_A\}$ are the parameters to be determined in the HU.
It is sufficient to estimate $\tr{\sigma(\lambda)}$ in order to compute $\tr{\frac{C}{\|C\|}\sigma(\lambda)}$ and $\tr{Z_A\sigma(\lambda)}$ via the following partition function relations:
\begin{equation}\label{eq:partition_function}
\frac{\partial}{\partial \lambda_C}\big(-\log Z(\lambda)\big) = \tr{\frac{C}{\|C\|}\sigma(\lambda)}, 
\qquad 
\frac{\partial}{\partial \lambda_A}\big(-\log Z(\lambda)\big) = \tr{Z_A \sigma(\lambda)},
\end{equation}
where the partition function $Z$ is defined as
\begin{equation}\label{eq:partition_function_main}
Z(\lambda) = \operatorname{tr}\!\left[\exp\!\left(H(\lambda,C,S)\right)\right].
\end{equation}
Thus, to estimate $\tr{\frac{C}{\|C\|}\sigma(\lambda)}$ and $\tr{Z_A\sigma(\lambda)}$ up to the error $\epsilon_p$, we can compute the partition function $Z(\lambda)$ under small perturbations $\delta_{\lambda_C}$ and $\delta_{\lambda_A}$ under error tolerance of $\mathcal{O}(\sqrt{\epsilon}_p)$, respectively.

Further, by commuting relations, the exponential of $H(\lambda,C,S)$ to 1D-Hamiltonian in~\Cref{eq:1d_hamiltonian} can be factorized as:
\begin{equation}\label{eq:hamiltonian_exponential}
\begin{aligned}
    & \exp\!\left(H(\lambda,C,S)\right)\\
 = & \exp\left(-\lambda_{C} \frac{a^{\scriptscriptstyle(1)}_1}{\|C\|}X_1Z_2\right)\left(\prod_{i=2}^{n-5} \exp\left(-\lambda_{C} \frac{a^{\scriptscriptstyle(2)}_i}{\|C\|}Z_{i-1} X_i Z_{i+1}\right)\right)\left(\prod_{i=2}^{n-6}\exp\left(-\lambda_{C} \frac{a^{\scriptscriptstyle(3)}_i }{\|C\|}Z_{i-1} Y_iY_{i+1} Z_{i+2}\right)\right) \\
    &\quad  \exp\left(-\lambda_{C} \frac{a^{\scriptscriptstyle(4)}_1}{\|C\|}Z_{n-6}Y_{n-5}Y_{n-4}X_{n-3}X_{n-2}X_{n-1}\right)\exp\left(-\lambda_{C} \frac{a^{\scriptscriptstyle(4)}_2}{\|C\|}X_{n}\right)\exp\left(-\lambda_{C} \frac{a^{\scriptscriptstyle(4)}_3}{\|C\|}A'\right) \\
    &\quad \quad\quad  \left(\prod_{A \in S} \exp\left(-\lambda_{A} Z_A\right)\right),
\end{aligned}
\end{equation}
where $A'$ is the Pauli string acting on the last three qubits, $A' =\frac{1}{6}(X_{n-3}X_{n-2}+Y_{n-3}Y_{n-2}+X_{n-2}X_{n-1}+Y_{n-2}Y_{n-1}+X_{n-3}X_{n-1}+Y_{n-3}Y_{n-1})$.
As for a given unitary $U_p =\exp(P)\in \mathbb{C}^{2^t \times 2^t}$ where $P$ is the Pauli strings up to a factor and $t$ is the number of qubits that $U$ acts on, we can reshape it into a rank-$2t$ tensor with 2-dimensional indices:
\[
U_p' \in \mathbb{C}^{\underbrace{2 \times 2 \times \cdots \times 2}_{2t \ \mathrm{times}}},
\]
with the outer indicies as ${I_{q_1}^{\mathrm{i}} I_{q_2}^{\mathrm{i}} \dots I_{q_t}^{\mathrm{i}}, I_{q_1}^{\mathrm{o}} I_{q_2}^{\mathrm{o}} \dots I_{q_t}^{\mathrm{o}}}$, where the superscript $q_i\in \{1,\dots, n\}$ denotes the qubit index, and the subscript $\mathrm{i}$ and $\mathrm{o}$ correspond to the input and output of the state, respectively. Thus, every $\exp$ element in~\Cref{eq:hamiltonian_exponential} can be represented with the rank-$2t$ tensor, where $t$ means the corresponding Pauli strings act nontrivially on---see~\Cref{fig:tensor_network}(a-c,f-i). 

Then, we show the contraction method for constructing the tensor network representing $\exp\!\left(H(\lambda,C,S)\right)$. For example, for the two tensors $\exp(-\lambda_{C} \frac{a^{\scriptscriptstyle(2)}_i}{\|C\|} Z_{i-1} X_i Z_{i+1})$ and $\exp(-\lambda_{C} \frac{a^{\scriptscriptstyle(2)}_{i+1}}{\|C\|} Z_{i} X_{i+1} Z_{i+2})$, we can connect them along the output indices of the first tensor, $I_{i}^{\mathrm{o}} I_{i+1}^{\mathrm{o}}$, to the input indices of the second tensor, $I_{i}^{\mathrm{i}} I_{i+1}^{\mathrm{i}}$, sequentially. In this way, we can generate a larger tensor representing $\prod_{i=2}^{n-5} \exp(-\lambda_{C} \frac{a^{\scriptscriptstyle(2)}_i}{\|C\|} Z_{i-1} X_i Z_{i+1})$, with outer indices $I_{1}^{\mathrm{i}} I_{2}^{\mathrm{i}} \dots I_{n-4}^{\mathrm{i}}, I_{1}^{\mathrm{o}} I_{2}^{\mathrm{o}} \dots I_{n-4}^{\mathrm{o}}$---see~\Cref{fig:tensor_network}(j). The same technique can be used to generate $\prod_{i=2}^{n-6} \exp(-\lambda_{C} \frac{a^{\scriptscriptstyle(3)}_i}{\|C\|} Z_{i-1} Y_i Y_{i+1} Z_{i+2})$ [\Cref{fig:tensor_network}(k)]. Hence, by connecting the input indices, $I_{q_i}^{\mathrm{i}}$, to the output indices, $I_{q_i}^{\mathrm{o}}$, of the tensors corresponding to the product elements from left to right in~\Cref{eq:hamiltonian_exponential} iteratively, we generate the tensor network representing $\exp\!\left(H(\lambda,C,S)\right)$, with $n$ input indices $I_1^{\mathrm{i}} I_2^{\mathrm{i}} I_3^{\mathrm{i}} \dots I_n^{\mathrm{i}}$ and $n$ output indices $I_0^\mathrm{o} I_1^\mathrm{o} I_2^\mathrm{o} \dots I_n^\mathrm{o}$---see \Cref{fig:tensor_network_large}. Finally, the trace value is obtained by contracting all input indices $I_{i}^{\mathrm{i}}$ with the output indices $I_{i}^{\mathrm{o}}$ of the network corresponding to the same qubits.

\begin{figure}[htbp]
    \centering
    \includegraphics[width=0.9\textwidth]{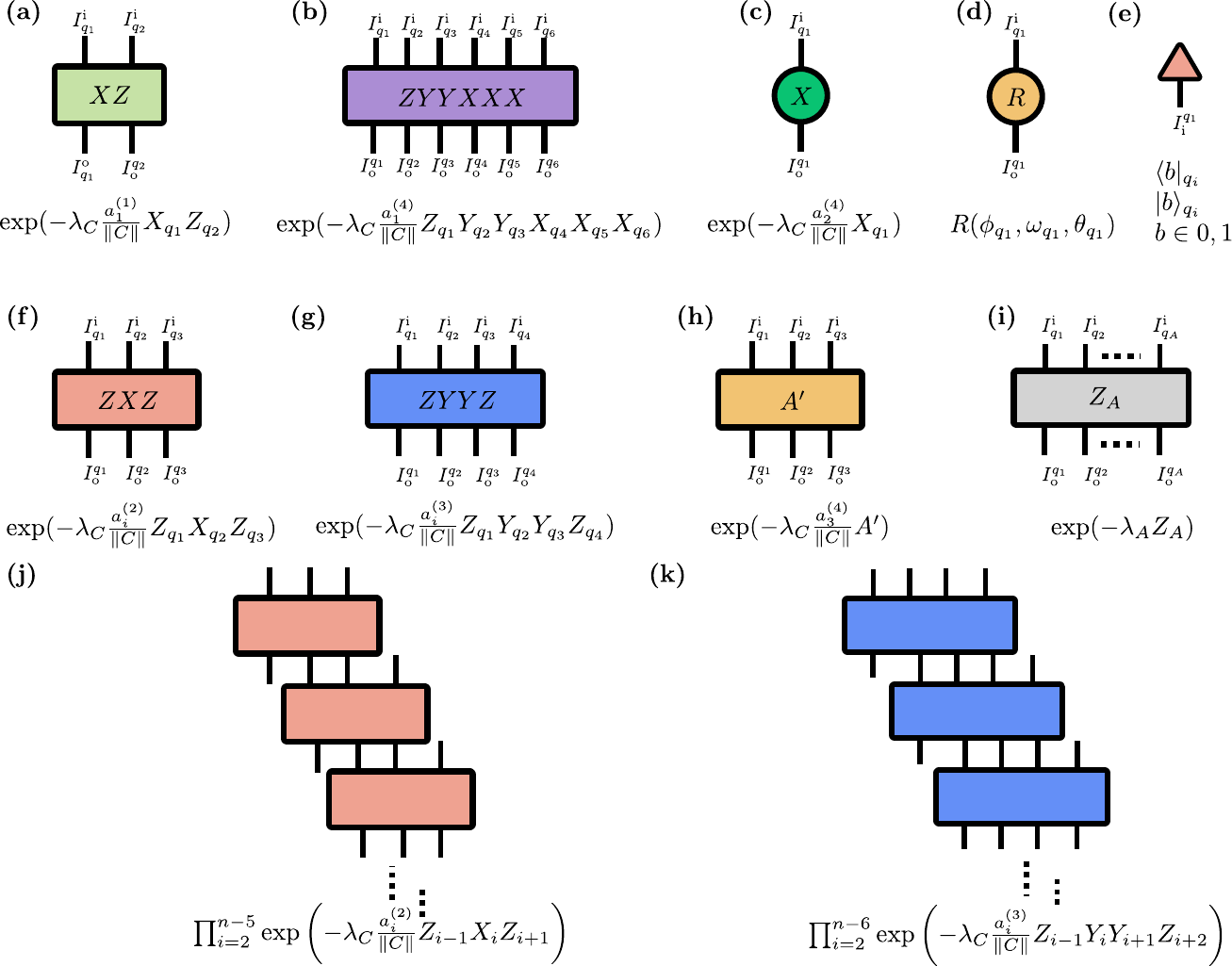}
    \caption{Parameterized tensors associated with the partition function of the 1D Hamiltonians $\GW(C,S,\epsilon)$[\Cref{eq:partition_function_main}], together with the additional tensor nodes used for randomized rounding. The explicit mathematical expressions corresponding to each tensor are displayed beneath the respective nodes. Square nodes denote tensors acting on multiple qubits, triangular nodes represent single-qubit vector tensors, and circular nodes correspond to single-qubit unitaries. Colors and text labels are used to differentiate tensors beyond shape alone. The index $I^\mathrm{i}_{q_i}$ indicates the input leg associated with the $q_i$-th qubit, while $I^\mathrm{o}_{q_i}$ denotes the corresponding output leg.
    }
    \label{fig:tensor_network}
\end{figure}

\begin{figure}[htbp]
    \centering
    \includegraphics[width=0.8\textwidth]{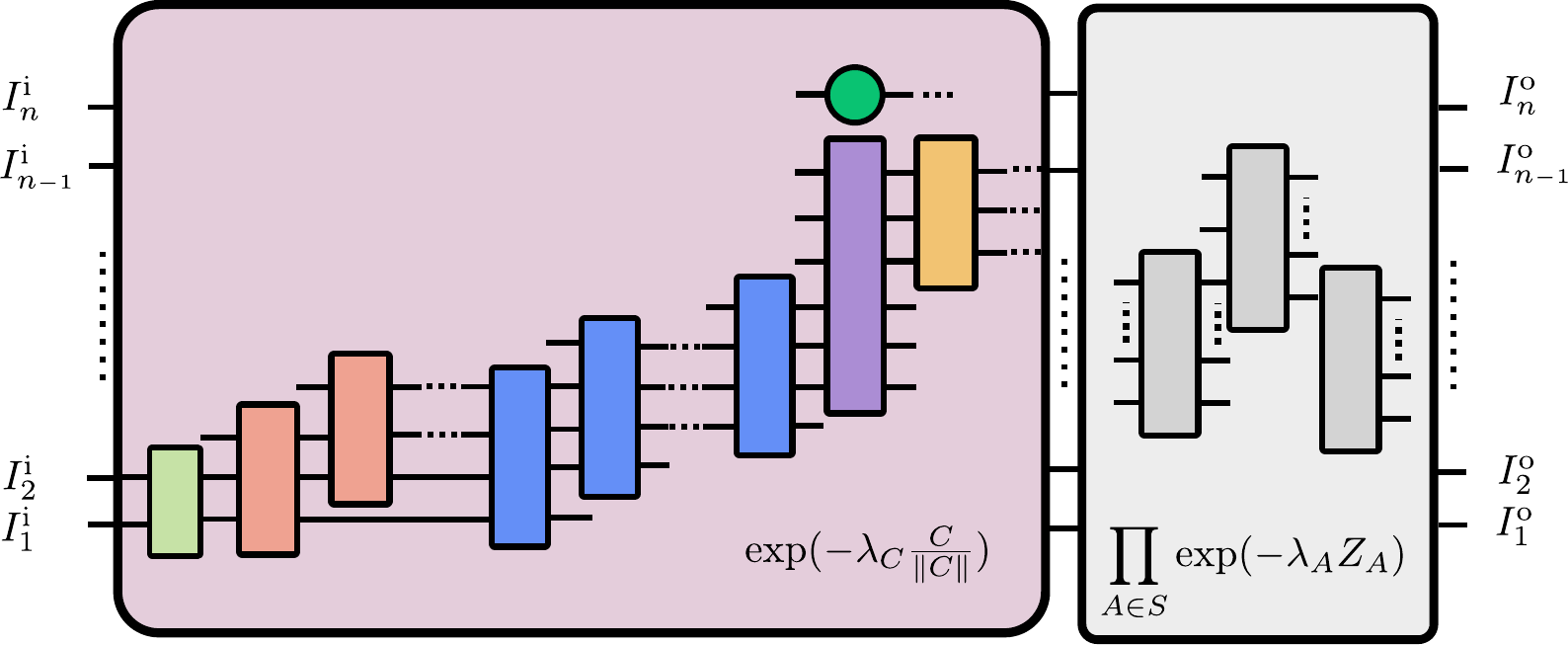}
    \caption{Tensor network representation of 
            $\exp\!\left(H(\lambda,C,S)\right)$. 
            The network is built by sequentially contracting the tensors associated with the terms in 
            $\exp(-\lambda_{C} \frac{C}{\|C\|})$ along their shared indices (large pink square), followed by contraction with the tensors corresponding to the constraint Hamiltonians (large grey square). 
            External legs $I_1^{\mathrm{i}}, I_2^{\mathrm{i}}, \dots, I_n^{\mathrm{i}}$ and 
            $I_1^{\mathrm{o}}, I_2^{\mathrm{o}}, \dots, I_n^{\mathrm{o}}$ denote input and output indices for each qubit. 
            The trace of the operator is obtained by contracting each input index $I_i^{\mathrm{i}}$ with the output index $I_i^{\mathrm{o}}$ iteratively.}
    \label{fig:tensor_network_large}
\end{figure}

Further, to perform the randomized rounding method, we need to evaluate $\operatorname{sign}{\bra{i}\sqrt{\sigma(\lambda)}\,U\ket{0}}$ in~\Cref{prop:rando_rounding_mc} where the $U$ is randomly generated from an ensemble of random orthogonal matrices---see details in~\Cref{prop:rando_rounding_mc}.
In the numerics, we construct such an ensemble $U$ by considering the real part of the tensor product of parameterized universal rotations on each qubit:
\begin{equation}
U = \bigotimes_{i=1}^n R(\phi_i,\omega_i,\theta_i)
\end{equation}
where
\begin{equation}\label{eq:universal_rotation}
    R(\phi_i,\omega_i,\theta_i) = 
    \begin{bmatrix}
    e^{-i(\phi_i+\omega_i)/2} \cos\frac{\theta_i}{2} & - e^{i(\phi_i-\omega_i)/2} \sin\frac{\theta_i}{2} \\
    e^{-i(\phi_i-\omega_i)/2} \sin\frac{\theta_i}{2} & e^{i(\phi_i+\omega_i)/2} \cos\frac{\theta_i}{2}
    \end{bmatrix}.
\end{equation}
If we consider $R(\phi_i,\omega_i,\theta_i) = \Re\left(R(\phi_i,\omega_i,\theta_i)\right)+i\Im\left(R(\phi_i,\omega_i,\theta_i)\right)$, the elements induced by the tensor product of $R(\phi_i,\omega_i,\theta_i)$ are a linear combination of $2^n$ terms:
\begin{equation}
\bigotimes_{j\in \{n\}} \Re\left(R(\phi_j,\omega_j,\theta_j)\right)\bigotimes_{k\in \{n\}}i\Im\left(R(\phi_k,\omega_k,\theta_k)\right),
\end{equation}
where the one with the even number of imaginary parts is real. As the $\Re\left(R(\phi_i,\omega_i,\theta_i)\right)$ is diagonal and 
$\Im\left(R(\phi_i,\omega_i,\theta_i)\right)$ has the matrix form:
\begin{equation}
    \begin{bmatrix}
     A& B\\
    B & -A
    \end{bmatrix},
\end{equation}
and we know
\begin{equation}
    \begin{bmatrix}
     A_1& B_1\\
    B_1 & -A_1
    \end{bmatrix} \begin{bmatrix}
     A_2& B_2\\
    B_2 & -A_2
    \end{bmatrix} = (A_1A_2+B_1B_2)I 
\end{equation}
is also diagonal. Thus, we construct an ensemble of random orthogonal matrices for the randomized rounding method. In the tensor network method, to evaluate $\bra{i}\sqrt{\sigma(\lambda)}\,U\ket{0}$, we first add $U$ tensor using $n$ $R(\phi_i, \omega_i, \theta_i)$ tensors[\Cref{fig:tensor_network}(d)] with random sampled parameters, and input index $I_i^{\mathrm{i}}$ and output index $I_i^{\mathrm{o}}$, respectively. The sampled $\bra{i}$ can be generated with $n$ vector tensors[\Cref{fig:tensor_network}(e)] containing indices $I_0^{\mathrm{i}},I_1^{\mathrm{i}},\dots,I_n^{\mathrm{i}}$, respectively, and so as the $\ket{0}$ containing indices $I_0^{\mathrm{o}},I_1^{\mathrm{o}},\dots,I_n^{\mathrm{o}}$, respectively. Thus, by contracting the the vector tensor of $\bra{i}$, the tensor network of $\sigma(\lambda/2)$, the random tensor $U$, and the vector tensor of $\ket{0}$ sequentially along the corresponding indices, we can obtain $\bra{i}\sqrt{\sigma(\lambda)}\,U\ket{0}$---see \Cref{fig:tensor_network_rounding} for details.
Here, we ignore the normalized factor to the state $\sqrt{\sigma(\lambda)}$ as this won't change the sign of outputs.
The method can also be extended to the quantum computer easily as $R(\phi_i,\omega_i,\theta_i)$ is often considered as an essential quantum gate. As the tensor network is a structured 1D network for partition function and randomized rounding, the time complexity of the contraction is linear to the number of elements in the tensor network---the number of elements in the cost Hamiltonian $C$ and the size of constraint Hamiltonian set $S$---that is $\mathcal{O}(n+|S|)$. 
The space complexity to store the tensor network depends on the largest tensor nodes, that can come from the cost Hamiltonian $C$ or the constraint Hamiltonians $\{Z_A\}$, that is $\mathcal{O}(\max\{1, 2^{k}\})$, where $k$ is the largest number of qubits that the constraint Hamiltonian $Z_A\in S$ acts nontrivially on.

\begin{figure}[htbp]
    \centering
    \includegraphics[width=0.8\textwidth]{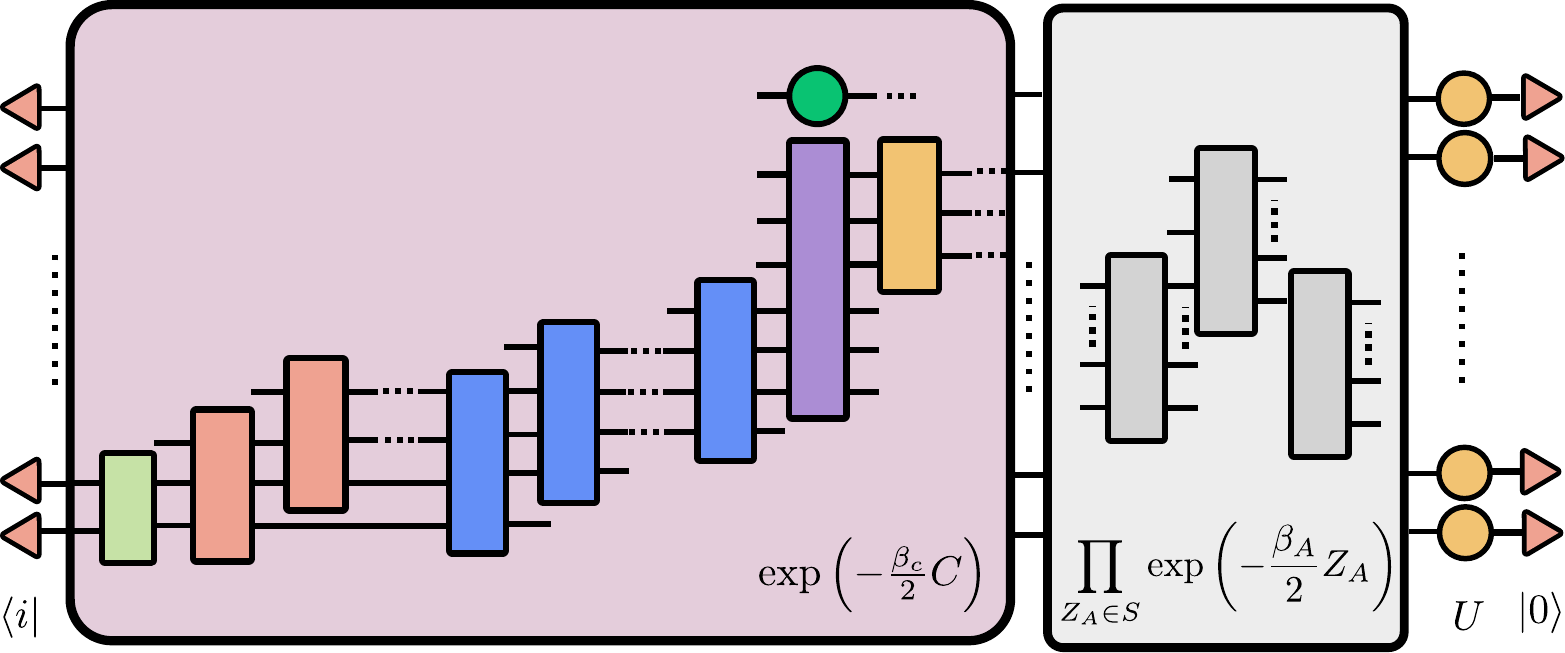}
    \caption[randomized rounding tensor network]{
        Randomized rounding applied to the tensor network representation of 
        $\bra{i}\sqrt{\sigma(\lambda)}\,U\ket{0}$[\Cref{prop:rando_rounding_mc}]. 
        The left and right columns of vector nodes correspond to the sampled bra state $\bra{i}$ and the reference ket state $\ket{0}$, respectively, with each tensor associated with a distinct qubit. 
        The middle large pink square and grey square denote the tensor network of $\exp\!\left(H(\lambda/2,C,S)\right)$. The left column of circle nodes denote the random tensor $U$ generated by the tensor product of $n$ single-qubit universal rotation $R(\phi_i,\omega_i,\theta_i)$[\Cref{eq:universal_rotation}] with random sampled parameters. Contraction of the full network is performed sequentially along the corresponding indices.
    }
    \label{fig:tensor_network_rounding}
\end{figure}

\subsection{The sparse matrix-vector approach}\label{app:details_numerics_sparse}
For the Kronecker graph, since the sampled Hamiltonian (see~\Cref{prop:kronecker_sample}) is generally non-commutative, the tensor network method is not efficient. Instead, we employ a random sampling approach to estimate the trace. The basic idea is to sample random states $\ket{\psi_i}$ from the Haar measure. Then, the trace of a selected Pauli strings $H_k$ in $\{H(\lambda, C, S)\}$ is estimated stochastically as:
\begin{equation}
\tr{H_k \exp\left(H(\lambda,C,S)\right)} \approx \sum_i\frac{1}{L} \bra{\psi_i} \exp\left(H(\lambda/2,C,S)\right) H_k \exp\left(H(\lambda/2,C,S)\right)\ket{\psi_i}
\end{equation}
where $L$ is the number of samples for the trace estimation. Through Chebyshev or Hoeffding inequalities, the stochastic error scales as $\mathcal{O}(1/\sqrt{L})$. Similarly, one obtains the normalization factor $\tr{\exp\left(H(\lambda,C,S)\right)}$. 

For randomized rounding, we again start by generating a Haar-random state $\ket{\psi_i}$. Each amplitude in the computational basis is then rounded to $\pm 1$, depending on its sign, with an overall normalization factor $1/\sqrt{2^n}$. This yields a normalized rounded state $\ket{\tilde{\psi}_i}$. The randomized rounding is evaluated by:
\begin{equation}
     \frac{1}{L'} \sum_i \bra{\tilde{\psi}_i}\exp\left(H(\lambda/2,C,S)\right) C \exp\left(H(\lambda/2,C,S)\right) \ket{\tilde{\psi}_i}
\end{equation}
where $L'$ is the number of samples for the randomized rounding procedure.

\section{On the structure of graphs corresponding to commuting Hamiltonians}\label{sec:graphs_commuting_Hamiltonians}
As one of our main numerical examples to demonstrate our methods corresponds a commuting Hamiltonian, it is worthwhile to discuss the structure of such instances. First of all, it is important to distinguish the case of what can be called fully commuting Hamiltonian and those that are commuting after a suitable regrouping.

\subsection{Fully commuting Hamiltonians}
We call Hamiltonians to be fully commuting if we have that \emph{all the Pauli strings in its decomposition commute with each other.} As we will observe later, this is a strictly stronger requirement that having a commuting Hamiltonian after a suitable regrouping and not all models that can be efficiently simulated with $1D$ tensor networks need to satisfy this. 
Such models turn out to have significantly simpler structure, as we explain now, and for local dimension $d=2$ it is possible to reduce solving $\GW$ for them to computing some eigenvalues. As we will see, the reason for that is that fully commuting Hamiltonians cannot be adjacency matrices of connected graphs unless the diagonal group is trivial.
\begin{proposition}\label{prop:commuting_disconnected}
Assume that $C=\sum_j\alpha_jP_{(x_j,y_j)}$ corresponds to a fully commuting Hamiltonian, i.e. for all $(x_i,z_i),(x_j,z_j)$ we have: 
\begin{align}
[P_{(x_i,z_i)},P_{(x_j,z_j)}]=0.
\end{align}
If $\mathcal{D}_C\not=\{\pm I\}$, then the underlying graph is disconnected.
\end{proposition}
\begin{proof}
Note that as $\mathcal{D}_C\not=\{\pm I\}$, we have at least one nontrivial $Z$ string in the diagonal group, say $Z_A$. Now let $i\in A$ and assume w.l.o.g. that $i=1$. We now claim that $X_1Z_B$ for $Z_{B}$ any (potentially empty) subset that does not contain $1$ cannot be a term in the algebra generated by the Hamiltonian. Indeed, it is not difficult to see that $[X_iZ_B,Z_A]\not=0$, contradicting the commutativity of the Hamiltonian (and therefore of the algebra the terms generate). However, the only Pauli terms that map bitstrings to each other that only differ at site $1$ are precisely of the form $X_1Z_B$. Now assume there is a path of length $k$ between two strings of the form $y_1=0x_2\ldots x_n$ and $y_2=1x_2\ldots x_n$. This implies that $\bra{y_1}A^k\ket{y_2}\not=0$, from which we infer that for at least one string of the form $X_1Z_B$ we have that $\tr{X_1Z_BA^k}\not=0$, a contradiction with such strings not being in the algebra. We conclude that no such path can exist and the vertices $y_1$ and $y_2$ are not connected.
\end{proof}
\begin{remark}
    In the case where the algebra is trivial this is not true, just consider $C=\sum_iX_i$.
\end{remark}
Thus, in the fully commuting case the graph decomposes into various components and we can solve the QUBO or GW-SDP on each component separately. But as Paulis commute with each other, we can find a basis that diagonalizes each block simultaneously. In addition, again because of commutativity, any Pauli string that contains $Z$ terms must either contain an even number of them or have a $Z$ term on which no other string acts nontrivially. As Pauli strings that have an even number of $Z$ terms admit an eigendecomposition with phase states (those with $\pm1$ amplitudes), we conclude that the QUBO on each block can be solved by diagonalizing each block separately. Thus, we conclude that solving QUBO on fully commuting Hamiltonians can be done in polynomial time by first finding the various blocks of the matrix and then solving the problem individually. Thus, it is possible to solve QUBO for such instances by considering spectral data alone.

\subsection{Hamiltonians that become commuting after regrouping}\label{sec:commuting_regrouping}
In the previous section we showed that so-called fully commuting instances (all Pauli strings in the expansion commute) are not good candidates for getting nontrivial instances for MAXCUT where our methods offer a large advantage, as they can be solved by spectral methods. However, we now argue that Hamiltonians that only becoming commuting after a suitable regrouping of the Pauli strings are not affected by such arguments. Such instances seem to have a richer structure and we were not able to identify arguments why they should be easy to solve in general.

To illustrate what we mean by a Hamiltonian that becomes commuting after regrouping, let us consider the following example on $4$ qubits:
\begin{align}\label{equ:commuting_example}
C=XXII+YYII+IXXI+IYYI+XIXI+YIYI+XXXX+IIIX.
\end{align}
It is not the case that all the individual Pauli strings commute, as e.g. $[XXII,IYYI]\not=0$. In addition, it is easy to see that $\dot\cD_C=\pm\{Z_1Z_2,Z_2Z_3,Z_1Z_3\}$ and one can numerically check that $\lambda_{\max}(C)>\GW(C)$. Furthermore, the graph is connected. All of these points indicate that the graph defined in~\Cref{equ:commuting_example} behaves qualitatively differently than the ones discussed in the previous section, that were fully commuting. But this example can still be made commuting by regrouping the terms as $A=XXII+YYII+IXXI+IYYI+XIXI+YIYI$ and writing $C=A+XXXX+IIIX$. One can then readily check that all terms in the Hamiltonian above commute with each other. We can then extend the Hamiltonian above to an arbitrarily large number of qubits without losing the $1D$, commuting or small diagonal algebra structure by just appending a commuting $1D$ Hamiltonian with trivial subalgebra starting at qubit $4$, as done in Eq.~\eqref{eq:1d_hamiltonian} of the main text.

\paragraph{Complexity of commuting instances that are not fully commuting:} as discussed before, for instances like~\Cref{equ:commuting_example}, spectral data on its own is not sufficient to determine the value of QUBO and the SDP provides a strictly better approximation. That being said, it remains unclear to what extent these instances are truly hard from a complexity point of view, as we are currently unable to show a reduction to instances that are better understood.

Nevertheless, we performed various sanity checks and numerically computed a variety of parameters of the underlying graphs that are related to more efficient algorithms to solve the QUBO. Simply put, if one of these parameters of the graph, say tree-width, scaled linearly with the the number of qubits (logarithmically in dimension), we would have algorithms to solve $\QUBO(C)$ in time polynomial in $2^n$, as there are algorithms whose scaling is only exponential in these parameters times the dimension~\cite{Cygan2015}.

These all seem to scale superlinearly fast (see Plot~\ref{fig:complexity_parameter}), which would rule out the possibility of computing the exact value of QUBO for such instances in time that is $\textrm{poly}(2^n)$. That being said, these are only numerical tests, so we cannot rule out that these numbers plateau at higher system sizes, or that there are other algorithms that are tailored to solve QUBO on such instances.

Nevertheless, our methods can sometimes deliver approximations to $\QUBO$ in time $\polylog(n)$, which would be exponentially faster than the parametrized complexity algorithms outlined above that require access to the whole graph. Unfortunately, we are unaware of other algorithms that could profit from our concise input model to run further benchmarks. 

\section{Improved stability bounds}\label{app:stability}
We now show we can obtain better stability bounds for the solution of the SDP $\GW(C,S,0)$ for certain sets $S$ by solving a linear program. We start with a general result on the continuity of the SDP: 
\begin{lemma}\label{lem:obj.at.X.ineq.aff}
Let $C\in \R^{2^n\times 2^n}$, and let $\mathcal A:\R^{2^n\times 2^n}\to \R^{|S|+1}$ be a linear operator defined by 
\begin{equation}
\mathcal AY=(\tr{ Y},\tr{ Z_{A_1} Y},\dots,\tr{ Z_{A_{|S|}}Y})\,,
\end{equation}
with adjoint $\R^{|S|}\to\R^{2^n\times 2^n}$
\begin{equation}
    \mathcal A^*(\xi)=\sum_{A_i\in S}Z_{A_i}\xi_i.
\end{equation}
Consider the dual problem to $GW(C,S,0)$:
\begin{equation}\label{eq:dual.no.tr.cons.ineq.aff}
\begin{array}{rl}
\inf\limits_{\xi\in \R^{|S|+1}}& \xi_0\\
\text{s.t.  } &\mathcal A^\top \xi \geq C\,,
\end{array}
\end{equation}
Assume both problems admit solutions $Y^\star$  and $\xi^\star$, respectively, with value $\lambda^*$.
Then, for all $X\geq 0$, the following holds:
\begin{equation}\label{equ:continuity_estimate}
\tr{ C X} \ge\lambda^\star - \langle\xi^*|\mathcal{A}X-e_0\rangle.
\end{equation}
\end{lemma}
\begin{proof}
Note that strong duality holds for $\GW$, as the identity is always a feasible point.
Set $Z^\star=C-\mathcal A ^\top \xi^\star$.
Then $Z^\star\geq 0$.
The Lagrangian has the form
\begin{equation}
L(Y,Z,\xi)=\tr{ CY} -\tr{ Z Y}-\langle\xi|\mathcal{A}X-e_0\rangle,.
\end{equation}
By the Karush–Kuhn–Tucker (KKT) conditions, we have
\begin{equation}\label{eq:KKT1.ineq.aff}
0=\frac{\partial L}{\partial Y}(Y^\star,Z^\star,\xi^\star)=C-Z^\star-\mathcal A^* \xi^\star\,.
\end{equation}
\begin{equation}\label{eq:KKT2.ineq.aff}
\tr{ Z^\star Y^\star}=0\,,
\end{equation}
\begin{equation}\label{eq:KKT3.ineq.aff}
\langle\xi^*|\mathcal{A}Y^*-e_0\rangle=0.
\end{equation}
Let $X\geq0$.
By~\Cref{eq:KKT1.ineq.aff}, $C=Z^\star+\mathcal{A}^\top \xi^\star$.
Then
\begin{equation}\label{eq:bound4.ineq.aff}
\begin{array}{rl}
\tr{ C X}-\lambda^\star=&\tr{ C( X-Y^\star)}\\
=&\tr{( Z^\star+\mathcal{A}^\top \xi^\star)( X-Y^\star)}  \\
=& \tr{ Z^\star X} + \tr{ \mathcal{A}^* \xi^\star (X-Y^\star)}\\
=& \tr{ Z^\star X} +   \langle\xi^*|\mathcal{A}(X-Y^\star )\rangle\\
=& \tr{ Z^\star X} +   \langle\xi^*|\mathcal{A}X-e_0\rangle\\
\ge &   \langle\xi^*|\mathcal{A}X-e_0\rangle\,. 
\end{array}
\end{equation}
The first equality follows from $\tr{C Y^\star} =\lambda^\star$, the third equality is due to \Cref{eq:KKT2.ineq.aff}, the last equality uses \Cref{eq:KKT3.ineq.aff}, and the last inequality is based on the positive semidefiniteness of $Z^\star$ and $X$.
Hence, the result follows.
\end{proof}

From this we can immediately obtain better continuity estimates for certain constraint sets:
\begin{cor}\label{equ:cor_stability}
For a $S\subseteq 2^{[n]}\backslash\{\varnothing\}$, $\epsilon>0$ and $C\in \mathbb{R}^{D\times D}$, let $\Xi(S,\epsilon)$ be the solution to the linear program:
\begin{align}
\sup_{\xi\in\R^{m}} &\|\xi\|_{\ell_1}\\
&\sum_{i} \xi_iZ_{A_i}\geq -\GW(C,S,\epsilon)I
\end{align}
Then for all $\epsilon>0$:
\begin{align}
|\GW(C,S,\epsilon)-\GW(C,S,0)|=\epsilon \Xi(S,\epsilon)
\end{align}
\end{cor}
\begin{proof}
Note that as $\GW(C,S,\epsilon)\geq\GW(C,S,0)$, the $0$-th coordinate of the optimal solution $\xi^*$ to~\Cref{eq:dual.no.tr.cons.ineq.aff} is smaller than $\GW(C,S,\epsilon)$. This immediately gives us the inequality:
\begin{align}\label{equ:matrix_ineq}
\sum_{i} \xi_i^*Z_{A_i} \geq C-\GW(C,S,\epsilon)
\end{align}
Now recall that we assume that the diagonal of the matrix $C$ is $0$. Evaluating inequality \Cref{equ:matrix_ineq} on the computational basis state $x$ allows us to conclude that:
\begin{align}\label{equ:matrix_ineq_2}
\langle x|\sum_{i} \xi_i^*Z_{A_i}|x\rangle \geq-\GW(C,S,\epsilon),
\end{align}
which implies that $\sum_i |\xi_i^*|\leq \Xi(S)$.
The claim then follows by applying H\"older's inequality to~\Cref{equ:continuity_estimate}, as by construction $\|\cA(X_\epsilon^*)\|_{\ell_\infty}\leq \epsilon$ and $\||\xi^*\|_{\ell_1}\leq \Xi(S)$.
\end{proof}
Note that this bound is often much better than the general bound of~\cite[Theorem 5]{henze2025solvingquadraticbinaryoptimization}, as it provides a linear scaling with $\epsilon$, whereas that bound only gave a $\epsilon^{1/3}$ scaling. In addition, for many choices of $S$, the bound can become independent of the size of $S$, only depending on the value of the program. For instance, consider the case $S=\{\{1\},\ldots,\{n\}\}$. Then it is not too difficult to show that $\Xi(S,\epsilon)=\GW(C,S,\epsilon)$, as we can pick the computational basis states to generate any sign pattern on the $\xi$. In other words, for such sets of constraints, we have that we achieve a relative error by solving the problem up to $\epsilon$.

An important class of examples for our stability bounds are Hamiltonians like the one in~\Cref{equ:commuting_example}, which has the commuting group given by $\dot\cD_C=\pm\{Z_1Z_2,Z_2Z_3,Z_1Z_3\}$. A direct inspection shows that for that choice $\Xi(\dot\cD_C,\epsilon)=3\GW(C,\dot\cD_C,\epsilon)$. 
A direct consequence is that we have a relative error $\epsilon$ for solving $\GW(C,\dot\cD_C,0)$ for this particular case by solving  $\GW(C,\dot\cD_C,\epsilon)$, a result we used in~\Cref{sec:numerics}.

\end{document}